\documentclass[a4paper,twocolumn,11pt,accepted=2024-08-27]{quantumarticle}
\pdfoutput=1

\usepackage[utf8]{inputenc}

\usepackage[T1]{fontenc}
\usepackage{hyperref}

\usepackage{braket}
\usepackage{here}
\usepackage{multirow}
\usepackage{longtable}
\usepackage{ascmac}

\usepackage{mathtools} 
\usepackage{graphicx}
\usepackage{amsfonts,amsmath,amssymb,amsthm}

\theoremstyle{plain}
\newtheorem{thm}{Theorem}
\newtheorem*{thm*}{Theorem}

\usepackage{bm}
\usepackage{algorithm}
\usepackage[noend]{algpseudocode}
\usepackage[flushleft]{threeparttable}

\algblock{Input}{EndInput}
\algnotext{EndInput}
\algblock{Output}{EndOutput}
\algnotext{EndOutput}

\newtheorem{lemma}{Lemma}

\begin{document}
\title{Quantum-enhanced mean value estimation via adaptive measurement}
\author{Kaito Wada}
\email{wkai1013keio840@keio.jp}
\affiliation{Department of Applied Physics and Physico-Informatics, Keio University, 3-14-1 Hiyoshi, Kohoku-ku, Yokohama, Kanagawa, 223-8522, Japan}
\author{Kazuma Fukuchi}
\affiliation{Department of Applied Physics and Physico-Informatics, Keio University, 3-14-1 Hiyoshi, Kohoku-ku, Yokohama, Kanagawa, 223-8522, Japan}
\author{Naoki Yamamoto}
\email{yamamoto@appi.keio.ac.jp}
\affiliation{Department of Applied Physics and Physico-Informatics, Keio University, 3-14-1 Hiyoshi, Kohoku-ku, Yokohama, Kanagawa, 223-8522, Japan}
\affiliation{Quantum Computing Center, Keio University, 3-14-1 Hiyoshi, Kohoku-ku, Yokohama, Kanagawa, 223-8522, Japan}

\begin{abstract}
Quantum-enhanced (i.e., higher performance by quantum effects than any classical methods) mean value estimation of observables is a fundamental task in various quantum technologies; in particular, it is an essential subroutine in quantum computing algorithms.
Notably, the quantum estimation theory identifies the ultimate precision of such an estimator, which is referred to as the quantum Cram\'{e}r-Rao (QCR) lower bound or equivalently the inverse of the quantum Fisher information. 
Because the estimation precision directly determines the performance of those quantum technological systems, it is highly demanded to develop a generic and practically implementable estimation method that achieves the QCR bound.
Under imperfect conditions, however, such an ultimate and implementable estimator for quantum mean values has not been developed.
In this paper, we propose a quantum-enhanced mean value estimation method in a depolarizing noisy environment that asymptotically achieves the QCR bound in the limit of a large number of qubits.
To approach the QCR bound in a practical setting, the method adaptively optimizes the amplitude amplification and a specific measurement that can be implemented without any knowledge of state preparation.
We provide a rigorous analysis for the statistical properties of the proposed adaptive estimator such as consistency and asymptotic normality.
Furthermore, several numerical simulations are provided to demonstrate the 
effectiveness of the method, particularly showing that the estimator needs only a modest number of measurements to almost saturate the QCR bound.
\end{abstract}

\maketitle

\section{Introduction}\label{sec:intro}
Mean value (or expectation value) estimation of quantum observables is an important task in many quantum 
information technologies, such as subroutines in quantum algorithms for both near-term 
noisy and future fault-tolerant quantum devices. 
For instance, variational quantum algorithms require estimating mean values iteratively 
to train the parameters of a parameterized quantum circuit 
\cite{cerezo2021variational,peruzzo2014variational,Kandala2017Nat,fuller2021approximate,gao2021applications,amaro2022case, zoufal2022variational,li2017efficient, yao2021adaptive,Benedetti_2021, wada2022,larose2019variational, cerezo2022variational}.
Surely, the performance of those information processing technologies heavily depends on 
the estimation precision. 
In addition, real quantum devices suffer from noise induced by the interaction of system 
and environment, and such noise deteriorates the efficiency for reading out calculation 
results. 
Therefore, developing efficient mean-value estimation methods (in a noisy environment) 
is highly important~\cite{brassard2002quantum,knill2007optimal,wang2019accelerated,huang2020predicting,wang2021minimizing,johnson2022reducing,huggins2022nearly}.

For this purpose, the statistical estimation theory 
\cite{cramer1946mathematical,shao2003mathematical} and their quantum extension 
\cite{holevo2011probabilistic,hayashi2006quantum} are useful.
In particular, the seminal Cram\'{e}r-Rao inequality and Fisher information, which 
identify the limit of estimation precision, have been extended to quantum versions 
\cite{helstrom1969quantum,helstrom1968minimum,hayashi2006quantum,braunstein1994statistical,paris2009quantum}; that is, the quantum Cram\'{e}r-Rao (QCR) inequality and quantum 
Fisher information. 
As in the classical case, they can be used to determine the ultimate 
performance of several quantum information processing systems such as quantum sensors  \cite{giovannetti2011advances,toth2014quantum,degen2017quantum,pirandola2018advances,PhysRevA.87.043833,hong2021quantum}. 
Moreover, the QCR acts as a fundamental limitation on information extractable from 
quantum mechanical systems, e.g., \cite{faist2021time}.
Therefore, extensive investigations have been conducted to characterize estimators 
that (almost) achieve the QCR lower bound i.e., the ultimate precision bound for optimized quantum measurement~\cite{PhysRevLett.75.2944,berry2000optimal,barndorff2000fisher,nagaoka2005asymptotic,Fujiwara_2006,Fujiwara_2011,PhysRevLett.109.130404,PhysRevLett.129.240503,PhysRevResearch.5.013138}.
Note that whether such optimal estimators are implementable depends on available quantum operations in the target system.
Devising an implementable and optimal (in the sense of QCR) estimator is thus one of the most important subjects in quantum sensing and metrology.

In this way, the quantum statistical estimation theory should be appropriately applied to develop efficient estimators for various quantities in quantum computing, such as mean values of quantum observables.
At the same time, toward the early success of quantum computing, there is an increasing need for efficient estimators that can be realized with as few demanding quantum operations as possible.
For instance, some less-demanding implementation methods for phase estimation algorithm~\cite{higgins2007entanglement,higgins2009demonstrating,svore2013faster,PhysRevLett.117.010503,o2019quantum,cimini2023experimental} and amplitude estimation algorithm~\cite{PhysRevA.93.062306,suzuki2020amplitude,aaronson2020quantum,grinko2021iterative,nakaji2020faster,brown2020quantum,tanaka2021amplitude,uno2021modified,giurgica2022low,tanaka2022noisy,callison2022improved} 
have been developed, and their performance has been investigated using the statistical estimation theories.
Nevertheless, in practical noisy settings before the fault-tolerant era, devising an estimator that is both easy to implement and capable of achieving the ultimate QCR bound remains an ongoing subject of research.
Focusing on the mean-value estimation problem in a noisy environment, Ref.~\cite{wang2021minimizing} formulated this problem 
within the framework of amplitude estimation and proposed a Bayesian inference method 
for this problem; however, there was no discussion on how to achieve the QCR bound.
To our best knowledge, only Ref.~\cite{uno2021modified} gave a detailed analysis on 
the QCR bound in the noisy amplitude estimation problem;
in particular, under the assumption that the system is subjected to depolarization 
noise, the estimator given in~\cite{uno2021modified} achieves the precision that is yet strictly lower than the ultimate limit.

In this paper, we focus on the problem of constructing a less-demanding quantum mean-value estimator in a depolarizing noisy environment.
In light of quantum estimation theory and quantum metrology, the following questions are important.
{\it (i) What is the quantum-enhanced precision limit (or, the optimal QCR bound) over less-demanding protocols for estimating a quantum mean value?
(ii) Can we construct an implementable estimator such that the estimation precision achieves the QCR bound?
(iii) Also, is there any mathematical guarantee of the estimation performance?}
The main contribution of this paper is to give answers to all the above questions.

A summary of those contributions are as follows. 
The details are provided in Section~\ref{sec:QME}. 
First, the answer to the question (i) is given in Theorem~\ref{thm:QFI_limit}; we derive the precision limit for estimating a quantum mean value, in the setup where we are allowed to use noisy quantum states but no demanding controlled amplifications.
Based on the findings obtained from the precision limit, we construct a modified maximum likelihood estimator that adaptively optimizes a specific measurement (positive operator-valued measure (POVM)) and the number of amplitude amplification operations, by maximizing the magnitude of 
quantum enhancement.
The optimality of this estimator is supported by Theorem~\ref{thm:opt_marginalPOVM}, stating that there exists a POVM that makes the ratio of the quantum and classical Fisher information close enough to 1 (leading that the QCR bound is almost achieved)  in the limit of large number of qubits. 
Note that, as mentioned in the second paragraph, the implementation of such an optimal estimator is in general nontrivial, but our estimator can be efficiently implemented without any knowledge of the target state, which is usually prepared by a black-box unitary in the mean value estimation problem.
Thus, this is the answer to the question (ii).
As for the question (iii), we prove that the estimator satisfies the following 
statistical properties.
\begin{itemize}
    \item Consistency (Theorem~\ref{thm:consistency}): 
    As the number of measurements increases, the estimation accuracy becomes better. 
    More precisely, the probability that the estimated value is arbitrarily close to the 
    true value approaches 1 in the infinite limit of the number of measurements.
    
    \item Asymptotic normality (Theorem~\ref{thm:asym_normal}): 
    The estimator regularized by its Fisher information asymptotically follows the 
    standard normal distribution, as the number of measurements increases.
    As a result, the estimation error asymptotically achieves the classical Cram\'{e}r-Rao 
    lower bound.
\end{itemize}
Therefore, the proposed estimator has mathematical guarantees in its performance.

In addition to the above theoretical contributions, Section~\ref{sec:results} is 
devoted to thorough numerical simulations. 
In particular, we will demonstrate that the above asymptotic properties hold even 
with a modest number of each measurement ($\sim 500$) and a modest number of qubits 
($\sim 20$).
Also, we focus on the previously-reported problem that the precision of the quantum 
mean value estimation under the depolarizing noise may significantly deteriorate 
depending on the target value to be estimated; but we will numerically demonstrate that 
the proposed method can also get around this problem due to the optimization. 
Note here that the above-mentioned Bayesian inference method~\cite{wang2021minimizing} 
optimizes a quantum circuit containing multi-parameter in order to alleviate the target 
value dependency of estimation efficiency, while our method optimizes just one discrete 
parameter, which is clearly computationally advantageous.


The mean value estimation is essential in many quantum algorithms, such as quantum simulation, but is too time-consuming.
Thus, it becomes crucial for practical and large-scale quantum computing to leverage quantum enhancement for efficiently estimating mean values.
Since our adaptive mean-value estimation method utilizes the maximal quantum enhancement, we believe that our method takes a pivotal role specifically in the early stage of fault-tolerant quantum computing.

\section{Preliminary}\label{method_sec:aewpe}

The quantum-enhanced amplitude estimation algorithm studied in~\cite{suzuki2020amplitude,tanaka2021amplitude,uno2021modified,tanaka2022noisy}, which does not use the hard-to-implement phase estimation algorithm~\cite{brassard2002quantum}, consists of two major steps. 
The first step is to perform amplitude amplification on the quantum state whose amplitude is to be estimated, and then make a measurement; this procedure is repeated with different operation times. 
The second step is to estimate the amplitude by classical post-processing for the combined measurement result, i.e., the maximum likelihood estimation.

The problem is to estimate the amplitude $\sin{(\phi^*)}$, with unknown $\phi^* \in (0, \pi/2)$, of the following $(n + 1)$-qubit quantum state $\mathcal{A}\ket{0}_{n+1}$: 
\begin{align}\label{eq:target_state}
    &\mathcal{A}\ket{0}_{n+1} \notag\\
    &:= \cos{\left(\phi^*\right)}\ket{\psi_0}_n\ket{0}_1 + \sin{\left(\phi^*\right)}\ket{\psi_1}_n\ket{1}_1,
\end{align}
where $\ket{0}_{n+1}$ is the computational basis of $(n+1)$-qubit, $\ket{\psi_0}_n$ and $\ket{\psi_1}_n$ are 
normalized $n$-qubit quantum states.
Here, $\mathcal{A}$ is a black-box (i.e., unknown yet implementable) operation on a device, which is often called an oracle in the Grover search algorithm.
Note that a single application of $\mathcal{A}$ or $\mathcal{A}^\dagger$ is counted as one query.
The first step is to amplify the amplitude via the operator $\mathcal{Q}$ defined as
\begin{align}\label{eq:conv_amplification_operator}
    \mathcal{Q} &:= \mathcal{A}\left(2\ket{0}_{n+1}\bra{0}_{n+1}-I_{n+1}\right)\mathcal{A}^\dagger\left(I_n\otimes Z\right),
\end{align}
where $I_n$ is the identity operator on the $n$-qubit system, and $Z$ denotes the $2\times 2$ Pauli Z matrix.
From the assumption of implementability of $\mathcal{A}$ and $\mathcal{A}^\dagger$ as oracle operators, the amplification operator $\mathcal{Q}$ is also implementable.
As shown originally in~\cite{brassard2002quantum}, the amplification operator corresponds to the following rotation gate in the subspace $\tilde{\mathcal{S}}$ spanned by $\ket{\psi_0}_n\ket{0}_1$ and $\ket{\psi_1}_n\ket{1}_1$:
\begin{equation}\label{eq:grover_rotation}
    \left.\mathcal{Q}\right|_{\tilde{\mathcal{S}}}=\begin{bmatrix}
    \cos{\left(2\phi^*\right)}&-\sin{\left(2\phi^*\right)}\\
    \sin{\left(2\phi^*\right)}&\cos{\left(2\phi^*\right)}
    \end{bmatrix}=e^{-i2\phi^* \tilde{Y}},
\end{equation}
where $\tilde{Y}$ denotes the Pauli Y matrix with respect to the basis of $\tilde{\mathcal{S}}$.
Thus, $m$ applications of $\mathcal{Q}$ to the state (\ref{eq:target_state}) yield
\begin{align}\label{eq:amplified_state_ampest}
    \mathcal{Q}^{m}\mathcal{A}\ket{0}_{n+1}&=\cos{\left[(2m +1)\phi^*\right]}\ket{\psi_0}_n\ket{0}_1\notag\\[4pt]
    &~~~+\sin{\left[(2m +1)\phi^*\right]}\ket{\psi_1}_n\ket{1}_1.
\end{align}
We now measure the last qubit of Eq.~(\ref{eq:amplified_state_ampest}) by the computational basis $\ket{0}_1$ and $\ket{1}_1$. 
Then the probability of obtaining "1" is given by 
\begin{equation}\label{eq:bernoulli_prob_no_noise}
    \mathbf{\tilde{P}}\left(m;\phi^*\right):=\frac{1}{2}-\frac{1}{2}\cos{\left[2(2m+1)\phi^*\right]}.
\end{equation}
%
Note that this is equivalent to the Bernoulli trial with success probability $\mathbf{\tilde{P}}\left(m;\phi^*\right)$. 
By the applications of $\mathcal{Q}$, the resolution on $\phi^*$ in the probability (\ref{eq:bernoulli_prob_no_noise}) is enhanced by the factor $2m+1$; on the other hand, the ambiguity emerges in the sense that the probability $\mathbf{\tilde{P}}\left(m;\phi\right)$ becomes periodic with respect to $\phi$.

In the second step, we estimate the amplitude based on the measurement results obtained in the first step by the maximum likelihood estimation method.
Note that in the first step, we prepare the quantum states with different resolutions in order to eliminate 
the ambiguities, as similar to the methods for phase measurement~\cite{mitchell2005metrology,higgins2009demonstrating,berry2009perform,PhysRevA.102.042613,Dutkiewicz2022heisenberglimited,cimini2023experimental}.
Specifically, for each of predefined odd numbers $2m_k+1~(k = 1, 2, \cdots, M)$, we prepare the state (\ref{eq:amplified_state_ampest}) and measure it $N$ times with the computational basis independently.
Here, $M$ denotes the total number of types of quantum states to be prepared.
We write $x^{(k)} \in \{0, 1, \cdots, N\}$ as the number of hitting "1" for the state $\mathcal{Q}^{m_k}\mathcal{A}\ket{0}_{n+1}$. 
The measurement results for $M$ different states are put together as $\bm{x}_M := (x^{(1)}, x^{(2)}, \cdots , x^{(M)})$. 
Then the likelihood function to have $\bm{x}_M$ is given by
\begin{equation}\label{eq:likelihood_ampest}
    \tilde{\mathcal{L}}_M(\phi;\bm{x}_M) := \prod_{k = 1}^M {\tilde{F}}_k(x^{(k)};m_k,\phi),
\end{equation}
where $\tilde{F}_k(x^{(k)};m_k,\phi^*)$ is the probability of obtaining $x^{(k)}$. 
Given the measurements $\bm{x}_M$, the maximum likelihood estimation assumes that the value of $\phi$ that maximizes $\tilde{\mathcal{L}}_M(\phi;\bm{x}_M)$ is a plausible estimate of the true value $\phi^*$.
In other words, the maximum likelihood estimate $\hat{\phi}_M$ for $\phi^*$ from the $M$ series of measurements is defined as
\begin{equation}
\label{prelimi_optim phase}
    \hat{\phi}_M := \underset{\phi\in\left[0,{\pi}/2\right]}{\rm argmax}~\mathcal{\tilde{L}}_M(\phi;{\bm{x}}_M).
\end{equation}
Note that the maximum likelihood estimator for $\sin{(\phi^*)}$ corresponds to $\sin{(\hat{\phi}_M)}$ due to the invariance property of maximum likelihood estimators. 
Thanks to the quantum-enhanced resolution via amplitude amplification and the elimination of its ambiguities by the product of likelihoods, the estimation precision of $\phi^*$ and 
$\sin{(\phi^*)}$ can be improved quadratically with respect to the number of total queries of $\mathcal{A}$ and $\mathcal{A}^\dagger$, by using a specific sequence $m_k=2^{k-1}$ as demonstrated in~\cite{suzuki2020amplitude}.

We now consider the $n$-qubit depolarizing channel: 
\begin{equation}\label{eq:dim_d_depolarizing}
    \mathcal{D}[\rho]:=p\rho +\frac{1-p}{d}I_n,~~~d:=2^n, 
\end{equation}
where $\rho$ is an arbitrary density operator and $1-p$ is the probability that depolarization occurs.
As in~\cite{wang2021minimizing,uno2021modified,tanaka2021amplitude}, we here consider a noise model in which depolarization occurs with probability $1-p_{\rm q}$ for a single application $\mathcal{A}$ or $\mathcal{A}^\dagger$.
This assumption may be justified when the actual implementation of oracles $\mathcal{A}$ and $\mathcal{A}^\dagger$ requires complicated circuits compared to other operations in Eq.~(\ref{eq:conv_amplification_operator}).
Then the probability of obtaining "1" for the final state under this noise model 
is given by
\begin{equation}\label{eq:bernoulli_prob_noise}
    \mathbf{\tilde{P}}_{\mathcal{D}}\left(m;\phi^*\right):=\frac{1}{2}-\frac{p_{\rm q}^{2m+1}}{2}\cos{\left[2(2m+1)\phi^*\right]}.
\end{equation}
The subscript $\mathcal{D}$ means that the quantum state to be measured has passed through the above-defined depolarizing channels. 
Then the classical Fisher information associated with the probability (\ref{eq:bernoulli_prob_noise}) for $\phi^*$ is calculated as
\begin{align}\label{eq:cFI_odd}
    &\mathcal{\tilde{I}}_{\rm c}(m;\phi^*)\notag\\
    &:=\frac{(2m+1)^2p_{\rm q}^{2(2m+1)}\sin^2{\left[2(2m+1)\phi^*\right]}}{\mathbf{\tilde{P}}_{\mathcal{D}}(m;\phi^*)\left(1-\mathbf{\tilde{P}}_{\mathcal{D}}(m;\phi^*)\right)}.
\end{align}
Equation~(\ref{eq:cFI_odd}) indicates that the estimation may become ineffective 
if $m$ satisfies $\sin{\left[2(2m+1)\phi^*\right]}\simeq 0$.
More precisely, if $m$ and $\phi^*$ satisfy this condition, the classical Cram\'{e}r-Rao lower bound of the estimation error for $\phi^*$ significantly deteriorates~\cite{wang2021minimizing,tanaka2021amplitude}.
In addition, even if we use the above-described maximum likelihood estimation based on the $M$ series of measurements for $m_k~(k=1,2,\cdots,M)$, which are predefined independently to $\phi^*$, the same deterioration may occur~\cite{tanaka2021amplitude}. 
Importantly, this phenomena are also seen in the quantum mean value estimation problem discussed in the next section.

\section{Quantum-enhanced mean value estimation}\label{sec:QME}

First, in Section~\ref{sec:mean_val_est_problem}, we formulate the quantum mean value estimation problem in a similar way to the amplitude estimation problem shown in the previous section.
Then, we provide the precision limit in estimating the mean value when we can prepare and measure noisy amplitude-amplified quantum states, in Section~\ref{sec:precision_limit}.
To achieve the precision limit, we next provide the efficient and implementable measurements in Section~\ref{sec:two_type_meas} and thereby develop an adaptive mean-value estimation algorithm in Section~\ref{sec:proposed_alg}.
In Section~\ref{sec:stat_prop_est_main}, we discuss the statistical properties of our method such as consistency, asymptotic normality, and the (classical) Fisher information that can become large by the adaptive optimization.
Finally, a summary of our protocol is provided in Section~\ref{sec:summary}.

\subsection{Mean value estimation problem}\label{sec:mean_val_est_problem}
Our goal is to estimate the mean value of a known Hermitian operator $\mathcal{O}$ with eigenvalues $\pm1$, i.e., $\langle\mathcal{O}\rangle:=\bra{A}\mathcal{O}\ket{A}$, where $\ket{A}:=A\ket{0}_n$ is an $n$-qubit quantum state with a state preparation unitary $A$.
Here, we assume that $A$ and $A^\dagger$ are black-box operations yet implementable on a device, similar to $\mathcal{A}$ in the previous section. 
Note that $\mathcal{O}$ is a unitary operator, which is assumed to be implementable; also, the projective measurement $\{(I_n+\mathcal{O})/2,(I_n-\mathcal{O})/2\}$ is implementable, which is often assumed in quantum algorithms.
A typical $\mathcal{O}$ is a single Pauli string such as $\mathcal{O}=X\otimes Z \otimes \cdots$; 
note that our method is applicable to estimate the mean value of a linear combination of Pauli 
strings that form a completely anti-commuting set, because there exists a unitary operator 
that groups the linear combination into a single Pauli string~\cite{zhao2020measurement}.

In the estimation problem, we aim for maximizing the estimation precision of the target mean value $\langle\mathcal{O}\rangle$ in given total queries to $A$ and $A^\dagger$~\cite{knill2007optimal,PhysRevA.102.022408,wang2021minimizing,huggins2022nearly}.
In particular, we address this problem without the quantum phase estimation algorithm which requires many controlled operations~\cite{knill2007optimal}.
The amplitude estimation method (with no controlled-$\mathcal{Q}$ operations) in the previous section can be applied to this problem; hence, ideally, we need quadratically less queries of $A$ and $A^\dagger$ compared to the standard algorithm without amplitude amplification, to estimate the mean value with specified precision --- that is, the quantum-enhanced mean value estimation.







Here, we describe the mean value estimation problem in a similar way to the previous section.
Suppose $\ket{A}$ is not an eigenstate of $\mathcal{O}$ (thus $|\langle\mathcal{O}\rangle|\neq1$).
Then the two quantum states $\ket{A}$ and $\mathcal{O}\ket{A}$ are linearly independent and 
form the following subspace~\cite{wang2021minimizing,koh2020framework}:
\begin{align}
\label{eq:def_subspace}
    \mathcal{S}&:={\rm Span}\left\{|A\rangle,\mathcal{O}|A\rangle\right\}={\rm Span}\left\{|A\rangle,|A^{\perp}\rangle\right\},
\end{align}
where $|A^{\perp}\rangle$ is an (unknown) normalized state orthogonal to $\ket{A}$ obtained by Gram-Schmidt procedure to $\ket{A}$ and $\mathcal{O}\ket{A}$.
We denote $\ket{A}$ and $|A^{\perp}\rangle$ as $\ket{\bar0}$ and $\ket{\bar1}$, respectively, 
and identify $\mathcal{S}$ as the 1-qubit Hilbert space likewise the idea of qubitization~\cite{Low2019hamiltonian,martyn2021grand}. 
We then define an $n$-qubit operator $Q$, which corresponds to the amplitude amplification operator in Eq.~(\ref{eq:conv_amplification_operator}), as
\begin{equation}\label{eq:amplification_operator_expectation}
    Q:=A\left(2\ket{0}_n\bra{0}_n-I_{n}\right)A^\dagger\mathcal{O}.
\end{equation}
Recall that the oracles ${A}$ and ${A}^\dagger$ are implementable, meaning that the amplification operator ${Q}$ is also implementable.
This operator keeps the subspace $\mathcal{S}$ invariant because $\mathcal{O}\ket{A}=\braket{\mathcal{O}}\ket{A}+\sqrt{1-\braket{\mathcal{O}}^2}\ket{A^\perp}$ and $\mathcal{O}\ket{A^\perp}=\sqrt{1-\braket{\mathcal{O}}^2}\ket{A}-\braket{\mathcal{O}}\ket{A^\perp}$ hold from the definition of $\ket{A^\perp}$ and the fact that $\mathcal{O}^2=I_n$.
Then, similar to Eq.~(\ref{eq:grover_rotation}), 
the representation of $Q$ on $\mathcal{S}$ is expressed as a unitary operator: $\left.Q\right|_{\mathcal{S}}=\bar{Z}(\cos{(\theta^*)}\bar{Z}+\sin{(\theta^*)}\bar{X})={\rm exp}(i\theta^*\bar{Y})$, 
where we define $\bar{X},\bar{Y},\bar{Z}$ as the Pauli X, Y, Z matrices in the basis $\ket{\bar0}$ and $\ket{\bar1}$.
The rotational angle $\theta^*$ is related to the target mean value as $\theta^*:=\arccos\left(\langle\mathcal{O}\rangle\right)$; in what follows $\theta^*$ is also referred to as the target value.
Note that, for the target mean value $\langle\mathcal{O}\rangle=\cos\theta^*$, the domain of $\theta^*$ is given by $\theta^*\in(0,\pi)$, which is twice as that of $\phi^*$ in the previous section.

Using several queries of the rotation $Q$ on $\mathcal{S}$, we can prepare and measure the quantum state with quantum-enhanced resolution by $\alpha\in\mathbb{N}$, at the cost of $\alpha$ queries to $A$ and $A^\dagger$.
Applying $m$ times of the rotation $Q$ to $\ket{\bar{0}}$ followed by $\mathcal{O}$ and going back to the computational basis by $A^\dagger$, we obtain the following quantum state
\begin{align}\label{eq:even_state}
    \cos{\left(\frac{(2m+2)\theta^*}{2}\right)}\ket{0}_n+\sin{\left(\frac{(2m+2)\theta^*}{2}\right)}A^\dagger\ket{\bar{1}},
\end{align}
where this state uses $2(m+1)$ queries to $A$ or $A^\dagger$ in total.
Here, the basis $A^\dagger\ket{\bar{1}}$ is actually unknown by the black-box assumption of $A^\dagger$; thus, the class of implementable measurements on the state (\ref{eq:even_state}) is restricted, as detailed in the following sections.
On the other hand, the projective measurement of $\mathcal{O}$ on $Q^m\ket{\bar{0}}$ should be considered in $\mathcal{S}$ because of the implicit $\theta^*$-dependence of $\mathcal{O}$.
That is, this process is equivalent to the projective measurement of $\bar{Z}\oplus \left.\mathcal{O}\right|_{\mathcal{S}^\perp}$ on the following state with $2m+1$ queries to $A$ or $A^\dagger$:
\begin{equation}\label{eq:odd_state}
    \cos{\left(\frac{(2m+1)\theta^*}{2}\right)}\ket{\bar{0}}+\sin{\left(\frac{(2m+1)\theta^*}{2}\right)}\ket{\bar{1}}.
\end{equation}
This means that we can measure the quantum state with $(2m+1)$-fold enhanced resolution of $\theta^*$, although available POVM is restricted.
Therefore, for a given $\alpha\in\mathbb{N}$, we can prepare and measure the quantum state Eq.~(\ref{eq:even_state}) ($\alpha=2m+2$) or Eq.~(\ref{eq:odd_state}) ($\alpha=2m+1$), according to the parity of $\alpha$; we write these states in a unified manner as $\rho(\alpha;\theta^*)$.
In the following, we call $\alpha$ the {\it amplified level}.
Note that $\rho(\alpha;\theta^*)$ requires $\alpha$ queries to $A$ or $A^\dagger$ in total.

Here, we assume the state $\rho(\alpha;\theta^*)$ is subjected to the (global) depolarization noise with probability $1-p_{\rm q}~(>0)$ for each use of $A$ ($A^\dagger$), as in the previous case.
Then, the output state with $\alpha$ queries to $A$ or $A^\dagger$ is given by
%
\begin{align}
\label{eq:amplified_state_in_S_noise}
    \rho_{\mathcal{D}}(\alpha;\theta^*)
      :=p_{\rm q}^{\alpha}\rho(\alpha;\theta^*)+\frac{1-p_{\rm q}^{\alpha}}{d}I_n.
\end{align}
We remark that $A$ and $A^\dagger$ (e.g., the time evolution of a global Hamiltonian) are usually comprised of a large number of gates, which would sufficiently scramble local noise into global depolarization noise.
Such assumptions are verified in terms of theoretical~\cite{foldager2023can,dalzell2024random} and (small-size) experimental~\cite{tanaka2021amplitude} aspects, while we leave further verification in a large-size experiment as future work.


\subsection{Limit of estimation precision with noisy amplitude-amplified states}\label{sec:precision_limit}

%
%

%
In the following, we derive the precision limit in estimating $\theta^*$ with the use of the noisy amplitude-amplified states $\rho_{\mathcal{D}}(\alpha;\theta^*)$, based on the quantum Fisher information of these states; see Appendix~\ref{apdx:A} for a summary of quantum estimation theory.
To this end, we suppose that any POVM including the joint measurement on multiple $n$-qubit systems can be performed, for the time being.

For a given number $N_{\rm q}$ of queries to $A$ or $A^\dagger$, let us prepare a set of $\rho_{\mathcal{D}}(\alpha;\theta^*)$ by splitting the total queries.
More specifically, we consider any partition $\{\alpha'_k\}$ of given total 
queries $N_{\rm q}$ such that $N_{\rm q}=\sum_{k=1}^{M'} \alpha'_k$ for some positive integer $M'$.
Then, we define the quantum state before measurement as
\begin{equation}\label{eq:product_state}
    \rho_{\mathcal{D}}^{{\rm (bm)}}(\theta^*):=\rho_{\mathcal{D}}(\alpha'_1;\theta^*)\otimes \cdots\otimes \rho_{\mathcal{D}}(\alpha'_{M'};\theta^*).
\end{equation}
Note that each $\rho_{\mathcal{D}}(\alpha'_k;\theta^*)$ uses $\alpha'_k$ queries to $A$ or $A^\dagger$, and thus the quantum state $\rho_{\mathcal{D}}^{{\rm (bm)}}$ uses $N_{\rm q}$ queries in total.
Here, for all partitions, we provide the upper bound of the total quantum Fisher information $\mathcal{I}_{\rm q,tot}[\rho_{\mathcal{D}}^{{\rm (bm)}}]$ of the state $\rho_{\mathcal{D}}^{{\rm (bm)}}$ in estimating $\theta^*$.
\begin{thm}\label{thm:QFI_limit}
    {\rm {(Precision limit in the mean-value estimation problem)}}
    For a given number $N_{\rm q}$ of queries to $n$-qubit state preparation oracles $A$ and $A^\dagger$, we consider all partitions $\{\alpha'_k\}~(\alpha'_k\in\mathbb{N})$ of $N_{\rm q}$ such that $N_{\rm q}=\sum_{k=1}^{M'}\alpha'_k$ for some positive integer $M'$.
    Here, each use of $A$ and $A^\dagger$ induces the $n$-qubit depolarization noise with probability $1-p_{\rm q}$.
    Then, the total quantum Fisher information $\mathcal{I}_{\rm q,tot}$ regarding $\theta^*$ of the quantum state $\rho_{\mathcal{D}}^{(\rm bm)}(\theta^*)$ in Eq.~(\ref{eq:product_state}) satisfies the following inequalities:
    $$
    \mathcal{I}_{\rm q,tot}\left[\rho_{\mathcal{D}}^{(\rm bm)}\right] \leq \frac{N_{\rm q}^2p_{\rm q}^{2N_{\rm q}}}{2^{1-n}+\left(1-2^{1-n}\right)p_{\rm q}^{N_{\rm q}}},~N_{\rm q}\leq \alpha_{\rm B},
    $$
    and
    $$
    \mathcal{I}_{\rm q,tot}\left[\rho_{\mathcal{D}}^{(\rm bm)}\right] \leq \frac{N_{\rm q}\alpha_{\rm B}p_{\rm q}^{2\alpha_{\rm B}}}{2^{1-n}+\left(1-2^{1-n}\right)p_{\rm q}^{\alpha_{\rm B}}},~N_{\rm q} > \alpha_{\rm B},
    $$
    where $\alpha_{\rm B}$ is defined by 
    \begin{equation}
        \alpha_{\rm B}:=\underset{\alpha\in\mathbb{N}}{\rm argmax}~\frac{\alpha p_{\rm q}^{2\alpha}}{{2^{1-n}}+\left(1-2^{1-n}\right)p_{\rm q}^{\alpha}}.    
    \end{equation}
    Furthermore, if $N_{\rm q}\leq \alpha_{\rm B}$ or $N_{\rm q}=r\alpha_{\rm B}$ for some $r\in\mathbb{N}$, there exists a partition $\{\alpha'_k\}$ satisfying the above equality.
\end{thm}
The proof is given in Appendix~\ref{apdx:proof_upb}, and a similar evaluation (i.e., the loss of quadratic term on ${N_{\rm q}}$) of quantum Fisher information can also be found in the previous works on noisy quantum metrology~\cite{demkowicz2012elusive,PhysRevLett.113.250801,PhysRevLett.131.090801}.
Note that, if the argmax returns multiple values, we take the minimum as $\alpha_{\rm B}$.
A particularly important fact is that the equality in the case $N_{\rm q}=r\alpha_{\rm B}>\alpha_{\rm B}$ holds for the partition $M'=r$ and $\alpha'_{k}=\alpha_{\rm B}~\forall k$; that is, 
\begin{equation}\label{eq:product state with alpha_B}
    \rho_{\mathcal{D}}^{{\rm (bm)}}(\theta^*)=\rho_{\mathcal{D}}(\alpha_{\rm B};\theta^*)\otimes \cdots\otimes \rho_{\mathcal{D}}(\alpha_{\rm B};\theta^*)
\end{equation}
enables us to achieve the upper bound of $\mathcal{I}_{\rm q,tot}$. 
Here, $\alpha_{\rm B}$ can be expressed as
\begin{equation}
\label{def of alpha B}
        \alpha_{\rm B}=\underset{\alpha\in\mathbb{N}}{\rm argmax}~
        \frac{\mathcal{I}_{\rm q}(\alpha)}{\alpha}, 
\end{equation}
where $\mathcal{I}_{\rm q}(\alpha)$ denotes the quantum Fisher information of $\rho_{\mathcal{D}}(\alpha;\theta^*)$ regarding $\theta^*$:
\begin{equation}\label{eq:unified_qFI}
    \mathcal{I}_{\rm q}(\alpha):=\frac{\alpha^2p_{\rm q}^{2\alpha}}{{2^{1-n}}+\left(1-{2^{1-n}}\right)p^{\alpha}_{\rm q}}.
\end{equation}
We can derive Eq.~\eqref{eq:unified_qFI} using the standard recipe found in e.g.,~\cite{jiang2014quantum,yao2014multiple,PhysRevA.97.042322}.
That is, fixing $\alpha$ to $\alpha_{\rm B}$ provides the best usage of $A$ and $A^\dagger$ in estimating $\theta^*$ under the noisy condition.
In this paper, we particularly focus on the regime $\alpha_{\rm B}\gtrsim 100$ (i.e., $\lesssim$ 1\% depolarization noise on each $A$) for the early stage of fault-tolerant quantum computing~\cite{PRXQuantum.3.010345}.
As for the case $N_{\rm q}\leq \alpha_{\rm B}$,
if we coherently use all queries (i.e., $M'=1$ and $\alpha'_1=N_{\rm q}$), the quantum Fisher information of $\rho_{\mathcal{D}}^{(\rm bm)}$ achieves the upper bound, which contains the quadratic term with respect to the total queries $N_{\rm q}$.

From the quantum Cram\'{e}r-Rao (QCR) inequality~\cite{helstrom1969quantum,helstrom1968minimum,hayashi2006quantum,braunstein1994statistical,paris2009quantum}, the upper bounds in Theorem~\ref{thm:QFI_limit} set the precision limit in estimating $\theta^*$ with given queries $N_{\rm q}$, via the quantum states $\rho_{\mathcal{D}}(\alpha'_k;\theta^*)$.
That is, for any (unbiased) estimator $\hat{\theta}$ of $\theta^*$ based on the measurement result from the quantum states, the mean squared error (MSE) of $\hat{\theta}$ is lower bounded as 
\begin{equation}\label{eq:fundamental_precision_limit1}
    {\rm MSE}(\hat{\theta})\geq \frac{1}{N^2_{\rm q}} \frac{2^{1-n}+\left(1-2^{1-n}\right)p_{\rm q}^{N_{\rm q}}}{p_{\rm q}^{2N_{\rm q}}},~N_{\rm q}\leq \alpha_{\rm B}
\end{equation}
and
\begin{equation}\label{eq:fundamental_precision_limit2}
    {\rm MSE}(\hat{\theta})\geq \frac{1}{N_{\rm q}\alpha_{\rm B}} \frac{2^{1-n}+\left(1-2^{1-n}\right)p_{\rm q}^{\alpha_{\rm B}}}{p_{\rm q}^{2\alpha_{\rm B}}},~N_{\rm q}> \alpha_{\rm B}.
\end{equation}
Note that we can obtain the lower bound of MSE regarding $\langle\mathcal{O}\rangle=\cos\theta^*$ by simply multiplying this inequality by $1-\langle\mathcal{O}\rangle^2$.
These inequalities hold for any POVM independent to $\theta^*$; thus they characterize the limit of estimation precision that cannot be improved by any measurement.
In the following, we mainly focus on the precision limit given by Eq.~(\ref{eq:fundamental_precision_limit2}) because the other bound is asymptotically negligible in a large number of queries $N_{\rm q}$.



\subsection{The two-type measurement}\label{sec:two_type_meas}
To most efficiently read out the quantum mean value from $\rho_{\mathcal{D}}$, we need a POVM whose classical Fisher information achieves the quantum Fisher information defined for this state.
In fact, we found the following 3-valued POVM that satisfies the above requirement for specific values of $\theta^*$ and $\rho_{\mathcal{D}}(\alpha;\theta^*)$ with even $\alpha=2m+2$, as shown in Appendix~\ref{apdx:analysis_POVM}:
\begin{align}\label{eq:POVM_bf_marginal}
    M_0&:=\ket{{0}}_n\bra{{0}},~M_1:=A^\dagger\ket{\bar{1}}\bra{\bar{1}}A,\notag\\[6pt]
    M_2&:=I_n-M_0-M_1,
\end{align}
where we recall that $\bra{0}_n A^\dagger|\bar{1}\rangle=\langle{\bar{0}}|\bar{1}\rangle=0$.
Note that this POVM can be obtained from an optimal POVM that exactly achieves the quantum Fisher information at the cost of $\theta^*$-dependancy in the POVM elements; also see Appendix~\ref{apdx:analysis_POVM}.

Suppose we can perform such an optimal measurement on each $n$-qubit system of $\rho_{\mathcal{D}}^{(\rm bm)}$ in a separable way.
Then, the corresponding (total) classical Fisher information $\mathcal{I}_{\rm c,tot}$, which depends on $\rho_{\mathcal{D}}^{(\rm bm)}$ and the selection of POVM, is equal to $\mathcal{I}_{\rm q,tot}[\rho_{\mathcal{D}}^{(\rm bm)}]$ and furthermore achieves the upper bound in Theorem~\ref{thm:QFI_limit} by setting $\alpha'_k=\alpha_B~\forall k$.
However, both the optimal measurement and the POVM \eqref{eq:POVM_bf_marginal} contain the unknown state $A^\dagger\ket{\bar{1}}$, and thus it is highly nontrivial to perform the measurements on standard quantum computing devices under the black-box assumption of $A$ and $A^\dagger$.
Hence, we approximate the ideal 3-valued POVM by marginalizing it, to have the following 2-valued POVM:
\begin{align}
\label{eq:even_POVM}
    M^{(\rm even)}_1&:=I_n-M^{(\rm even)}_0,\notag\\[4pt]
    M^{(\rm even)}_0&:=\ket{0}_n\bra{0},
\end{align}
where the superscript \textit{even} means that we can perform this measurement on the quantum state $\rho_{\mathcal{D}}$ with \textit{even} $\alpha$ (i.e., $\alpha=2m+2$).
This POVM can be implemented without any knowledge of the black-box operations $A$ and $A^\dagger$, by the computational basis measurement followed by post-processing for classifying the obtained measurements into "0" or "1".

The measurement~\eqref{eq:even_POVM} has powerful estimation capabilities in the sense of its Fisher information, as follows.
Measuring the state (\ref{eq:amplified_state_in_S_noise}) with $\alpha=2m+2$ by this POVM, we obtain the measurement result corresponding to $M^{(\rm even)}_1$ with probability 
\begin{align}\label{eq:P_even}
    &\mathbf{P}^{(\rm even)}_{\mathcal{D}}(m;\theta^*)\nonumber\\
    &:=\frac{d-1}{d}+p_{\rm q}^{2(m+1)}\left[\sin^2{\left[(m+1)\theta^*\right]}-\frac{d-1}{d}\right].
\end{align}
Then, the classical Fisher information with Eq.~(\ref{eq:P_even}) is calculated as
\begin{align}\label{eq:even_FI}
    &\mathcal{I}^{(\rm even)}_{\rm c}(m;\theta^*)\notag\\[4pt]
    &:=\frac{(2m+2)^2 p_{\rm q}^{4(m+1)}\sin^2{\left[(2m+2)\theta^*\right]}}{4\mathbf{P}^{(\rm even)}_{\mathcal{D}}(m;\theta^*)\left(1-\mathbf{P}^{(\rm even)}_{\mathcal{D}}(m;\theta^*)\right)}.
\end{align}
Also, we rewrite the corresponding quantum Fisher information (\ref{eq:unified_qFI}) as $\mathcal{I}^{(\rm even)}_{\rm q}(m):=\mathcal{I}_{\rm q}(2m+2)$ for simplicity.

\begin{thm}\label{thm:opt_marginalPOVM}
{\rm (Optimality of the marginalized POVM)}
\begin{align}\label{eq:even_FI_infty}
    &\mathcal{I}^{(\rm even)}_{\rm c}(m;\theta^*)\notag\\
    &~~~=\kappa\mathcal{I}^{(\rm even)}_{\rm q}(m)\times\left[1+\frac{1}{2^n}\frac{2(1-p_{\rm q}^{2m+2})}{ p_{\rm q}^{2m+2}}\right]
\end{align}
holds, where $\kappa$ is defined as 
\begin{align}\label{eq:kappa}
    \kappa
    := \frac{\sin^2{[(m+1)\theta^*]}}{1-p_{\rm q}^{2m+2}\cos^2{[(m+1)\theta^*]}+\varepsilon_n(m;\theta^*)}.
\end{align}
Here, $\varepsilon_n(m;\theta^*)=O(2^{-n})$ if $\cos{[(m+1)\theta^*]}\neq 0$.
When $\theta^*/2\pi$ is an irrational number, there exists an $m$ (i.e., the number of querying $Q$) such that $\kappa_{\infty}:=\lim_{n\to \infty}\kappa$ is arbitrarily close to $1$.
\end{thm}

This theorem means that, to estimate a {\it particular} $\theta^*$, the POVM (\ref{eq:even_POVM}) can become optimal in the sense of Fisher information, as the number of qubits increases; the proof of Theorem~\ref{thm:opt_marginalPOVM} and additional analysis are provided in Appendix~\ref{apdx:analysis_POVM}.
Note that a more detailed analysis shows that the factor $\varepsilon_n(m;\theta^*)$ diverges at certain points $(m,\theta^*)$ such that $\cos{[(m+1)\theta^*]}=0$, while in other points, it converges to zero exponentially fast with respect to the number of qubits $n$.
In the limit of a large number of qubits $n$, $\mathcal{I}_{\rm c}^{\rm (even)}$ still depends on $\theta^*$ by the factor $\kappa_{\infty}$ in the domains $\sin{[2(m+1)\theta^*]}\neq 0$, which differs from the classical Fisher information of the original 3-valued POVM; see Appendix~\ref{apdx:analysis_POVM}.

Based on Theorem~\ref{thm:opt_marginalPOVM}, in the next subsection we introduce an algorithm that optimizes $m$ (or equivalently $\alpha$) in a certain finite range.
Although this theorem does not guarantee $\mathcal{I}^{(\rm even)}_{\rm c}\approx \mathcal{I}^{(\rm even)}_{\rm q}$ for the optimized $\alpha$, we numerically verify this approximation holds well when using modest number of qubits.
Furthermore, this optimization is often valid even when $\theta^*/2\pi$ is a rational number as demonstrated later.
As a consequence, this algorithm enables us to have nearly optimal POVMs for \textit{almost all} $\theta^*$ in the large system dimension, in the sense of total Fisher information i.e., $\mathcal{I}_{\rm c,tot}\approx \mathcal{I}_{\rm q,tot}$.

Although the even POVM (\ref{eq:even_POVM}) can become optimal in the sense of Fisher information, this does not necessarily guarantee that the actual estimator (particularly 
the maximum likelihood estimator based on the POVM) produces the best estimate that achieves the QCR bound.
Actually, the probability (\ref{eq:P_even}) is an even function with respect to 
the value $\theta^*-\pi/2$, and the maximum likelihood estimation associated with 
this even POVM cannot distinguish the sign of mean values. 
To construct an estimator that may correctly estimate the true target value, we therefore 
introduce the following second POVM that breaks the above-mentioned symmetry:
\begin{align}
\label{eq:odd_POVM}
    M_1^{(\rm{odd})} &:= I_n-M_0^{(\rm{odd})},\nonumber\\[4pt]
    M_0^{(\rm{odd})} &:= \frac{I_n + \mathcal{O}}{2}.
\end{align}
From the assumption, the projective measurement of $\mathcal{O}$ is implementable.
Here, we recall that the process to measure $Q^m\ket{\bar{0}}$ by this POVM is equivalent to the projective measurement of $\bar{Z}\oplus \left.\mathcal{O}\right|_{\mathcal{S}^\perp}$ on the state $\rho(\alpha;\theta^*)$ with odd $\alpha=2m+1$, as mentioned in Section~\ref{sec:mean_val_est_problem}.
For this reason, we have added the superscript \textit{odd} for the POVM (\ref{eq:odd_POVM}).
Then, the probability of obtaining "1" is given as
\begin{equation}
\label{eq:odd_number_prob}
    \mathbf{P}^{(\rm odd)}_{\mathcal{D}}(m;\theta^*) = \frac{1}{2} - \frac{p_{\rm q}^{2m+1}}{2}\cos{\left[(2m + 1)\theta^*\right]}.
\end{equation}
The coefficient of $\theta^*$ differs by the factor 2 compared to that of $\phi^*$ in Eq.~(\ref{eq:bernoulli_prob_noise}), because the presence of sign for quantum mean values doubles the domain of the target value. 
Since the probability \eqref{eq:odd_number_prob} is an odd function unlike Eq.~(\ref{eq:P_even}), the measurement can distinguish the sign of the target mean value. 
Note that, in contrast to the even POVM, the classical Fisher information corresponding to the odd POVM, $\mathcal{I}^{(\rm odd)}_{\rm c}(m;\theta^*)$, always deviates from the quantum Fisher information even if the number of qubits increases; see Appendix~\ref{apdx:analysis_POVM}.

Hence, our estimator is composed of the even and odd POVMs, where the probability of obtaining the measurement result "1" is summarized as follows:
\begin{align}\label{eq:unified_P}
    &\mathbf{P}_{\mathcal{D}}(\alpha;\theta^*)\notag\\[4pt]
    &:=\begin{dcases}
    \frac{1}{2}-\frac{p_{\rm q}^{\alpha}}{2}\cos{\left(\alpha\theta^*\right)},&\alpha\text{ is odd}\\[6pt]
    \frac{d-1}{d}+p_{\rm q}^{\alpha}\left[\sin^2{\left(\frac{\alpha}{2}\theta^*\right)}-\frac{d-1}{d}\right],&\alpha\text{ is even}
    \end{dcases}.
\end{align}
%
Here $\alpha\in\mathbb{N}$ represents the number of queries to the operator $A$ or $A^\dagger$.
We perform the measurement via Eq.~(\ref{eq:odd_POVM}) when $\alpha$ is odd $(\alpha=2m+1)$ and via Eq.~(\ref{eq:even_POVM}) when $\alpha$ is even $(\alpha=2m+2)$. 
Also, the classical Fisher information corresponding to each measurement is represented as 
\begin{align}
\label{eq:unified_cFI}
    \mathcal{I}_{\rm c}(\alpha;\theta^*)&:=\frac{\alpha^2 p_{\rm q}^{2\alpha}\sin^2{(\alpha\theta^*)}}{4\mathbf{P}_{\mathcal{D}}(\alpha;\theta^*)(1-\mathbf{P}_{\mathcal{D}}(\alpha;\theta^*))}.
\end{align}
%

\subsection{The proposed adaptive estimation algorithm}\label{sec:proposed_alg}

\begin{figure*}
    \centering
    \includegraphics[scale=0.4]{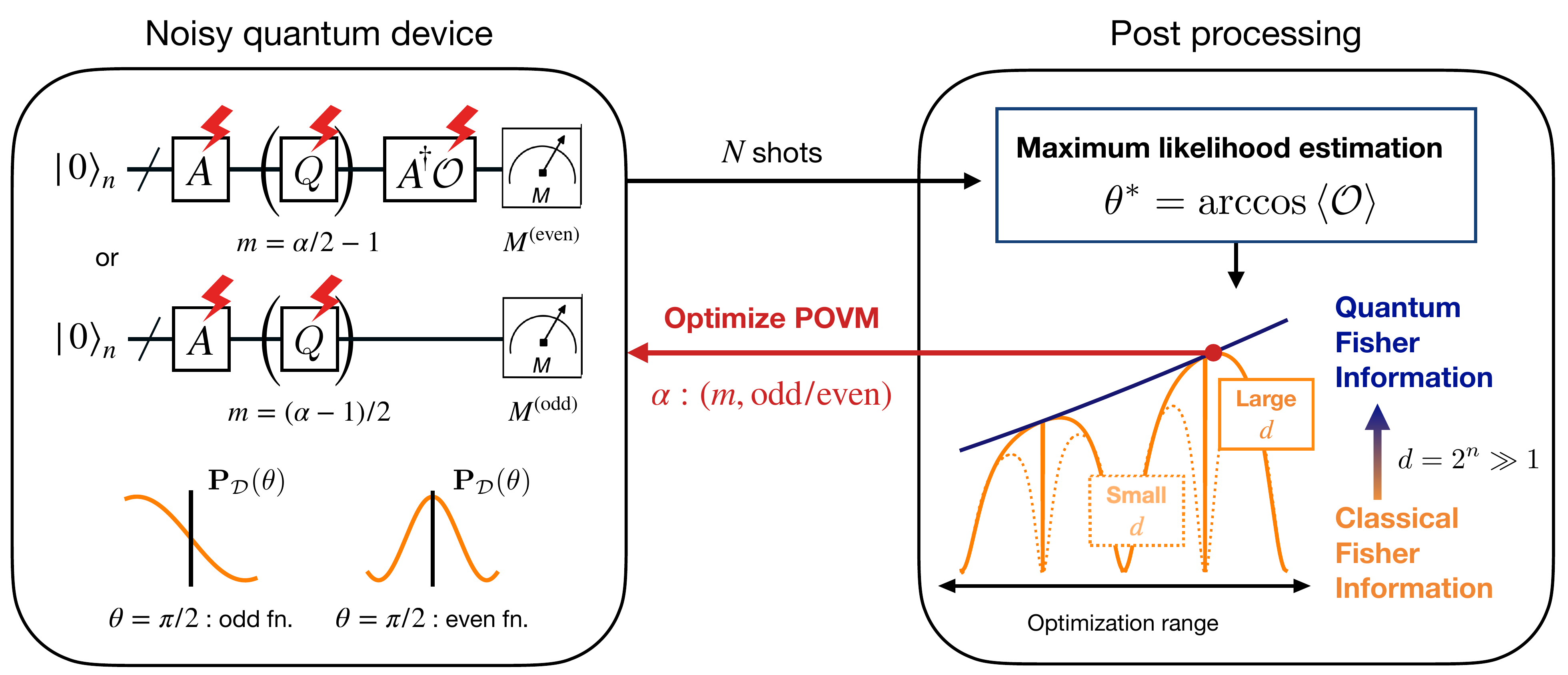}
    \caption{Overview of the proposed algorithm for estimating 
    $\langle\mathcal{O}\rangle=\langle A|\mathcal{O}|A\rangle$. 
    Here, $\ket{A}$ is a target state prepared by operating an $n$-qubit quantum 
    circuit $A$ on the computational basis $\ket{0}_{n}$, and $\mathcal{O}$ is 
    a target Hermitian operator with eigenvalues $\pm 1$. 
    We use the modified maximum likelihood estimation method based on the specific 
    quantum circuits with the amplitude amplification operator $Q$ and the two POVMs 
    $M^{(\rm even)},M^{(\rm odd)}$. 
    The measurement based on the two types of POVM yields probability distributions 
    with different symmetries, and as a result the algorithm can estimate the mean 
    value $\langle A|\mathcal{O}|A\rangle$ correctly including its sign.
    We assume that the depolarization noise is induced when implementing the state preparation $A$ or $A^\dagger$.
    Even under this noisy condition, the algorithm can estimate the target value efficiently in the sense of the Fisher information by optimizing the number of querying $Q$ (i.e., $m$) and the type of POVM (i.e., even or odd), which are encoded together into a 1-dimensional discrete parameter $\alpha$.
    That is, the optimization improves the classical Fisher information so that 
    it gets sufficiently close to the quantum Fisher information, when the system 
    dimension $d=2^n$ is enough large.
    }
    \label{alg:AAS_alg}
\end{figure*}

Based on the discussion so far, the algorithm should satisfy the following conditions: 
\begin{itemize}
\item
When $\sin{(\alpha\theta^*)}\simeq0$, the classical Fisher information 
(\ref{eq:unified_cFI}) becomes nearly zero, and thus the estimation becomes 
inefficient by the Cram\'{e}r-Rao inequality. 
Hence, such $\alpha$ should be avoided, in both POVMs (\ref{eq:even_POVM}) 
and (\ref{eq:odd_POVM}). 
\item
We want to employ measurements whose classical Fisher 
information approaches the quantum Fisher information for enhancing the 
estimation power. 

\item We want the total quantum Fisher information $\mathcal{I}_{\rm q,tot}$ to achieve the upper bound, shown in Theorem~\ref{thm:QFI_limit}; thus $\alpha$ should be chosen to approach $\alpha_{\rm B}={\rm argmax}_{\alpha\in\mathbb{N}}~
\mathcal{I}_{\rm q}(\alpha)/\alpha$, generating the state \eqref{eq:product state with alpha_B}.

\end{itemize}
Here we describe our mean value estimation method that adaptively adjusts the 
amplified level $\alpha$ and thereby satisfies the above three requirements. 
Importantly, this method inherits several good properties of the maximum 
likelihood estimation, as proven in Section~\ref{sec:stat_prop_est_main}.
Note that, in the context of the amplitude estimation algorithms 
described in Section~\ref{method_sec:aewpe}, Ref.~\cite{uno2021modified} 
discusses the similar POVM in Eq.~(\ref{eq:even_POVM}). 
However, unlike our method, the algorithm of \cite{uno2021modified} does not 
employ an adaptive optimization algorithm for estimation; thus it cannot achieve the quantum Fisher information for almost all $\theta^*$ and also suffers from the problem of vanishing classical Fisher information.

Our algorithm is based on the maximum likelihood estimation method for the 
random variables subjected to the probability distribution (\ref{eq:unified_P}). 
An overview is shown in Fig.~\ref{alg:AAS_alg}.
The algorithm is composed of the following five procedures.

\begin{itemize}
    \item[(i)]
    Set the first amplified level $\alpha_1=1$, and set the number of measurements $N$ and the number of steps $M$ as some natural numbers.
\end{itemize}
\begin{itemize}
    \item[(ii)] 
    If the amplified level $\alpha_k$ is even, then 
    we measure $N$ copies of the quantum state $\rho_{\mathcal{D}}(\alpha_k;\theta^*)$ of Eq.~(\ref{eq:amplified_state_in_S_noise})  with the POVM (\ref{eq:even_POVM}).
    If $\alpha_k$ is odd, we perform the POVM (\ref{eq:odd_POVM}) to each of $N$ copies of $\rho_{\mathcal{D}}(\alpha_k;\theta^*)$.
    We write the total number of hitting "1" as $x^{(k)}\in\{0,1,\cdots,N\}$. 
    Here, $x^{(k)}$ is a realization of the Binomial random variable
    \begin{align}
        ~~~X^{(k)}\sim {\rm Bin}(N, \mathbf{P}_{\mathcal{D}}(\alpha_k;\theta^*)).
    \end{align}

    \item[(iii)] Calculate the maximum likelihood estimate $\hat{\theta}_k$ based on the total $k$ measurement results $\bm{x}_k:=(x^{(1)},x^{(2)},\cdots,x^{(k)})$. 
    That is, $\hat{\theta}_k$ maximizes the likelihood function
    \begin{align}\label{eq:hierarchical_structure}
    ~~~\mathcal{L}_k(\theta;\bm{x}_k):=\prod_{l=1}^k F_l\left(x^{(l)};\alpha_l(\bm{x}_{l-1}),\theta\right),
    \end{align}
    where $F_l\left(x^{(l)};\alpha_l(\bm{x}_{l-1}),\theta^*\right)$ is the probability distribution of the random variable $X^{(l)}$.
\end{itemize}
\begin{itemize}
    \item[(iv)] 
    Determine $\alpha_{k+1}$ for the next measurement based on the current estimate $\hat{\theta}_k$ by numerically solving the following optimization problem:
    \begin{equation}\label{eq:alpha_optimization}
    \alpha_{k+1}:=\underset{\alpha\in D_k}{\rm argmax}~\frac{\mathcal{I}_{\rm c}(\alpha;\hat{\theta}_k)}{\alpha}\frac{\sin^2{\left(\alpha\hat{\theta}_k\right)}}{1-\delta\cos^2{\left(\alpha\hat{\theta}_k\right)}},
    \end{equation}
    where $D_k\subset \mathbb{N}$ is a (finite) optimization range and $\delta \in (0, 1]$ is a parameter for regularization.
    
    \item[(v)] Repeating the above procedure (ii)--(iv) for $k=1,2,...,M$, we obtain a final estimate $\hat{\theta}_M$.
    (We omit the optimization in (iv) when $k=M$.)
\end{itemize}

We make some remarks on the proposed algorithm.
First of all, we assume that the probability of depolarization $1-p_{\rm q}$ is identified by some experiments on the quantum device to be used in advance and the calibration on $p_{\rm q}$ is perfect.
The impact of calibration errors of $p_{\rm q}$ on this algorithm is analyzed numerically in Appendix~\ref{apdx:additional_exp}. 
Unlike Bayesian estimation, our algorithm does not require any prior information on the target value.

Then, we note that the measurement $\bm{x}_k$ in (iii) is the collection of all the results 
via the measurements with $\alpha_l~\forall l=1,2,\cdots,k$.
As well as the method in Section~\ref{method_sec:aewpe}, the ambiguity in estimating $\theta^*$, which arises from the periodic nature of $\mathbf{P}_{\mathcal{D}}(\alpha_k;\theta^*)$, can be eliminated by the product of likelihoods with different resolutions.
Because the amplified level $\alpha_k$ depends on the previous results $\bm{x}_{k-1}$ 
(except for the case $k=1$), the likelihood function 
$\mathcal{L}_k(\theta;\bm{x}_k)$ has a hierarchical structure, unlike the case 
of Eq.~(\ref{eq:likelihood_ampest}).

The optimization problem \eqref{eq:alpha_optimization} comes from the third requirement described in the beginning of this subsection; we will provide a detailed explanation in the next two paragraphs.
Since the elements of $D_k$ are integers, the optimization problem can be efficiently solved using a simple brute force method.
Also, it is obvious that the optimization does not yield an amplified level such that $\sin{(\alpha\hat{\theta}_{k})}\simeq0$, and thus our method can avoid vanishing classical Fisher information regardless of the target value.
It should be noted that this algorithm finally outputs an estimate $\hat{\theta}_M$ together with the realized amplified levels $\{\alpha_k\}_{k=1}^M$ in a single run which uses $N_{\rm q}=N\sum_{k=1}^M \alpha_k$ queries to $A$ and $A^\dagger$ in total.

Here, we 
explain the meaning of the term 
$\mathcal{I}_{\rm c}/\alpha$ 
in the optimization problem \eqref{eq:alpha_optimization}, 
which is the objective function when $\delta=1$ (i.e., no regularization). 
By Theorem~\ref{thm:opt_marginalPOVM}, the objective function can be written as $\kappa_{\infty}\mathcal{I}_{\rm q}(\alpha)/\alpha$ for even $\alpha$ when $d=2^n$ is sufficiently large, because we can assume $\cos{(\alpha\hat{\theta}_k/2)}\neq 0$ without loss of generality.
Now, $\kappa_{\infty}\in[0,1)$ is a periodic function while $\mathcal{I}_{\rm q}(\alpha)/\alpha$ is a unimodal function, with respect to $\alpha$; the former function oscillates much faster than the change of the latter unless the noise level is too big, resulting that the maximum point of the objective function $\kappa_{\infty}\mathcal{I}_{\rm q}(\alpha)/\alpha$ is close to a point that maximizes $\mathcal{I}_{\rm q}(\alpha)/\alpha$ within $D_k$.
Furthermore, the numerical simulation shown later suggests that, at around the maximum point of $\kappa_{\infty}\mathcal{I}_{\rm q}(\alpha)/\alpha$, we observe $\kappa_{\infty}\approx 1$. 
This means that, from Theorem~\ref{thm:opt_marginalPOVM}, our POVM becomes nearly optimal i.e., $\mathcal{I}_{\rm c}\approx \mathcal{I}_{\rm q}$.
Note that the odd classical Fisher information is comparable to the quantum Fisher information when $\alpha_{k+1}$ takes a relatively small value, and therefore the odd $\alpha$ (i.e., $\alpha=2m+1$) also takes an important role in such a regime.

Importantly, if we set $D_k$ such that $D_k\subset D_{k+1}$ and $\alpha_{\rm B}\in D_{k'}$ for some $k'$, then the output of this optimization problem converges to $\alpha_{\rm B,c}(\theta^*)$ as the step $k$ increases, under the condition $\hat{\theta}_k\to\theta^*$ which is actually valid as proven in the next subsection.
Here, $\alpha_{\rm B,c}(\theta^*)$ is defined as a solution of Eq.~(\ref{eq:alpha_optimization}) with $D_k=\mathbb{N}$ and $\hat{\theta}_k=\theta^*$.
The above discussion says that the factor $\alpha_{\rm B,c}(\theta^*)$ would be close to $\alpha_{\rm B}={\rm argmax}_{\alpha\in\mathbb{N}}~
\mathcal{I}_{\rm q}(\alpha)/\alpha$, which yields the state \eqref{eq:product state with alpha_B} and thereby achieves the upper bound of the total quantum Fisher information.
Thus, our algorithm uses the maximal quantum enhancement specified by $\alpha_{\rm B}$ in Theorem~\ref{thm:QFI_limit} as the step $k$ increases.
In the following section, we numerically verify these desirable properties of our algorithm.

The optimization presented above has the following issue that forced us to introduce a regularization function.
That is, if $\kappa_{\infty}$ is too close to one, the probability $1-\mathbf{P}_{\mathcal{D}}(\alpha;\theta^*)$ becomes small for such an $\alpha$ when $d \gg 1$, meaning that we need a large number of shots $N$ to have a good estimate.
To circumvent this issue, we introduce a regularization function ${\rm Reg}(\alpha;\hat{\theta}_k)$ such that ${\rm Reg}\simeq 0$ at $\kappa_{\infty} \simeq 1$ and take the objective function $\kappa_{\infty}\mathcal{I}_{\rm q}/\alpha \times {\rm Reg}$ rather than $\kappa_{\infty}\mathcal{I}_{\rm q}/\alpha$.
Then, $\kappa_{\infty}\mathcal{I}_{\rm q}/\alpha \times {\rm Reg}$ is not maximized at an $\alpha$ such that $\kappa_{\infty}$ is too close to 1. 
But we want $\kappa_{\infty}$ to be close enough to 1, hence the regularization term ${\rm Reg}$ should satisfy ${\rm Reg}\simeq 0$ only in the narrow region of $\alpha$ satisfying $\kappa_{\infty} \simeq 1$.
We employ the following regularization function satisfying those requirements:
\begin{align}\label{eq:reg}
    {\rm Reg}(\alpha;\hat{\theta}_k)
    :=\frac{\sin^2{\left(\alpha\hat{\theta}_k\right)}}
           {1-\delta\cos^2{\left(\alpha\hat{\theta}_k\right)}}.
\end{align}
When we employ the parameter $\delta$ closer to 1, the regularization function becomes sharper around $\alpha$ such that $\kappa_{\infty} \simeq 1$.
Hence, we arrive at the objective function given in Eq.~(\ref{eq:alpha_optimization}).
In the numerical simulation shown later, we choose $\delta=0.95$.

Also, the optimization ranges of $\{D_k\}$ control the following trade-off.
That is, a large $\alpha$ is preferred to enhance the Fisher information per shot, but if $\alpha$ dramatically (e.g., super-exponentially) increases with respect to $k$, then the estimation will be failed for a practical $N$ because the likelihood function $\mathcal{L}_{k}(\theta;\bm{x}_k)$ has many peaks around the true target value.
A more detailed analysis related to this trade-off can be found in Ref.~\cite{hayashi2018resolving}.

\subsection{Statistical properties of the estimator}\label{sec:stat_prop_est_main}

Although the random variables $X^{(k)}$ describing the measurement results depend on each other and have the hierarchical structure as shown in 
Eq.~(\ref{eq:hierarchical_structure}), we can prove the following desirable 
statistical property of the estimator; i.e., the consistency.

\begin{thm}\label{thm:consistency}
  {\rm(Consistency; informal version.)} There exists a unique maximum likelihood 
  estimator $\hat{\theta}_M$ of $\mathcal{L}_M(\theta;\boldsymbol{X}_M)$ such that
  \begin{equation}\label{eq:consistency}
    \cos{\hat{\theta}_M}\to\langle\mathcal{O}\rangle,
  \end{equation}
  where $\to$ means the convergence in probability as $N\to\infty$.
  In addition, the random variable $\alpha_k(\bm{X}_{k-1})$ converges 
  to a constant $\alpha_k^*$ in the sense of probability as follows
  \begin{align}
      \lim_{N\to \infty}P\left[\alpha_k(\bm{X}_{k-1})
         =\alpha_k^*\right]= 1,~~~k=2,\cdots,M,
  \end{align}
where $\alpha_k\left(\bm{x}_{k-1}\right)$ is the optimized amplified level based 
on the measurement results $\bm{x}_{k-1}$; that is, $\{ \alpha_k \}$ become independent with each other, as $N \to \infty$.
\end{thm}

The proof is given in Appendix~\ref{apdx:stat_prop}, where $\alpha_k^*$ is given in Eq.~\eqref{apdx:alpha_star}.
Theorem~\ref{thm:consistency} means that the proposed estimator is valid in the sense that more data 
leads to more accurate estimation. 
That is, unlike the previous method~\cite{wang2021minimizing} which uses Bayesian 
inference based on the measurement with the odd POVM (\ref{eq:odd_POVM}) (not the 
even POVM), the estimate approaches the true value with high probability as the 
number of measurements $N$ increases, as implied by Eq.~(\ref{eq:consistency}). 
Note that the Bayesian inference may return estimates far from the true value 
as discussed in~\cite{johnson2022reducing}.

Furthermore, the asymptotic variance of the estimator can achieve the Cram\'{e}r-Rao lower bound when $N$ is large.
To show this result, we focus on the following total classical/quantum Fisher 
information in our method:
\begin{align}
    \mathcal{I}_{\rm c,tot}\left(\theta^*\right)&=N\sum_{k=1}^M\mathbf{E}_{\bm{X}_{k-1}}\left[ \mathcal{I}_{\rm c}\left(\alpha_k\left(\bm{X}_{k-1}\right);\theta^*\right)\right],\\
    \mathcal{I}_{\rm q,tot}\left(\theta^*\right)&=N\sum_{k=1}^M\mathbf{E}_{\bm{X}_{k-1}}\left[ \mathcal{I}_{\rm q}\left(\alpha_k\left(\bm{X}_{k-1}\right)\right)\right].
\end{align}
The derivation is provided in Appendix~\ref{apdx:stat_structure}.
Since Theorem~\ref{thm:consistency} states that $\alpha_k\left(\bm{X}_{k-1}\right)$ converges to a 
constant $\alpha_k^*$ with high probability as $N$ increases, the total Fisher 
information also converges as follows;
\begin{align}\label{eq:conv_tot_Fisher}
    \frac{\mathcal{I}_{\rm c/q,tot}(\theta^*)}{N}&\to \frac{\mathcal{I}^*_{\rm c/q,tot}(\theta^*)}{N},
\end{align}
where $\mathcal{I}^*_{\rm c/q,tot}(\theta^*)$ is the asymptotic value of the total classical/quantum Fisher information defined as
\begin{align}
    \mathcal{I}^*_{\rm c,tot}\left(\theta^*\right)&:=N\sum_{k=1}^M \mathcal{I}_{\rm c}\left({\alpha}^*_k;\theta^*\right),\label{eq:asym_FI_1}\\
    \mathcal{I}^*_{\rm q,tot}\left(\theta^*\right)&:=N\sum_{k=1}^M \mathcal{I}_{\rm q}\left({\alpha}^*_k\right).
\label{eq:asym_FI}
\end{align}
Note that $\alpha_k^*$ depends on the target value $\theta^*$ implicitly as in Eq.~\eqref{apdx:alpha_star}.
Here, we provide the convergence theorem for the distribution of our estimator.

\begin{thm}\label{thm:asym_normal}
  {\rm(Asymptotic normality; informal version.)} If $\hat{\theta}_M$ is a maximum likelihood estimator of $\mathcal{L}_{M}(\theta;\bm{X}_M)$, then the following convergence holds;
  \begin{align}\label{eq:asym_norm}
    \sqrt{\mathcal{I}^{*}_{\rm c,tot}\left(\arccos{\langle\mathcal{O}\rangle}\right)}&\left(\cos\hat{\theta}_M-\langle\mathcal{O}\rangle\right)\notag\\
    &\to\mathcal{N}\left(0,1-\langle\mathcal{O}\rangle^2\right),
  \end{align}
  where $\mathcal{N}(0,1-\langle\mathcal{O}\rangle^2)$ denotes a centered normal distribution with variance $1-\langle\mathcal{O}\rangle^2$, and $\to$ means the convergence in distribution as $N$ increases.
\end{thm}

This theorem shows that the distribution of the estimator $\cos{\hat{\theta}_M}$ 
gets arbitrarily close to the normal distribution with mean $\braket{\mathcal{O}}$ 
in large $N$. 
In addition, the asymptotic variance of the estimator is proportional to the inverse of the asymptotic value of the total classical Fisher information. 
Thus, our estimator can asymptotically achieve the classical Cram\'{e}r-Rao lower bound. 
In the next section, we numerically demonstrate how fast the above two asymptotic properties behave with respect to $N$.

As mentioned in Section~\ref{sec:proposed_alg} and demonstrated later, the adaptive optimization (\ref{eq:alpha_optimization}) yields the nearly optimal measurement (i.e., $\mathcal{I}_{\rm c}\approx\mathcal{I}_{\rm q}$) at each step in a large number of qubits.
Furthermore, the output $\alpha_k$ of \eqref{eq:alpha_optimization} converges to $\alpha_{\rm B,c}(\theta^*)\approx \alpha_{\rm B}$ as the step $k$ increases for an appropriate choice of $\{D_k\}$.
As a result, the properties of our algorithm lead to the following relations regarding the total Fisher information:
\begin{align}
\label{eq:qfi_vs_cfi}
    {\mathcal{I}^*_{\rm c,tot}(\theta^*)}\approx {\mathcal{I}^*_{\rm q,tot}(\theta^*)}\underset{M\gg 1}{\approx} N^*_{\rm q} \frac{\mathcal{I}_{\rm q}(\alpha_{\rm B})}{\alpha_{\rm B}},
\end{align}
where $N^*_{\rm q}$ denotes the total number of queries for the asymptotic sequence $\{\alpha^*_k\}$:
\begin{equation}\label{eq:asym_query_main}
    N^*_{\rm q}:=N\sum_{k=1}^M \alpha_k^*.
\end{equation}
Note that this can be rewritten as $N^*_{\rm q}=\sum_{l=1}^{M'}\alpha'_l$ for $M'=NM$ and $\alpha'_l=\alpha^*_k$ if $(k-1)N+1\leq l\leq kN$ in the notation of Theorem~\ref{thm:QFI_limit}.
Importantly, the most right hand side in Eq.~(\ref{eq:qfi_vs_cfi}) matches the upper bound in Theorem~\ref{thm:QFI_limit}, where we use $N^*_{\rm q}~(\geq \alpha_{\rm B})$ queries to $A$ or $A^\dagger$ in total.
In the next section, we demonstrate that the approximation errors are sufficiently 
small for \textit{almost all} $\theta^*$, via an appropriate choice of 
$\{D_k\}_{k=1}^{M-1}$; 
that is, together with the asymptotic normality, our 
estimator nearly achieves the ultimate precision given by the inverse of total 
quantum Fisher information \eqref{eq:fundamental_precision_limit2}.

\subsection{Summary of the proposed protocol}\label{sec:summary}
Here we summarize the notable feature of the proposed adaptive method, with remarks on which part is theoretically guaranteed and which part is left for numerical verification.

The goal is to ultimately estimate $\langle\mathcal{O}\rangle=\cos\theta^*$, for the given $n$-qubit implementable operators $A$ and $\mathcal{O}$; 
note that, likewise the Grover operator, the implementability of these operators does not mean that all the operator components are known and accordingly $\langle\mathcal{O}\rangle$ is directly computable.
The noise parameter $p_{\rm q}$ are assumed to be known. 
Another important assumption is that $d=2^n\gg 1$. 
We do not need other prerequisite for the target system.

The estimation procedure is illustrated in Fig.~\ref{alg:AAS_alg}, where the amplification level $\alpha_k$, which determines the number of querying $Q$ and the type of POVM, is updated via Eq.~\eqref{eq:alpha_optimization} at each step $k$ (each step contains $N$ shots).
This optimization \eqref{eq:alpha_optimization} is the maximization of classical Fisher information per queries, $\mathcal{I}_{\rm c}/\alpha$, reflecting the goal to achieve quantum Fisher information $\mathcal{I}_{\rm q}$ at each step and $\alpha_k\to \alpha_{\rm B}:={\rm argmax}_{\alpha\in\mathbb{N}}~
\mathcal{I}_{\rm q}(\alpha)/\alpha$.
Here, $\alpha_{\rm B}$ provides the best coherent use of noisy $A$ or $A^\dagger$ in estimating $\theta^*$, which is proved in Theorem~\ref{thm:QFI_limit}.
From the analysis of our POVM \eqref{eq:even_POVM} in Theorem~\ref{thm:opt_marginalPOVM}, the objective function $\mathcal{I}_{\rm c}/\alpha$ can be written as $\kappa_{\infty}\mathcal{I}_{\rm q}/\alpha$, where $\kappa_{\infty}$ is a bounded periodic function regarding $\alpha$.
Then, to achieve the above goal, the following properties should be satisfied: (i) max~$\kappa_{\infty}\mathcal{I}_{\rm q}/\alpha$ in a given optimization range $D_k~(\ni \alpha)$ is close to max~$\mathcal{I}_{\rm q}/\alpha$ in $D_k$ and (ii) $\kappa_{\infty}\approx 1$ (i.e., $\mathcal{I}_{\rm c}\approx \mathcal{I}_{\rm q}$) at the maximum point.
Note that we gradually expand the range of $D_k$ as $k$ increases to eliminate the estimation ambiguity; see Section~\ref{method_sec:aewpe}.
In Section~\ref{sec:proposed_alg}, we provide an intuitive explanation (not a rigorous proof) that the properties (i) and (ii) hold, and furthermore, these properties will be numerically verified in the next section.

On the other hand, Theorems~\ref{thm:consistency} and \ref{thm:asym_normal} guarantee the statistical properties of the estimator.
More precisely, Theorem~\ref{thm:consistency} guarantees the consistency, stating that the estimate approaches the true value with high probability as the number of shots $N$ increases.
Also, Theorem~\ref{thm:asym_normal} guarantees the asymptotic normality, stating that the estimator asymptotically achieves the classical Cramér-Rao lower bound $1/\mathcal{I}^*_{\rm c,tot}(\theta^*)$.
Here, the (asymptotic) total classical/quantum Fisher information $\mathcal{I}^*_{\rm c,tot}(\theta^*)$ and $\mathcal{I}^*_{\rm q,tot}(\theta^*)$ are defined in Eqs.~\eqref{eq:asym_FI_1} and \eqref{eq:asym_FI}, respectively.
Therefore, if the nearly optimal sequence $\alpha_k$ is properly chosen, the estimator achieves the estimation precision $1/\mathcal{I}^*_{\rm c,tot}(\theta^*)\approx 1/\mathcal{I}^*_{\rm q,tot}(\theta^*)$, and furthermore, it also achieves the precision ${\alpha_{\rm B}}/{(N^*_{\rm q}\mathcal{I}_{\rm q}(\alpha_{\rm B}))}$ for the (asymptotic) total queries $N^*_{\rm q}$ in Eq.~\eqref{eq:asym_query_main}, when the iteration $k$ is sufficiently large.
This matches the precision limit in Theorem~\ref{thm:QFI_limit} (more precisely, Eq.~\eqref{eq:fundamental_precision_limit2}).
We numerically find that this happens in a moderate number of shots $N$, as will be demonstrated in Figs.~\ref{fig:rmse_vs_crbound}, \ref{apdx_fig:prob_conv} and \ref{fig:FI_dist_cfivsqfi}.

\section{Numerical Experiment}\label{sec:results}

In this section, we numerically verify the performance of the proposed algorithm. 
First, in the practical (i.e., a few hundred of) measurement number $N$, we 
demonstrate that the asymptotic properties given by Theorems~\ref{thm:consistency} and~\ref{thm:asym_normal} hold.
Next, by evaluating the asymptotic value of the total Fisher information, we confirm that the proposed algorithm retains a large classical Fisher information regardless of the target value $\theta^*$.
Also, we show that the desirable relations given by Eq.~(\ref{eq:qfi_vs_cfi}) hold for a system of dozens of qubits, that is, the total classical Fisher information becomes sufficiently close to the total quantum Fisher information for such systems.

\subsection{Asymptotic properties with respect to the number of measurements}
\label{sec:confirm_asym_prop}

\begin{figure*} 
 \centering
 \begin{tabular}{ccc}
     \includegraphics[scale=0.8]{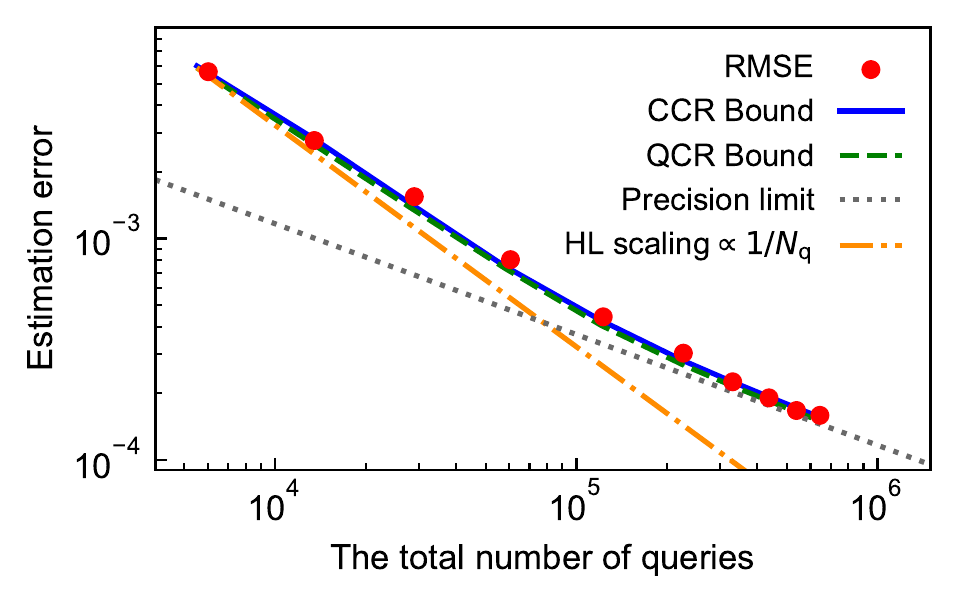}&&\includegraphics[scale=0.8]{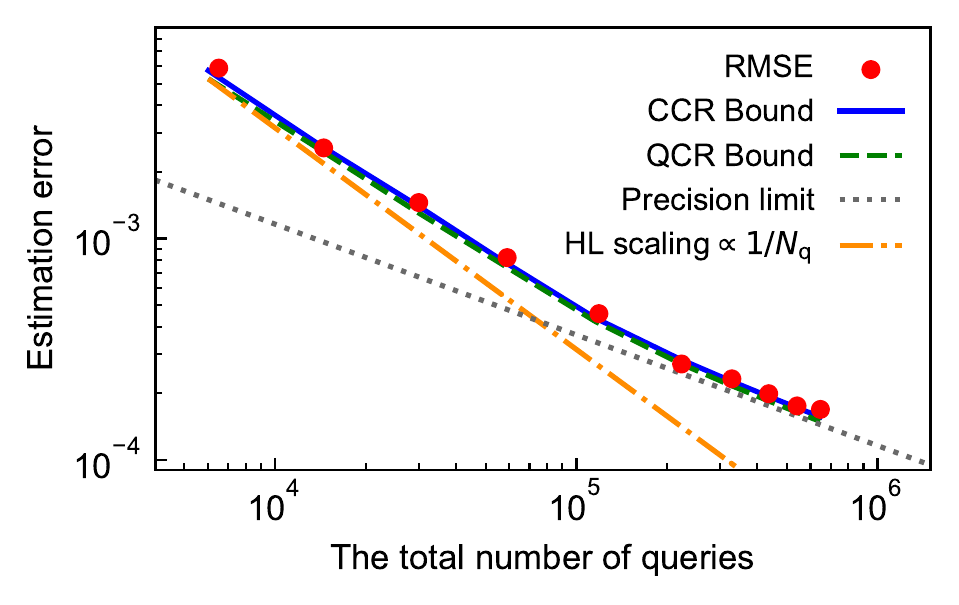}\\
     (a) $\cos\theta^*=0.042$&&(b) $\cos\theta^*=-0.1$\\
     &&\\
     \includegraphics[scale=0.8]{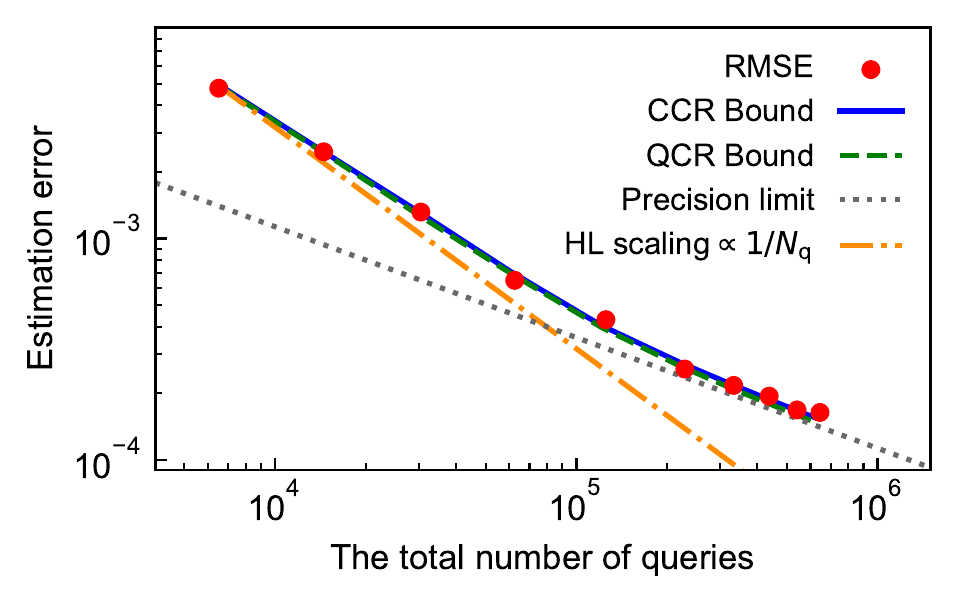}&&\includegraphics[scale=0.8]{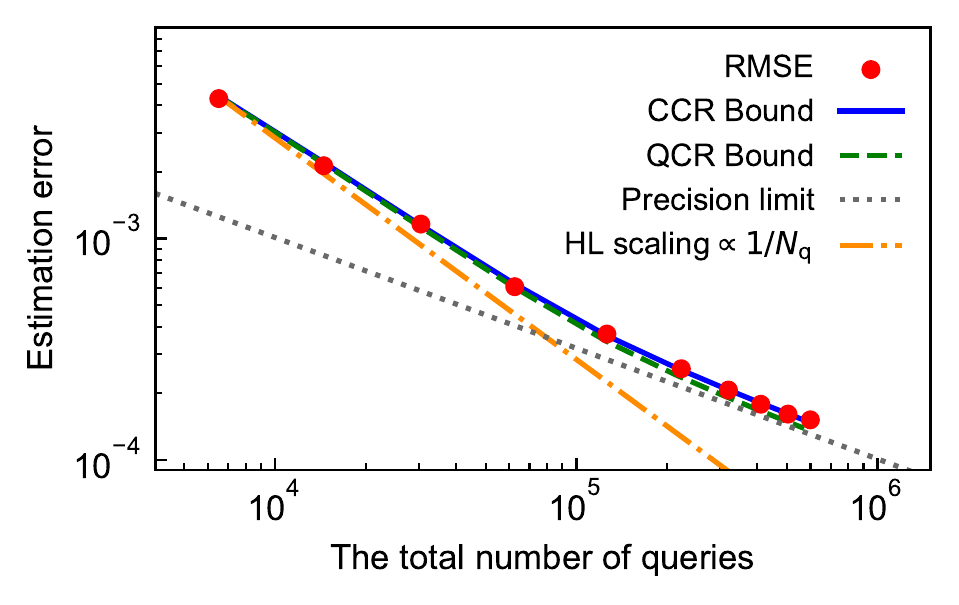}\\
     (c) $\cos\theta^*=0.25$&&(d) $\cos\theta^*=0.5$\\
     &&\\
     \includegraphics[scale=0.8]{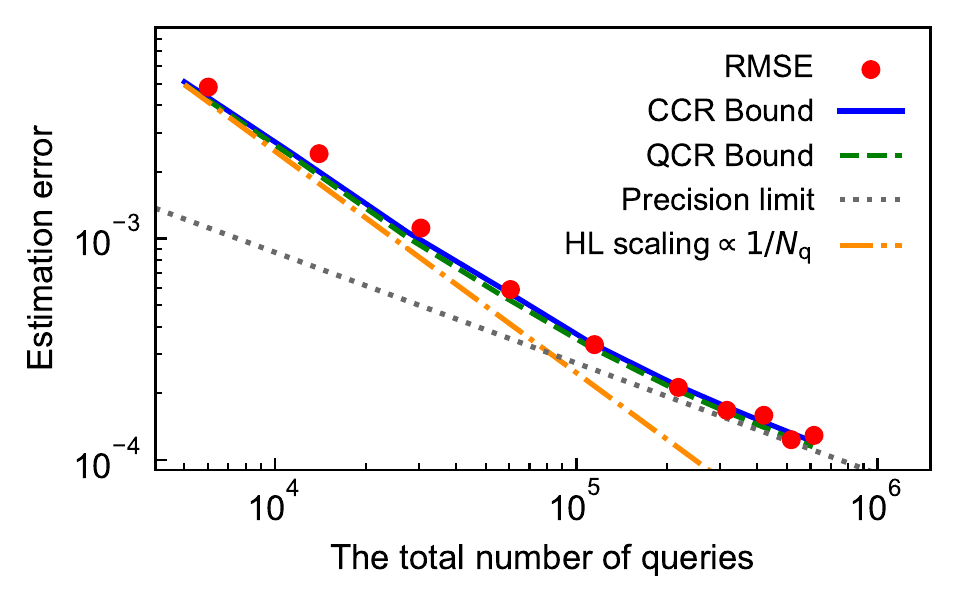}&&\includegraphics[scale=0.8]{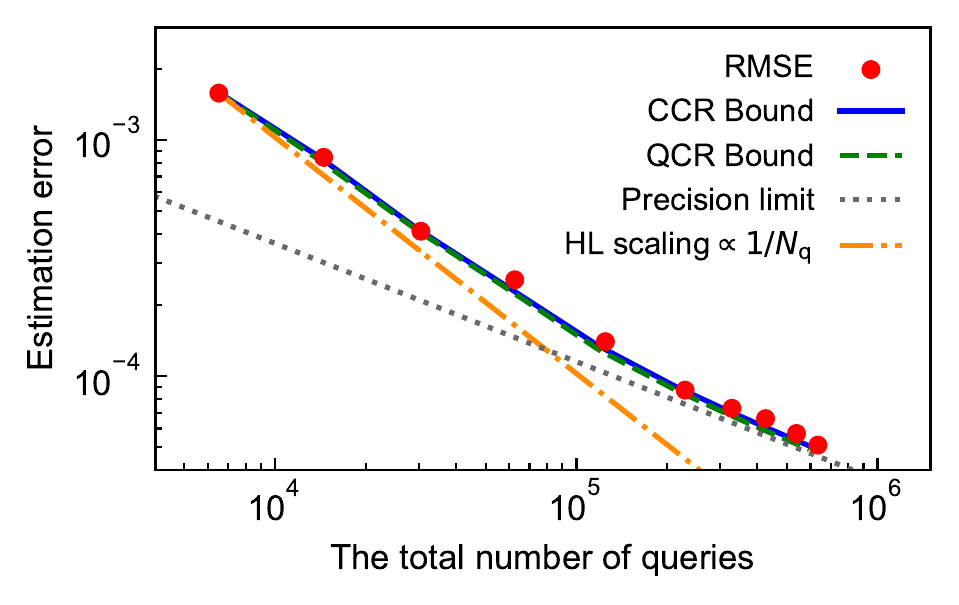}\\
     (e) $\cos\theta^*=-0.67$&&(f) $\cos\theta^*=0.95$\\
     &&
 \end{tabular}
 \caption{The estimation error of $\braket{\mathcal{O}}=\cos{\theta^*}$ and the total 
 number of queries for several target values (a)--(f).
 The x-axis shows the total number of queries to the state preparation $A$ or $A^\dagger$ in a single trial of our method.
 The root mean squared error (RMSE) of the estimator $\cos{\hat{\theta}}$ is estimated by the averaged value over 300 trials, as in Eq.~(\ref{rmse sample mean}).
 Since the number of total queries in a single trial of our method varies stochastically, the corresponding x-axis value of the RMSE dot is taken as the maximum value of total queries in the 300 trials.
 The blue solid and green dashed lines represent the asymptotic values of CCR/QCR bounds obtained from the corresponding classical/quantum Fisher information $\mathcal{I}^*_{\rm c/q, tot}(\theta^*)$, respectively.
 The orange dash-dotted and the gray dotted lines represent the Heisenberg-limited scaling and the precision limit derived in Theorem~\ref{thm:QFI_limit} (in the case of $N_{\rm q}>\alpha_{\rm B}=199$), respectively.
 }
 \label{fig:rmse_vs_crbound}
\end{figure*}

\begin{figure*}
 \centering
 \begin{tabular}{ccc}
 \includegraphics[scale=0.3]{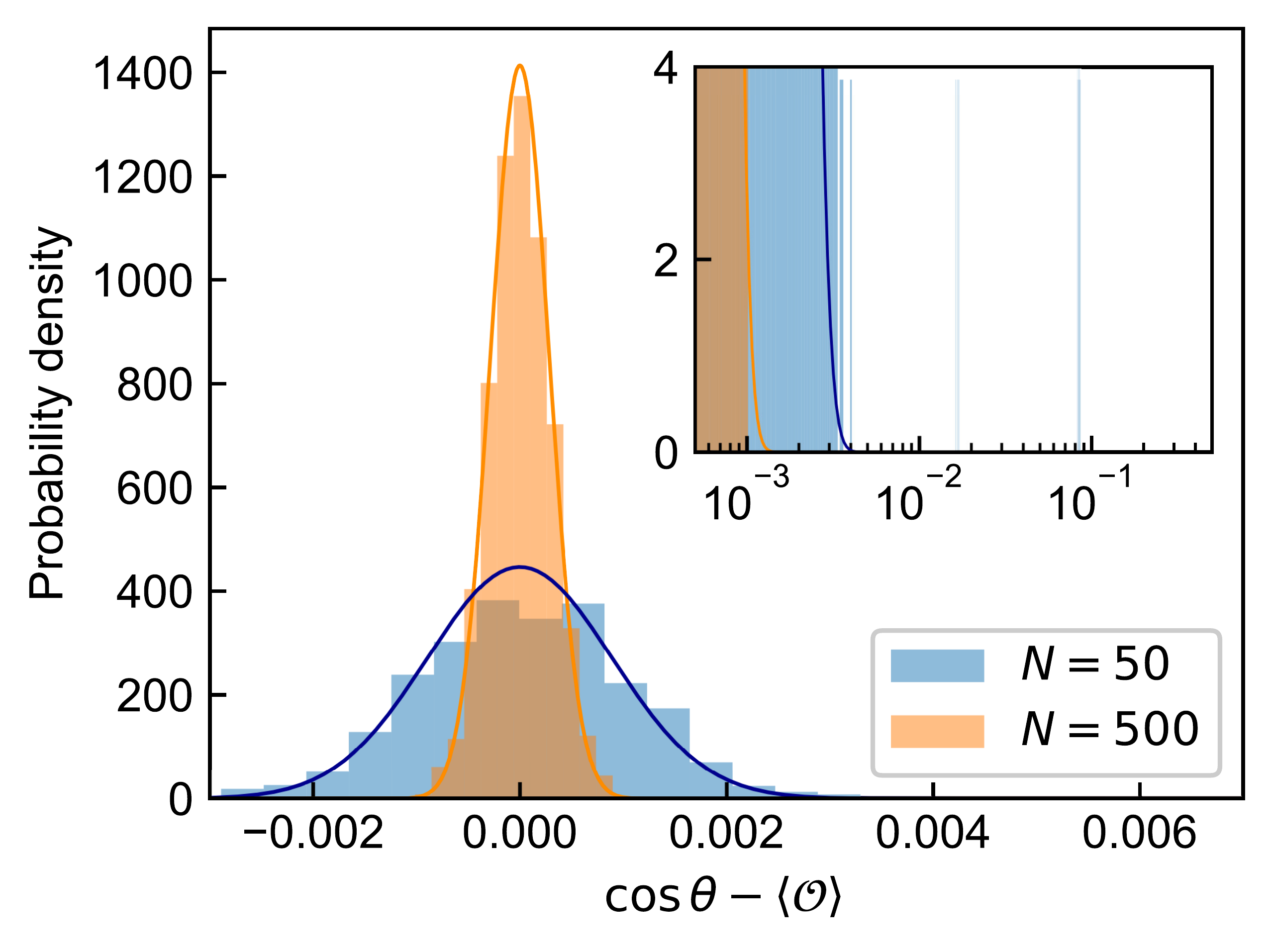}&&\includegraphics[scale=0.2944]{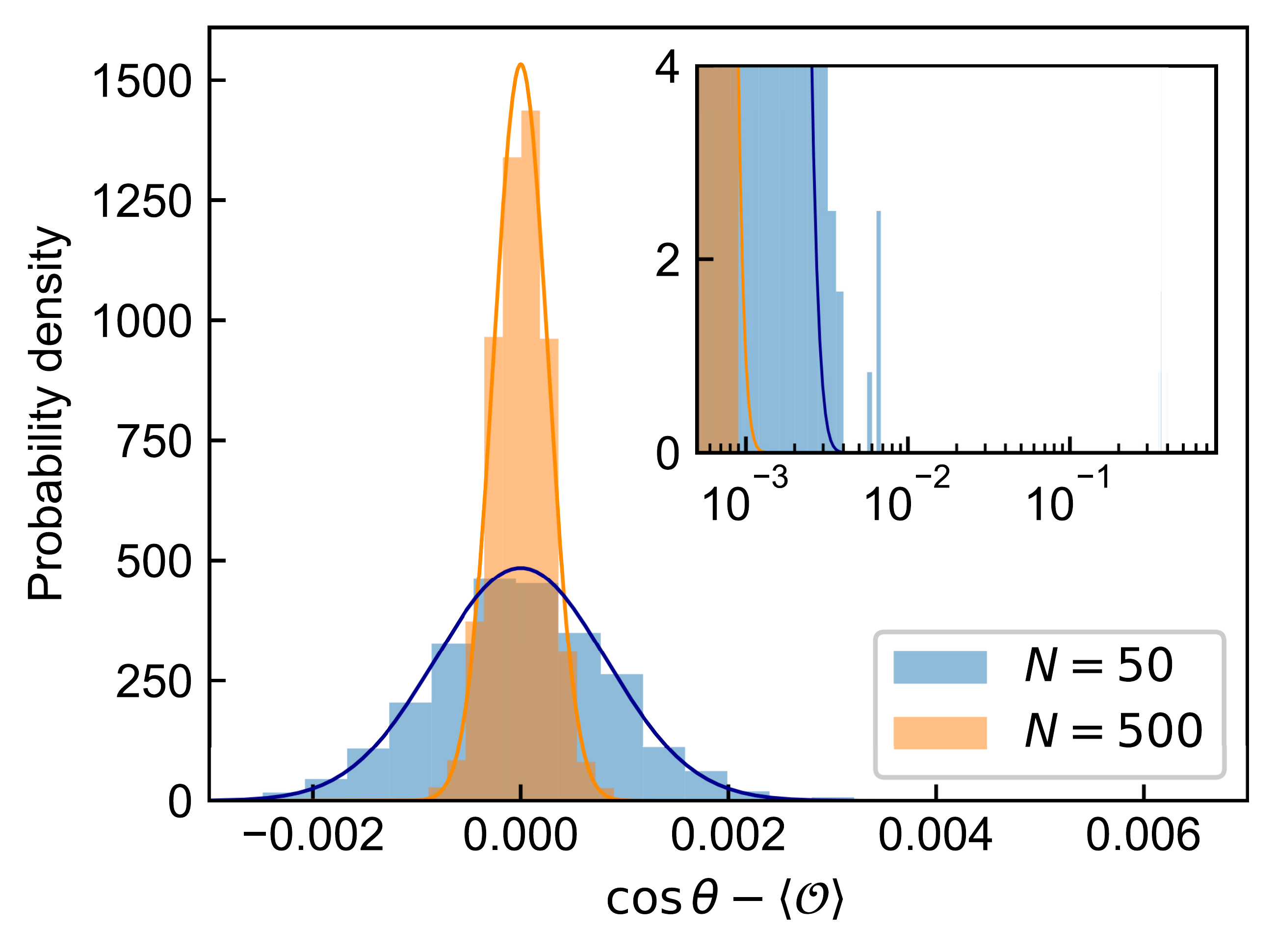}\\
 (a) $\braket{\mathcal{O}}=0.042$&&(b) $\braket{\mathcal{O}}=0.5$\\
 \end{tabular}
 \caption{The probability density of the maximum likelihood estimator $\cos{\hat{\theta}_M}$ in the case of $M=8$.
 The blue and orange plots correspond to the results with $N=50$ and $N=500$, respectively.
 The histograms represent the empirical probability density over 3000 trials, and the solid lines denote the probability density of the centered normal distribution with the variance corresponding to the asymptotic CCR bound.
 The upper right panels show the same plots (but the x-axis represents the absolute error) with much wider range in the horizontal axis to show the existence of outliers.
 The larger $N$ becomes, the sharper the density gets around the target value, which demonstrates the consistency in Theorem~\ref{thm:consistency}.
 Also, the outliers vanish in large $N$; thus the empirical density is nearly identical to the normal distribution, which is exactly the asymptotic normality in Theorem~\ref{thm:asym_normal}.
 }
 \label{apdx_fig:prob_conv}
\end{figure*}

To demonstrate the properties of our algorithm, we numerically evaluate the root mean squared error (RMSE) of $\cos{\hat{\theta}}$ defined as
\begin{align}\label{rmse sample mean}
    {\rm RMSE}\left[\cos{\hat{\theta}}\right]&:=\sqrt{{\mathbf{E}_{\hat{\theta}}}\left[\left(\cos{\hat{\theta}}-\cos{\theta^*}\right)^2\right]}\notag\\[4pt]
    &\simeq \sqrt{\frac{1}{\#}\sum_{i=1}^{\#}\left(\cos{\hat{\theta}[i]}-\cos{\theta^*}\right)^2},
\end{align}
where $\#$ is the total number of trials, and $\hat{\theta}[i]$ denotes the estimate of $i$th trial.
In the following, we use $\#=300$ samples to evaluate the RMSE.
We assume the depolarization noise with the parameter $p_{\rm q}=0.995$.
The number of qubits is set to 20, which corresponds to the system dimension $d=2^{20}$.
We also fix the number of measurements as $N=500$ for each circuit.
To obtain the maximum likelihood estimates, we used a modified brute force method, in which the search domain becomes narrowed as the measurement process proceeds.
The amplified level $\alpha_k$ of the $k$th measurement process $(k=2,...,M)$ is 
determined by solving the optimization problem (\ref{eq:alpha_optimization}), 
where the maximum likelihood estimate $\hat{\theta}_{k-1}$ has been obtained at 
the $(k-1)$th step. 
Also, in this work, the optimization range is chosen as 
\begin{equation}
    D_{k-1}:=\{2,2+1,\cdots,2^{k}\}.
\end{equation}
The exponential increase of the number of elements in $D_{k-1}$ is inspired by the fact that, when there is no noise, the exponential increment sequence $m_k=2^{k-1}$ achieves the Heisenberg-limited scaling~\cite{suzuki2020amplitude}; also see the discussion below Eq.~\eqref{prelimi_optim phase}.

Figure~\ref{fig:rmse_vs_crbound} shows the relationship between the RMSE Eq.~\eqref{rmse sample mean} and the total number of queries $N_{\rm q}$ calculated as 
\begin{equation}
\label{total q}
    N_{\rm q} = N\sum_{k=1}^{M}\alpha_k. 
\end{equation}
Here, we remark that $N_{\rm q}$ can be rewritten as $N_{\rm q}=\sum_{l=1}^{M'}\alpha_l'$ for $M'=NM$ and $\alpha'_l=\alpha_k$ if $(k-1)N+1\leq l\leq kN$ in the notation of Theorem~\ref{thm:QFI_limit}.
The red dots indicate the RMSE with $M=3,4,\cdots,12$ from left to right in each 
panel.
Since $N_{\rm q}$ is a random variable due to the randomness of $\alpha_k$, we plot 
the RMSE as a function of the maximum value of $N_{\rm q}$ in $\#=300$ trials, 
and thus the RMSE dots are overestimated with respect to the number of queries.
In addition, the classical Cram\'{e}r-Rao (CCR) lower bound and the quantum Cram\'{e}r-Rao (QCR) lower bound are depicted with the solid blue and dashed green 
lines, respectively; these are calculated by substituting the true value $\theta^*$ and the asymptotic sequence $\{\alpha_k^*\}_{k=1}^M$ into $\mathcal{I}^*_{\rm c/q,tot}\left(\theta^*\right)$ in Eqs.~(\ref{eq:asym_FI_1}) and (\ref{eq:asym_FI}).
Note that the difference between the bounds obtained from the $\#=300$ average 
of the total Fisher information (not shown in the figure), $\mathcal{I}_{\rm c/q,tot}\left(\theta^*\right)$, and the asymptotic CCR/QCR bounds from $\mathcal{I}^*_{\rm c/q,tot}\left(\theta^*\right)$ can be negligible, which clearly supports Eq.~(\ref{eq:conv_tot_Fisher}).
The orange dash-dotted and gray dotted lines represent the Heisenberg-limited (HL) scaling ${\rm RMSE}=O(1/N_{\rm q})$ and the precision limit given by Eq.~(\ref{eq:fundamental_precision_limit2}) with $1-\langle\mathcal{O}\rangle^2$, respectively.

Several important features are observed.
First, the CCR and QCR bounds are close to each other; the first part of Eq.~(\ref{eq:qfi_vs_cfi}) holds.
Note that this good approximation holds even when the target value $\theta^*/2\pi$ is a rational number as shown in the panel~(d); actually, we will see in the next subsection that this approximation (i.e., the first part of \eqref{eq:qfi_vs_cfi}) holds for almost all $\theta^*$ in the numerical simulation.
We then find that the RMSE almost achieves the CCR and accordingly QCR bounds in all the cases of six target values.
That is, $N=500$ is sufficient to obtain the asymptotic normality (\ref{eq:asym_norm}) in the chosen noise condition.
We here remark that by taking the number of shots $N$ (originally, we fix $N=500$ for all iteration steps) as a linear function $N_k$ of the iteration step $k$, 
as was done in Ref.~\cite{higgins2009demonstrating,berry2009perform}, the RMSE of our method can nearly achieve the QCR bound even for a regime with smaller total queries ($N_{\rm q}\sim O(10^2)$).
Moreover, the empirical probability density of our estimator shown in Fig.~\ref{apdx_fig:prob_conv} converges to the normal distribution whose variance corresponds to the asymptotic CCR bound as $N$ increases, which directly demonstrates the asymptotic normality~(\ref{eq:asym_norm}).
Recall that this key property holds because the amplified levels are almost identical to $\{\alpha_k^*\}_{k=1}^M$ due to Theorem~\ref{thm:consistency}, which is demonstrated in Fig.~\ref{apdx_fig:alpha_k} in Appendix~\ref{apdx:stat_prop}.

Moreover, as the step $M$ increases, we confirm that all of the RMSE, the CCR bound, and the QCR bound nearly achieve the precision limit; the second part of Eq.~(\ref{eq:qfi_vs_cfi}) also holds, in addition to the first part of the relations and Theorem~\ref{thm:asym_normal}.
Therefore, our algorithm estimates the mean value $\langle\mathcal{O}\rangle=\cos\theta^*$ with precision close enough to the ultimate limit given by Theorem~\ref{thm:QFI_limit} (or more precisely \eqref{eq:fundamental_precision_limit2}) in the current setup.
Note that the RMSE decreases nearly according to the Heisenberg-limited scaling when $N_{\rm q}$ is small.
Actually, this quadratic improvement is more distinct in the case where the depolarization noise is smaller, which certainly recovers the previous result~\cite{suzuki2020amplitude} in the ideal amplitude estimation; see Fig.~\ref{apdx_fig:rmse_vs_crbound} in Appendix~\ref{apdx:additional_exp}.

Finally, Appendix~\ref{apdx:additional_exp} also provides the case of $p_{\rm q}=0.99$; that is, 1$\%$ depolarization noise is added through each operation $A$ or $A^\dagger$, 
while Fig.~\ref{fig:rmse_vs_crbound} studies the case of only 0.5$\%$ noise.
The panels (a) and (b) in Fig.~\ref{apdx_fig:rmse_vs_crbound} show that, as expected, the RMSE loses the quadratic speedup with respect to $N_{\rm q}$ immediately and approaches the precision limit (i.e., the classical scaling regarding $N_{\rm q}$).

\subsection{Efficiency of the estimator}


As described in Section~\ref{sec:QME}, the classical Fisher information \eqref{eq:unified_cFI} vanishes at certain points of the amplified level $\alpha$. 
Here we show that, in the same numerical experiment as before, our algorithm can avoid those points thanks to the optimization of $\alpha$ and as a result Eq.~\eqref{eq:qfi_vs_cfi} holds for almost all $\theta^*$.
For this purpose, we calculate the asymptotic total Fisher information $\mathcal{I}^{*}_{\rm c/q,tot}(\theta^*)$ and the asymptotic queries $N^*_{\rm q}$ with equally distributed $10^5$ points of $\cos{\theta^*}$ in $(0,1)$, based on Eq.~(\ref{apdx:alpha_star}).

\begin{figure} 
 \centering
 \includegraphics[scale=0.85]{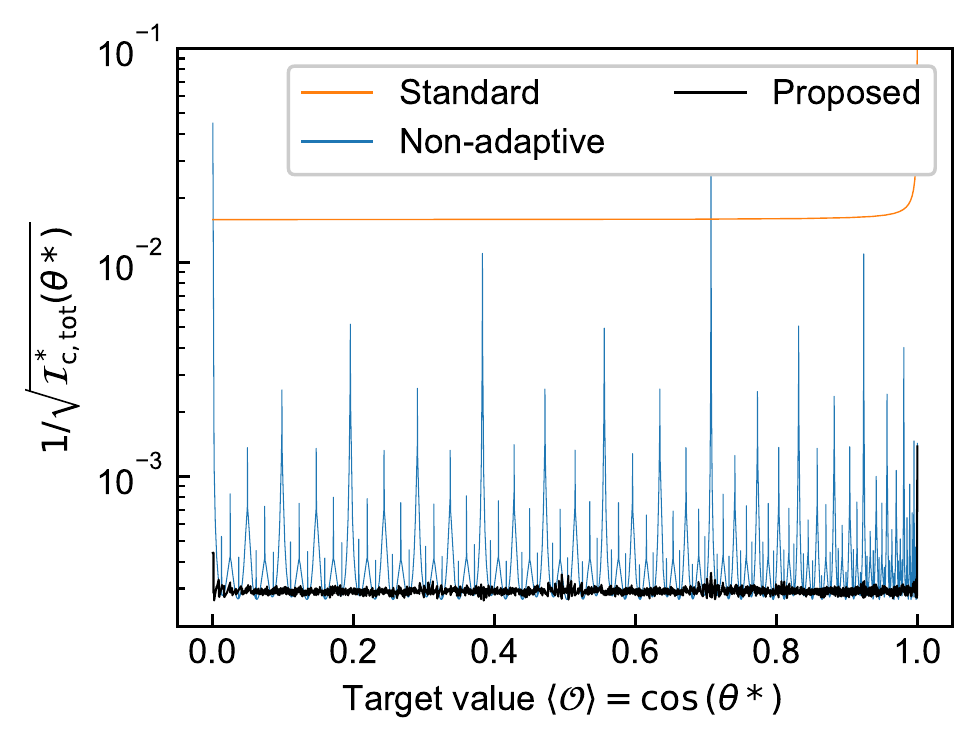}
 \caption{Comparison of the Cram\'{e}r-Rao lower bound in terms of the total (asymptotic) classical Fisher information ${\mathcal I}^*_{\rm c, tot}(\theta^*)$, of the three estimation methods for several values of $\theta^*$.
 The black (bottom), orange (top), and blue (middle) lines correspond to our method, the standard classical sampling method with $\alpha_k^*=1$, and the non-adaptive method with the predefined sequence $\alpha_1^*=1,~\alpha_k^*=2^k~(k\geq 2)$, respectively.
 Note that the non-adaptive method uses more queries than our method because $\alpha_k^*=2^k$ is the maximum number in the optimization range $D_{k-1}$; the blue line is lower than the black line at some points of $\theta^*$.}
 \label{fig:FI_dist_vs_others}
\end{figure}

Figure~\ref{fig:FI_dist_vs_others} shows the target-value dependency of the Cram\'{e}r-Rao lower bound in terms of the total (asymptotic) classical Fisher information ${\mathcal I}^*_{\rm c, tot}(\theta^*)$, of our method and the following two estimation methods.
That is, the classical Fisher information for the standard (i.e., no amplitude 
amplification) sampling method, which is usually employed for VQE~\cite{peruzzo2014variational} calculations, is defined by Eq.~(\ref{eq:asym_FI_1}) 
with $\alpha^*_k=1~\forall k$. 
Also, the classical Fisher information of the method~\cite{suzuki2020amplitude,tanaka2021amplitude,uno2021modified}, which uses amplitude 
amplification yet in a non-adaptive way, is defined by Eq.~(\ref{eq:asym_FI_1}) with 
$\alpha^*_1=1,~\alpha^*_k=2^{k}~\forall k\geq 2$, where the corresponding measurement 
is the even POVM~\eqref{eq:even_POVM}. 
Here, the number of measurement processes is chosen as $M=8$; in this case $D_{M-1}~(M\geq 8)$ contains $\alpha_{\rm B}$ under the chosen noise level.

First, the black (bottom) line reflects that our classical Fisher information 
does not almost depend on the target value, meaning the robustness of the estimator.
In addition, compared to the standard method, we observe an improvement of about two orders of magnitude for all target values in our Fisher information.
The improvement depends on the noise level, and therefore the estimation efficiency is further accelerated if the noise becomes smaller. 
In contrast to our total classical Fisher information, the performance of the non-adaptive method depicted with the blue (middle) line heavily depends on the target 
value.
Note that, since the total number of queries in the non-adaptive method is bigger 
than that in our method, there exist some target values $\cos{\theta^*}$ such that the 
total Fisher information of ours is less than that of the non-adaptive method.
On the other hand, there are some target values for which the estimation error 
bound deviates by about two orders of magnitude from ours.
Therefore, when the amplified level $\{\alpha_k\}_{k=1}^M$ is chosen non-adaptively, 
even if the asymptotic properties of maximum likelihood estimators hold, the estimation efficiency significantly decreases.

\begin{figure} 
 \centering
 \includegraphics[scale=0.85]{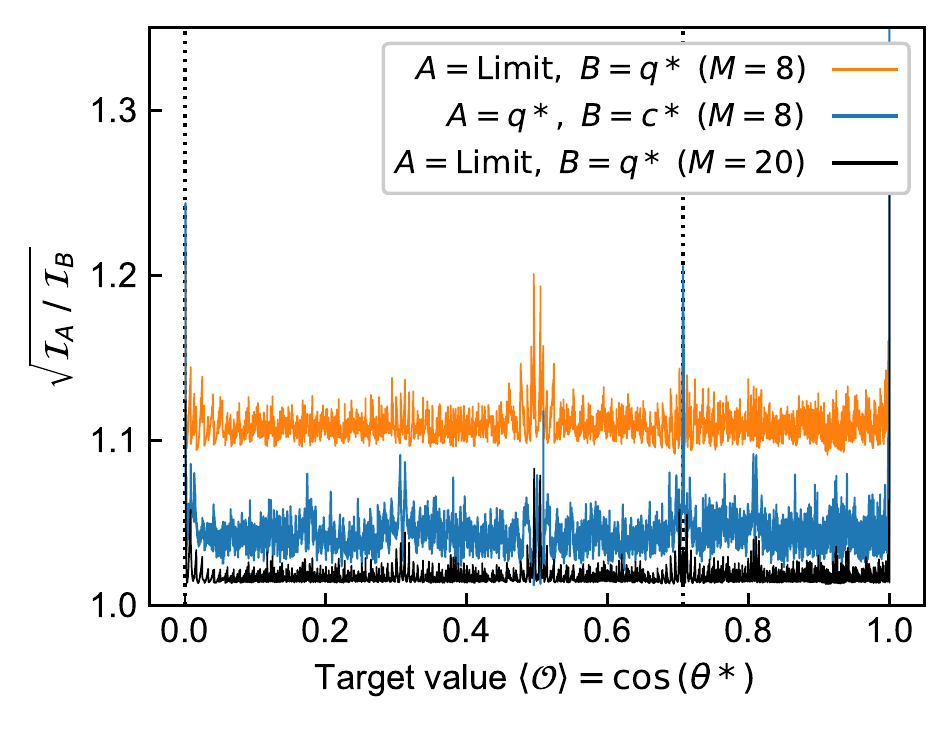}
 \caption{Comparison of the asymptotic total classical and quantum Fisher 
 information of our method in the case of 20-qubit system.
 The middle blue line represents the ratio between the quantum and classical Fisher information: $\mathcal{I}^*_{\rm q,tot}(\theta^*)$ over $\mathcal{I}^*_{\rm c,tot}(\theta^*)$ (indicated by 
 $A=q^*$ and $B=c^*$).
 On the other hand, the orange (top) and black (bottom) lines represent $N^*_{\rm q}\mathcal{I}_{\rm q}(\alpha_{\rm B})/\alpha_{\rm B}$ over $\mathcal{I}^*_{\rm q,tot}(\theta^*)$ (indicated by 
 $A={\rm Limit}$ and $B=q^*$) in the case of $M=8$ and $M=20$, respectively.
 The vertical dotted lines represent the target values corresponding to $\theta^*=\pi/j,~j=2,4$ from left to right.}
 \label{fig:FI_dist_cfivsqfi}
\end{figure}

Finally, Fig.~\ref{fig:FI_dist_cfivsqfi} shows the ratio of the asymptotic total classical and quantum Fisher information in our method.
For almost all target values, the ratio of $\mathcal{I}^*_{\rm q,tot}$ and $\mathcal{I}^*_{\rm c,tot}$ is close to 1, which is in line with the results discussed in the previous subsection.
Moreover, as expected from the objective function in Eq.~(\ref{eq:alpha_optimization}), relaxing the regularization as $\delta\to 1$, we confirm that the ratio gets closer to 1; see Fig.~\ref{fig_apdx:change_delta}, meaning 
that the first relation of Eq.~(\ref{eq:qfi_vs_cfi}) holds.
Recall now that, in general, the QCR bound (i.e., the ultimate limit of the 
estimation precision) can only be achieved when using the optimal measurement 
tailored for a given quantum state. 
Therefore, the result obtained here means that our estimation method selects the 
almost optimal measurements (or POVM) in the sense of Fisher information.
In addition, the ratio of $N_{\rm q}^*\mathcal{I}_{\rm q}(\alpha_{\rm B})/\alpha_{\rm B}$ and $\mathcal{I}^*_{\rm q,tot}$ gets closer to 1 as $M$ increases, meaning that the second relation in~\eqref{eq:qfi_vs_cfi} holds for almost all target values.
Since the right hand side of the second relation in Eq.~(\ref{eq:qfi_vs_cfi}) is equivalent to the precision limit by Theorem~\ref{thm:QFI_limit}, our algorithm therefore provides the almost best usage of $A$ and $A^\dagger$ in estimating the mean value from noisy quantum devices.

For the target values $\theta^*=\pi/2,\pi/4$, the difference of $\mathcal{I}^*_{\rm q,tot}$ and $\mathcal{I}^*_{\rm c,tot}$ is larger than that for other target values.
Note that such a difference is also distinct at $\theta^*=\pi/3,\pi/6$ when $\delta\to 1$; see Fig.~\ref{fig_apdx:change_delta}.
This is because the values of $(m+1)\theta^*$ (mod $2\pi$) obtained from the amplitude amplification are limited when $\theta^*/2\pi$ is a rational number.
In this case, $(m+1)\theta^*$ (mod $2\pi$) cannot arbitrarily get close to $\pi/2$ 
or $3\pi/2$, 
meaning that the classical Fisher information cannot arbitrarily get close to the 
quantum Fisher information; see Theorem~\ref{thm:opt_marginalPOVM} or Appendix~\ref{apdx:analysis_POVM}.
The peaks at these rational points in Fig.~\ref{fig:FI_dist_cfivsqfi} are 
very sharp, because the values of $(m+1)\theta^*$ (mod $2\pi$) fill the domain $[0,2\pi)$ exponentially fast in the optimization of amplified levels with the chosen $\{D_k\}_{k=1}^{M-1}$, when the target values shift slightly from these points.
Thus, these anomalous cases can be ignored in practice.

\section{Conclusions}\label{sec:conclusion}

We have proposed a quantum-enhanced mean value estimation method in a noisy 
environment (assumed to be the depolarization noise) that almost achieves the precision limit, when the target quantum state consists of a modest number of qubits.
Here, we derive the precision limit by evaluating the quantum Fisher information for the noisy quantum states with quantum-enhanced resolution regarding the target mean value, in the setup without demanding controlled amplifications.
This method employs a modified maximum likelihood estimation consisting of the amplitude amplification and the adaptive measurement;
the latter is derived from the ideal POVM achieving the quantum Fisher information, 
and notably, it can be implemented on standard quantum computing devices without any knowledge of the state preparation oracle for the target state.
The measurement and the number of queries of the amplification operators are adaptively optimized for enhancing the classical Fisher information toward achieving the quantum Fisher information.
Importantly, thanks to the maximum likelihood formulation and the two types of measurements 
with different symmetry in the probability distribution, our estimator enjoys some provable 
statistical properties such as consistency and asymptotic normality.

To show the effectiveness of the proposed estimator, we executed several numerical 
experiments.
The asymptotic statistical properties are evaluated in terms of the root mean 
squared error (RMSE), for several target values; the result was that, in all 
cases, the RMSE saturates the asymptotic classical and quantum Cram\'{e}r-Rao 
lower bound, with a modest number of measurements. 
In particular, we confirmed that the classical Fisher information almost saturates 
the ultimate quantum Fisher information in a large system dimension $d=2^{20}$ (corresponding to a twenty-qubit system).
In addition, we studied how the total Fisher information depends on the target 
value by evaluating the asymptotic classical/quantum Cram\'{e}r-Rao lower bounds. 
Although the previous researches~\cite{wang2021minimizing,tanaka2021amplitude} imply 
that the classical Fisher information significantly deteriorates for certain target 
values under depolarization noise, the proposed estimator retains a large Fisher 
information regardless of the target value, due to the adaptive optimization.

We believe that the estimation method presented in this paper paves the way for an 
interdisciplinary research in quantum computing and quantum sensing; actually a few 
such trials have been found in the literature~\cite{PhysRevA.91.062322,PhysRevA.93.040304}.
Moreover, it may be useful in a wide field of quantum information technologies beyond 
the subroutine in quantum computing algorithms.

\mbox{}
\\
{\bf Acknowledgements:} 
K.W. and K.F. thank IPA for the support through MITOU target program. 
We would like to thank Dr.~Yasunari Suzuki, Dr.~Yuuki Tokunaga, and Dr.~Tomoki Tanaka for helpful discussions.
K.W. was supported by JST SPRING, Grant Number JPMJSP2123.
This work was supported by MEXT Quantum Leap Flagship Program Grant Number JPMXS0118067285 and JPMXS0120319794.

\bibliographystyle{quantum}
\bibliography{references}


\appendix
















\renewcommand{\theequation}{A.\arabic{equation}}
\setcounter{equation}{0}

\section{Fisher information and Cram\'{e}r-Rao inequality}\label{apdx:A}

Let $\bm{X}$ be a multi-dimensional discrete random variable following a joint 
probability distribution $\mathcal{L}(\bm{x};\theta)$, where $\theta$ is 
an unknown parameter in $\mathbb{R}$ (in our case, the domain is $(0,\pi)$). 
Our goal is to estimate $\theta$ or generally a function of $\theta$, say 
$g(\theta)$ with $g(\cdot)$ a bijective function, by constructing an estimator 
$\hat{e}(\bm{X})$ for $g(\theta)$.
In general, if $\hat{e}(\bm{X})$ is an unbiased estimator, i.e., $\mathbf{E}_{\bm{X}}\left[\hat{e}(\bm{X})\right]=g(\theta)$, then the mean squared 
error of $\hat{e}(\bm{X})$ is bounded as 
\begin{equation}
\label{apdx:cl-cr}
    {\rm MSE}\left[\hat{e}(\bm{X})\right]\geq \left({\frac{\partial g(\theta)}{\partial\theta}}\right)^2\frac{1}{\mathcal{I}_{\rm c,tot}(\theta)}.
\end{equation}
This is called the Cram\'{e}r-Rao inequality~\cite{shao2003mathematical}. 
Also, $\mathcal{I}_{\rm c,tot}(\theta)$ is the total classical Fisher information 
defined as 
\begin{align}\label{apdx:total_Fisher}
    \mathcal{I}_{\rm c,tot}(\theta)&:=\mathbf{E}_{\bm{X}}\left[\left\{\frac{\partial}{\partial\theta}\ln \mathcal{L}(\theta;\bm{X})\right\}^2\right].
\end{align}
For simplicity, we write $\partial/\partial\theta$ as $\partial_\theta$ 
in what follows.

Next, let us consider the estimation problem for a single parameter embedded 
in a quantum state $\rho(\theta)$. 
Fixing a POVM for measuring the quantum state, we can obtain the probability 
distribution parameterized by $\theta$ and study Eq.~\eqref{apdx:cl-cr}. 
However, there is a freedom for designing a POVM in the quantum case. 
Utilizing this freedom, the lower bound of Eq.~(\ref{apdx:cl-cr}) can in fact be 
improved; that is, the following quantum Cram\'{e}r-Rao inequality holds (for simplicity, $g(\theta)=\theta$)~\cite{helstrom1969quantum,helstrom1968minimum,hayashi2006quantum,braunstein1994statistical,paris2009quantum}: 
\begin{align}
\label{apdx:q-cr}
    {\rm MSE}\left[\hat{e}(\bm{X})\right]\geq \frac{1}{\mathcal{I}_{\rm c}(\theta)}\geq \frac{1}{\mathcal{I}_{\rm q}(\theta)},
\end{align}
where $\mathcal{I}_{\rm c}$ is the classical Fisher information associated with the 
probability distribution for a POVM $\{M_k\}$ and the target state $\rho(\theta)$.
Here, $\mathcal{I}_{\rm q}$ denotes the quantum Fisher information defined by only the 
quantum state $\rho(\theta)$ as follows:
\begin{equation}
    \mathcal{I}_{\rm q}(\theta):={\rm tr}\left[L_{\rm SLD}^2\rho(\theta)\right],
\end{equation}
where $L_{\rm SLD}$ is the symmetric logarithmic derivative (SLD), which is defined as 
the Hermitian operator satisfying 
\begin{align}
      \partial_\theta \rho(\theta)
         =\frac{1}{2}\left(L_{\rm SLD}\rho(\theta)+\rho(\theta)L_{\rm SLD}\right).
\end{align}
Note that the second inequality Eq.~(\ref{apdx:q-cr}) holds for any POVM $\{M_k\}$ 
independent to $\theta$.
Therefore, the quantum Fisher information characterizes the fundamental limit of 
estimation precision that cannot improve via any measurement.

\renewcommand{\theequation}{B.\arabic{equation}}
\setcounter{equation}{0}
\setcounter{thm}{0}

\section{Proof of Theorem~\ref{thm:QFI_limit}}\label{apdx:proof_upb}
Here, we provide a proof of Theorem~\ref{thm:QFI_limit}. 
For convenience, we recall it.
\begin{thm}
    For a given number $N_{\rm q}$ of queries to $n$-qubit state preparation oracles $A$ and $A^\dagger$, we consider all partitions $\{\alpha'_k\}~(\alpha'_k\in\mathbb{N})$ of $N_{\rm q}$ such that $N_{\rm q}=\sum_{k=1}^{M'}\alpha'_k$ for some positive integer $M'$.
    Here, each use of $A$ and $A^\dagger$ induces the $n$-qubit depolarization noise with probability $1-p_{\rm q}$.
    Then, the total quantum Fisher information $\mathcal{I}_{\rm q,tot}$ regarding $\theta^*$ of the quantum state $\rho_{\mathcal{D}}^{{\rm (bm)}}(\theta^*):=\rho_{\mathcal{D}}(\alpha'_1;\theta^*)\otimes \cdots\otimes \rho_{\mathcal{D}}(\alpha'_{M'};\theta^*)$
    satisfies the following inequalities:
    $$
    \mathcal{I}_{\rm q,tot}\left[\rho_{\mathcal{D}}^{(\rm bm)}\right] \leq \frac{N_{\rm q}^2 p_{\rm q}^{2N_{\rm q}}}{2^{1-n}+\left(1-2^{1-n}\right)p_{\rm q}^{N_{\rm q}}},~N_{\rm q}\leq \alpha_{\rm B},
    $$
    and
    $$
    \mathcal{I}_{\rm q,tot}\left[\rho_{\mathcal{D}}^{(\rm bm)}\right] \leq \frac{N_{\rm q}\alpha_{\rm B}p_{\rm q}^{2\alpha_{\rm B}}}{2^{1-n}+\left(1-2^{1-n}\right)p_{\rm q}^{\alpha_{\rm B}}},~N_{\rm q} > \alpha_{\rm B},
    $$
    where $\alpha_{\rm B}$ is defined by 
    $$\alpha_{\rm B}:=\underset{\alpha\in\mathbb{N}}{\rm argmax}~\frac{\alpha p_{\rm q}^{2\alpha}}{{2^{1-n}}+\left(1-2^{1-n}\right)p_{\rm q}^{\alpha}}.$$
    Furthermore, if $N_{\rm q}\leq \alpha_{\rm B}$ or $N_{\rm q}=r\alpha_{\rm B}$ for some $r\in\mathbb{N}$, there exists a partition $\{\alpha'_k\}$ satisfying the above equality.
\end{thm}

\begin{proof}
    First, for any parition $\{\alpha'_k\}$, $\mathcal{I}_{\rm q,tot}[\rho_{\mathcal{D}}^{\rm (bm)}]$ is given by the sum of each quantum Fisher information for $\rho_{\mathcal{D}}(\alpha'_k;\theta^*)$ with respect to $\theta^*$:
    \begin{align}
        \mathcal{I}_{\rm q,tot}\left[\rho_{\mathcal{D}}^{\rm (bm)}\right]&= \sum_{k=1}^{M'} \frac{(\alpha'_k)^2 p_{\rm q}^{2\alpha'_k}}{{\frac{2}{2^n}}+\left(1-\frac{2}{2^n}\right)p_{\rm q}^{\alpha'_k}}\notag\\
        &\equiv \sum_{k=1}^{M'} \mathcal{I}_{\rm q}(\alpha'_k).    
    \end{align}
    Considering $$\mathcal{I}_{\rm q}(\alpha)={\alpha^2 p_{\rm q}^{2\alpha}}/({{2^{1-n}}+\left(1-{2^{1-n}}\right)p_{\rm q}^{\alpha}}),$$
    we can directly confirm that $\mathcal{I}_{\rm q}(\alpha)/\alpha$ is an increasing function in the regime $0<\alpha\leq \alpha_{\rm B}$, from the definition of $\alpha_{\rm B}.$
    Note that if the argmax in the definition of $\alpha_{\rm B}$ returns multiple values, we take the minimum as $\alpha_{\rm B}$.
    In the case of $N_{\rm q}\leq \alpha_{\rm B}$, for any partition $\{\alpha'_k\}$, $\alpha'_k\leq N_{\rm q}\leq \alpha_{\rm B}$ holds, and therefore we have 
    \begin{align}
        \sum_{k=1}^{M'} {\mathcal{I}_{\rm q}(\alpha'_k)}&=\sum_{k=1}^{M'} \alpha'_{k}\frac{\mathcal{I}_{\rm q}(\alpha'_k)}{\alpha'_k}\notag\\
        &\leq \sum_{k=1}^{M'} \alpha'_{k}\frac{\mathcal{I}_{\rm q}(N_{\rm q})}{N_{\rm q}}= {\mathcal{I}_{\rm q}(N_{\rm q})},
    \end{align}
    where we recall $N_{\rm q}=\sum_{k=1}^{M'} \alpha_k'$.
    The second equality holds for the following partition: $M'=1$ and $\alpha'_1=N_{\rm q}$.

    Next, we consider the second case $N_{\rm q}>\alpha_{\rm B}$. 
    By definition, $\mathcal{I}_{\rm q}(\alpha)/\alpha\leq  \mathcal{I}_{\rm q}(\alpha_{\rm B})/\alpha_{\rm B}$ holds for any $\alpha\in\mathbb{N}$, and this leads to
    \begin{align}
        \sum_{k=1}^{M'} \mathcal{I}_{\rm q}(\alpha'_k)\leq \sum_{k=1}^{M'} \frac{\alpha'_k}{\alpha_{\rm B}}\mathcal{I}_{\rm q}(\alpha_{\rm B})=\frac{\mathcal{I}_{\rm q}(\alpha_{\rm B})}{\alpha_{\rm B}}N_{\rm q}.
    \end{align}
    If there exists a natural number $r$ such that $N_{\rm q}=r\alpha_{\rm B}$, then the first equality is saturable for the partition: $M'=r,~\alpha'_k=\alpha_{\rm B}$ for all $k$. 
    This completes the proof of Theorem~\ref{thm:QFI_limit}.
\end{proof}

\renewcommand{\theequation}{C.\arabic{equation}}
\setcounter{equation}{0}

\section{Analysis for the POVMs}\label{apdx:analysis_POVM}

Our problem is to estimate the target parameter $\theta\in(0,\pi)$ embedded into the noisy quantum state~\eqref{eq:amplified_state_in_S_noise}.
Here, let us consider the following POVM:
\begin{align}
\label{apdx:qfi_POVM}
    M_0&:=\ket{{0}}_n\bra{{0}},~M_1:=A^\dagger\ket{\bar{1}}\bra{\bar{1}}A,\notag \\[4pt]
    M_2&:=I_n-M_0-M_1.
\end{align}
%
The POVM has the following meanings; $M_0$ and $M_1$ correspond to the measurement in 
the subspace basis $\ket{{0}}_n$ and $A^\dagger\ket{\bar{1}}$, respectively, and $M_2$ 
corresponds to the event such that the measured state is not in the subspace 
${\rm Span}\left\{\ket{{0}}_n,A^\dagger\ket{\bar{1}}\right\}$.
Now, the classical Fisher information $\mathcal{I}'_{\rm c}(\theta)$ for the probability distribution
\begin{align}
    \left\{{\rm tr}[M_k\rho_{\mathcal{D}}(\alpha;\theta)]\right\}_{k=0}^2,~~\alpha=2m+2
\end{align}
is calculated as follows:
\begin{widetext}
\begin{align}
    \mathcal{I}'_{\rm c}(\theta)&={4(\eta')^2\left(m+1\right)^2\sin^2{\left[2(m+1)\theta\right]}}\left\{{\eta'+\frac{2(1-\eta')}{d}-\frac{d(\eta')^2\cos^2{\left[2(m+1)\theta\right]}}{d\eta'+2(1-\eta')}}\right\}^{-1},~~~d=2^n\notag\\[4pt]
    &={4(\eta')^2\left(m+1\right)^2 \frac{d\eta'+2(1-\eta')}{d(\eta')^2}}\sin^2{\left[2(m+1)\theta\right]} \left[\left({1+\frac{2(1-\eta')}{d\eta'}}\right)^2-1+\sin^2{\left[2(m+1)\theta\right]}\right]^{-1}\notag\\[4pt]
    &\leq{4(\eta')^2\left(m+1\right)^2 \frac{d\eta'+2(1-\eta')}{d(\eta')^2}}\left({1+\frac{2(1-\eta')}{d\eta'}}\right)^{-2}=\frac{d(\eta')^2(2m+2)^2}{d\eta'+{2(1-\eta')}},
\end{align}
\end{widetext}
where $\eta':=p_{\rm q}^{\alpha}=p_{\rm q}^{2m+2}$.
Since a function $x^2/(K+x^2),~x\in[-1,1]$ with a positive constant $K$ is maximized at $x=\pm 1$, we obtain the third line, and therefore the condition for the equality is given by $\sin{[2(m+1)\theta]}=\pm 1$.
The most right hand side in the final line corresponds to the quantum Fisher information $\mathcal{I}_{\rm q}(\alpha)$ for $\alpha=2m+2$.
Thus, for $\theta$ satisfying the equality condition, the POVM (\ref{apdx:qfi_POVM}) is 
optimal in the sense that the corresponding classical Fisher information gives the quantum Fisher information. 
Moreover, in the limit of $d=2^n\to \infty$, the classical Fisher information 
is equal to the quantum Fisher information except for $(m, \theta)$ satisfying 
$\sin[2(m+1)\theta]=0$.
Note that, in general, the eigenstates of SLD operator can be used to construct the optimal measurement that exactly achieves the quantum Fisher information~\cite{braunstein1994statistical}.
In our case, the POVM of optimal measurement consists of quantum states in the form of superposition of $\ket{{0}}_n$ and $A^\dagger\ket{\bar{1}}$, where the coefficients of the states depend on the unknown target value; the POVM~\eqref{apdx:qfi_POVM} is obtained by removing this target-dependency in the coefficients.

Although the POVM (\ref{apdx:qfi_POVM}) is an optimal measurement for 
certain $\theta$ satisfying $\sin{[2(m+1)\theta]}=\pm 1$, its implementation is nontrivial especially for $M_1$ because $A^\dagger\ket{\bar{1}}$ is unknown.
Note that the optimal POVM obtained from the SLD operator has the same difficulty in implementation. 
Therefore, in this paper, we focus on the 2-valued POVM (\ref{eq:even_POVM}), which can be obtained from the POVM (\ref{apdx:qfi_POVM}) as follows:
\begin{align}
    M_0^{(\rm even)}= M_0 ,~M_1^{(\rm even)}=M_1+M_2.
\end{align}
This POVM removes the element $M_1$ from the original 3-valued POVM (\ref{apdx:qfi_POVM}), 
meaning that the quantum Fisher information is not achieved in general.

The classical Fisher information associated with the 2-valued POVM is calculated as, assuming $\cos{[(m+1)\theta]}\neq0$, 
\begin{widetext}
\begin{align}\label{apdx:trade-off_evenFI}
    \mathcal{I}^{(\rm even)}_{\rm c}(m;\theta)
    &=\frac{(2m+2)^2 \eta'\sin^2{\left[(m+1)\theta\right]}}{1-\eta'\cos^2{\left[(m+1)\theta\right]}+\frac{1-\eta'}{d\eta'\cos^2{\left[(m+1)\theta\right]}}\left(1-\frac{1-\eta'}{d}-2\eta'\cos^2{\left[(m+1)\theta\right]}\right)}\notag\\[4pt]
    &=\mathcal{I}^{(\rm even)}_{\rm q}(m)\frac{\sin^2{\left[(m+1)\theta\right]}\left(1+2\frac{1-\eta'}{d\eta'}\right)}{1-\eta'\cos^2{\left[(m+1)\theta\right]}+\frac{1-\eta'}{d\eta'\cos^2{\left[(m+1)\theta\right]}}\left(1-\frac{1-\eta'}{d}-2\eta'\cos^2{\left[(m+1)\theta\right]}\right)}\notag\\[4pt]
    &\equiv\mathcal{I}^{(\rm even)}_{\rm q}(m)\frac{\sin^2{\left[(m+1)\theta\right]}}{1-\eta'\cos^2{\left[(m+1)\theta\right]}+\varepsilon_n(m;\theta)}\left(1+2\frac{1-\eta'}{d\eta'}\right),
\end{align}
where we recall $\eta':=p_{\rm q}^{\alpha}=p_{\rm q}^{2m+2}$ and in the final line, we defined $\varepsilon_{n}$ as
\begin{equation}
    \varepsilon_n(m;\theta) := \frac{1-\eta'}{d\eta'\cos^2{\left[(m+1)\theta\right]}}\left(1-\frac{1-\eta'}{d}-2\eta'\cos^2{\left[(m+1)\theta\right]}\right).
\end{equation}
\end{widetext}
This factor $\varepsilon_{n}$ vanishes exponentially fast with respect to the number of qubits $n$ if $\cos{[(m+1)\theta]}\neq0$, as follows:
\begin{equation}
    \left|\varepsilon_n(m;\theta)\right| \leq\frac{1-\eta'}{\eta'\cos^2{\left[(m+1)\theta\right]}}\frac{1}{2^{n}}=O\left(\frac{1}{2^n}\right),
\end{equation}
where we recall $d=2^n$.
Here, we define the coefficient of $\mathcal{I}^{(\rm even)}_{\rm q}(m)$ as $\kappa$:
\begin{align}
    \kappa := \frac{\sin^2{[(m+1)\theta]}}{1-\eta'\cos^2{[(m+1)\theta]}+\varepsilon_n(m;\theta)}
\end{align}
Then, Eq.~(\ref{apdx:trade-off_evenFI}) is expressed as
\begin{align}\label{apdx:even_FI_infty}
    \mathcal{I}^{(\rm even)}_{\rm c}(m;\theta)=\kappa \mathcal{I}^{(\rm even)}_{\rm q}(m)\left[1+2\frac{1-\eta'}{2^n\eta'}\right].
\end{align}
Now, when $\theta/2\pi$ is an irrational number, the set of complex numbers 
$\{ e^{i(m+1)\theta} \}_{m=0,1,\cdots}$ with sufficiently large $m$ fill the 
unit circle in the complex plane. 
Thus, there exists an $m$ such that $\cos{[(m+1)\theta]}$ is arbitrarily close 
to 0 and accordingly 
\begin{equation}
    \kappa_{\infty}:=\lim_{d\to \infty} \kappa=\frac{\sin^2{[(m+1)\theta]}}{1-\eta'\cos^2{[(m+1)\theta]}}\in[0,1)
\end{equation}
is arbitrarily close to 1.
Therefore, from Eq.~(\ref{apdx:even_FI_infty}), the classical Fisher information 
sufficiently approaches the quantum Fisher information, under certain condition 
on $(m, \theta)$; moreover, taking an appropriate $m$, we can make the convergence 
exponentially fast with respect to the number of qubits (recall $d=2^n$ with $n$ 
the number of qubits).
This establishes Theorem~\ref{thm:opt_marginalPOVM}.

Here we point out that there is discontinuity in the classical Fisher information, 
which implies the instability of estimation. 
For this purpose, let us see the Fisher information in the limit $d\to \infty$:
\begin{widetext}
\begin{align}\label{apdx:asymptotic_behavior}
    \lim_{d\to \infty}\mathcal{I}^{(\rm even)}_{\rm c}(m;\theta)=\begin{dcases}
    0,&\cos{[(m+1)\theta]}=0\\
    \frac{\mathcal{I}^{(\rm even)}_{\rm q}(m)\sin^2{[(m+1)\theta]}}{1-\eta'\cos^2{[(m+1)\theta]}},&\cos{[(m+1)\theta]}\neq0
    \end{dcases}.
\end{align}
That is, for $\cos{[(m+1)\theta]}\neq0$, we obtain
\begin{align}
    \lim_{d\to \infty}\mathcal{I}^{(\rm even)}_{\rm c}(m;\theta)
    &\leq 4\eta'(m+1)^2=\lim_{d\to\infty}\mathcal{I}^{(\rm even)}_{\rm q}(m).
\end{align}
The first equality holds for $\cos{[(m+1)\theta]}=0$ even though $\mathcal{I}^{(\rm even)}_{\rm c}(m;\theta)$ vanishes at such points as shown in Eq.~(\ref{apdx:asymptotic_behavior}).
Therefore, the classical Fisher information is not continuous at $\cos{[(m+1)\theta]}=0$ in the limit $d=2^n\to\infty$.

Finally, as for the other 2-valued POVM (\ref{eq:odd_POVM}), comparing the Fisher 
information $\mathcal{I}_{\rm q/c}(\alpha;\theta)$ for $\alpha=2m+1$, we can obtain 
the following inequality for $d>2$:
\begin{align}
    \mathcal{I}^{(\rm odd)}_{\rm c}(m;\theta):=\frac{(2m+1)^2(p^{2m+1}_{\rm q})^2\sin^2{[(2m+1)\theta]}}{4\mathbf{P}^{(\rm odd)}_{\mathcal{D}}(m;\theta)\left(1-\mathbf{P}^{(\rm odd)}_{\mathcal{D}}(m;\theta)\right)}
    \leq {(2m+1)^2p^{2(2m+1)}_{\rm q}}<\mathcal{I}_{\rm q}(2m+1).
\end{align}
Hence, unlike the even classical Fisher information, the odd classical 
Fisher information is always smaller than the quantum Fisher information 
even when the number of qubits increases.

\renewcommand{\theequation}{D.\arabic{equation}}
\setcounter{equation}{0}
\section{Statistical structure of the algorithm}\label{apdx:stat_structure}

In our model, the distribution of random variables $\bm{X}_{M}$ has a hierarchical 
model as follows; 
\begin{equation}\label{apdx:total_likelihood_fn}
    \mathcal{L}_{M}\left(\bm{x}_M;\theta\right)=\prod_{m=1}^M F_m\left(x^{(m)};\alpha_m(\bm{x}_{m-1}),\theta\right),
\end{equation}
where $F_m$ is the probability function of Binomial distribution with success probability $\mathbf{P}_{\mathcal{D}}({\alpha}_m(\bm{x}_{m-1});\theta)$. 
Substituting Eq.~(\ref{apdx:total_likelihood_fn}) into Eq.~(\ref{apdx:total_Fisher}), we 
obtain the total classical Fisher information as
%
\begin{align}\label{apdx:expansion_totFI}
    \mathcal{I}_{\rm c,tot}(\theta)&=-\sum_{m=1}^M\mathbf{E}_{\bm{X}_m}\left[\partial_\theta^2 \ln F_m\left(X^{(m)};{\alpha}_m(\bm{X}_{m-1}),\theta\right)\right]\notag\\
    &=-\sum_{m=1}^M\sum_{\bm{x}_m=(x^{(1)},\cdots,x^{(m)})}\partial_\theta^2\ln F_m\left(x^{(m)};\alpha_m(\bm{x}_{m-1}),\theta\right)\mathcal{L}_m(\bm{x}_m;\theta)\notag\\
    &=-\sum_{m=1}^M\sum_{{x}^{(m)}}\sum_{\bm{x}_{m-1}=(x^{(1)},\cdots,x^{(m-1)})}\partial_\theta^2\ln F_m\left(x^{(m)};\alpha_m(\bm{x}_{m-1}),\theta\right)F_m\left(x^{(m)};\alpha_m(\bm{x}_{m-1}),\theta\right)\mathcal{L}_{m-1}(\bm{x}_{m-1};\theta)\notag\\
    &=-\sum_{m=1}^M\sum_{\bm{x}_{m-1}}\mathcal{L}_{m-1}(\bm{x}_{m-1};\theta)\sum_{{x}^{(m)}}\partial_\theta^2\ln F_m\left(x^{(m)};\alpha_m(\bm{x}_{m-1}),\theta\right)F_m\left(x^{(m)};\alpha_m(\bm{x}_{m-1}),\theta\right)\notag\\
    &=-\sum_{m=1}^M\sum_{\bm{x}_{m-1}}\mathcal{L}_{m-1}(\bm{x}_{m-1};\theta) \mathbf{E}_{Y^{(m)}}\left[\partial_\theta^2\ln F_m\left(Y^{(m)};\alpha_m(\bm{x}_{m-1}),\theta\right)\right],
\end{align}
\end{widetext}
where we used the conditional probability (\ref{apdx:total_likelihood_fn}) in the third line.
$Y^{(m)}$ denotes a random variable following $F_m(y^{(m)};\alpha_m(\bm{x}_{m-1}),\theta)$, 
with $y^{(m)}\in\{0,1,\cdots,N\}$. 
Note here that $\bm{x}_{m-1}$ is not a random variable, but a realized value of 
$\bm{X}_{m-1}$.
The expectation in the final line corresponds to the classical Fisher information in Eq.~(\ref{eq:unified_cFI}) multiplied by $N$. 
Thus, we can write the total classical Fisher information as
\begin{align}
    \mathcal{I}_{\rm c,tot}\left(\theta\right)=N\sum_{m=1}^M \mathbf{E}_{\bm{X}_{m-1}}\left[\mathcal{I}_{\rm c}\left({\alpha}_m(\bm{X}_{m-1});\theta\right)\right].
\end{align}
Similarly, we define the corresponding total quantum Fisher information as
\begin{align}
    \mathcal{I}_{\rm q,tot}\left(\theta\right)=N\sum_{m=1}^M \mathbf{E}_{\bm{X}_{m-1}}\left[\mathcal{I}_{\rm q}\left({\alpha}_m(\bm{X}_{m-1})\right)\right].
\end{align}

As proved in the following subsection, ${\alpha}_m(\bm{X}_{m-1})=\alpha^*_m$ holds with high probability for all $m$ as $N$ increases, where $\alpha_m^*$ is a solution of the following optimization problem for the target value $\theta$ (not an estimate $\hat{\theta}$):
\begin{equation}\label{apdx:alpha_star}
    \alpha_{m}^*:=\underset{\alpha\in D_{m-1}}{\rm argmax}~\frac{\mathcal{I}_{\rm c}(\alpha;\theta)}{\alpha}\frac{\sin^2{\left(\alpha{\theta}\right)}}{1-\delta\cos^2{\left(\alpha{\theta}\right)}},
\end{equation}
where we define $\alpha_1^*=1$.
Using Eq.~\eqref{apdx:alpha_star}, the asymptotic values of the total classical/quantum 
Fisher information are given by
\begin{align}
    {\mathcal{I}^*_{\rm c/q,tot}\left(\theta\right)}:=N\sum_{m=1}^M \mathcal{I}_{\rm c/q}\left({\alpha}^*_m;\theta\right).
\end{align}
Then we have 
\begin{align}
    \lim_{N\to \infty}\frac{{\mathcal{I}_{\rm c/q,tot}\left(\theta\right)}}{N}=\sum_{m=1}^M \mathcal{I}_{\rm c/q}\left({\alpha}^*_m;\theta\right)=\frac{{\mathcal{I}^*_{\rm c/q,tot}\left(\theta\right)}}{N},
\end{align}
which characterizes the variance of our estimator in the asymptotic regime.

Finally, we show the impact of $\delta$ on the asymptotic total classical/quantum Fisher information in Fig.~\ref{fig_apdx:change_delta}.
Here, the experimental settings are equal to those in Section~\ref{sec:results}.
As expected from the form of the objective function~\eqref{apdx:alpha_star}, we confirm that the ratio of ${\mathcal{I}^*_{\rm c/q,tot}\left(\theta\right)}$ gets closer to 1 when $\delta\to 1$.
Also, in the case of $\delta=0.99$, the peaks at the rational points $\pi/j,~(j=2,3,4,6)$ become more distinct than the case of $\delta=0.95$.

\begin{figure} 
 \centering
 \includegraphics[scale=0.85]{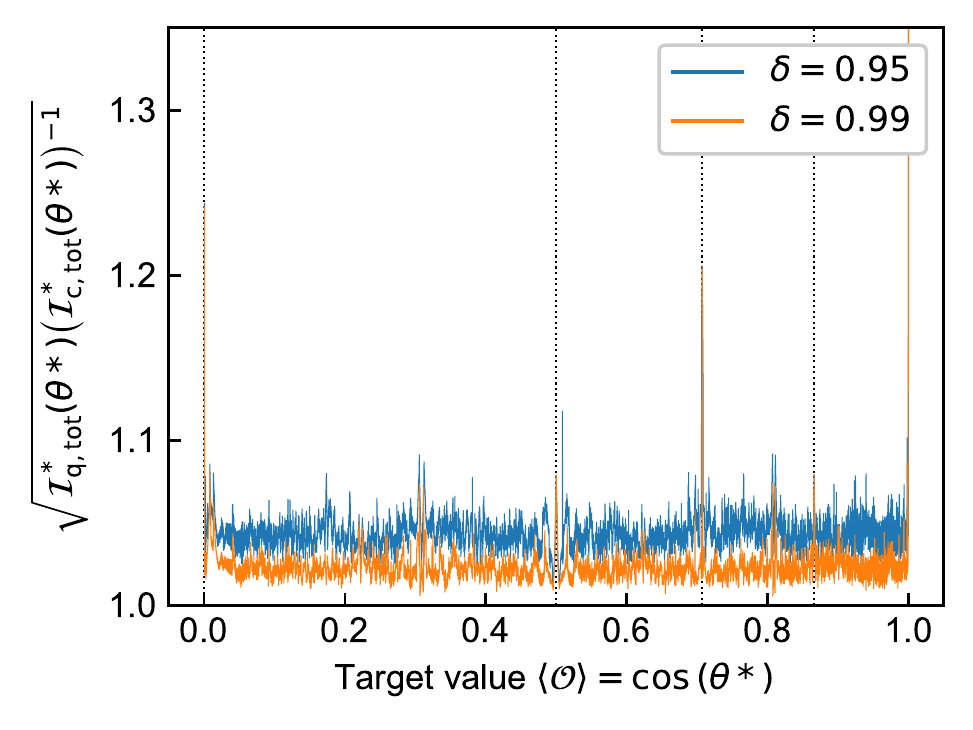}
 \caption{Comparison of the asymptotic total classical and quantum Fisher information of our method in the case of 20-qubit system and $M=8$. 
 The blue (top) and orange (bottom) lines represent the results for the case of the regularization parameter $\delta=0.95$ and $\delta=0.99$, respectively.
 The vertical dotted lines represent the target values corresponding to $\theta^*=\pi/j,~j=2,3,4,6$ from left to right.}
 \label{fig_apdx:change_delta}
\end{figure}

\renewcommand{\theequation}{E.\arabic{equation}}
\setcounter{equation}{0}
\setcounter{thm}{2}
\section{Proof for statistical properties}\label{apdx:stat_prop}

Let $\theta^*\in(0,\pi)$ satisfying $\cos{\theta^*}=\langle\mathcal{O}\rangle$ be a target value to be estimated.
Since the objective function in Eq.~(\ref{apdx:alpha_star}) vanishes at $\alpha\theta^*= l\pi~(l\in\mathbb{Z})$, then the definition of $\alpha_m^*$ leads to $\alpha_m^*\theta^*\neq l\pi$ for $m=2,3,\cdots,M$. 
Thus, we can take an integer $z_m$ satisfying
\begin{equation}
  \theta^*\in\Theta^{(m)}:=\left(\frac{z_m}{\alpha^*_m}\pi,\frac{z_m+1}{\alpha^*_m}\pi\right),~~~\forall m,
\end{equation}
and define $\Theta_M:=\cap_{m=1}^M\Theta^{(m)}$.

\begin{thm}
  There exists a unique maximum likelihood estimator $\hat{\theta}_M$ of $\mathcal{L}_M(\theta;\boldsymbol{X}_M)$ such that $\cos{\hat{\theta}_M}\to \cos{\theta^*}$ (convergence in probability) with the probability approaching 1 as $N\to\infty$, and 
  \begin{align}
      \lim_{N\to \infty}P\left[\alpha_k(\bm{X}_{k-1})=\alpha_k^*\right]= 1
  \end{align}
  holds for all $k=2,\cdots,M$.
\end{thm}

As a demonstration of Theorem~\ref{thm:consistency}, we provide an example of the realized amplified levels in Fig.~\ref{apdx_fig:alpha_k}.
Here, the experimental settings are equivalent to those in Section~\ref{sec:confirm_asym_prop}.
Figure~\ref{apdx_fig:alpha_k} shows that the amplified levels converge rapidly to the asymptotic values, and therefore the total Fisher information and the total number of queries also show rapid convergence.

\begin{figure*}[tb]
 \centering
 \begin{tabular}{ccc}
 \includegraphics[scale=0.75]{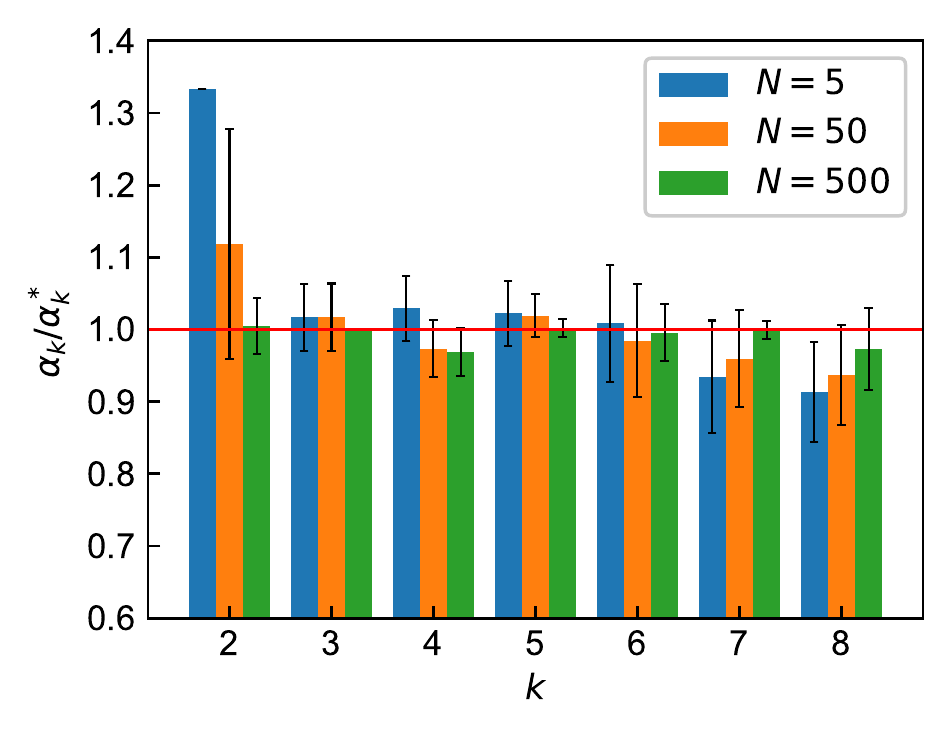}&&\includegraphics[scale=0.75]{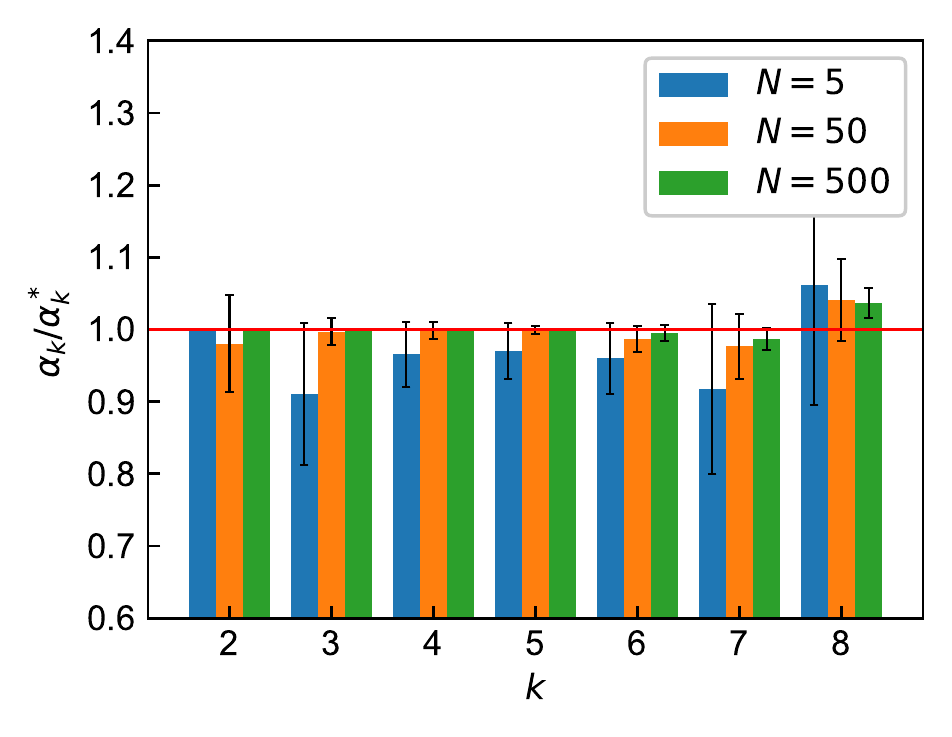}\\
 (a) $\cos\theta^*=0.042$&&(b) $\cos\theta^*=0.5$\\
 &&\\
 \end{tabular}
 \caption{Convergence of the amplified levels $\alpha_k$ with respect to $N$.
 The color bars represent the average of the realized amplified levels normalized by its asymptotic values $\alpha_k^*$ over 3000 trials, and the error bars denote the standard deviation.
 With a fixed value of $k$, the three color bars show the results corresponding to 
 $N=5~({\rm left,blue}),~N=50~({\rm center,orange}),~{\rm and}~N=500~({\rm right,green})$, 
 respectively.
 }
 \label{apdx_fig:alpha_k}
\end{figure*}

\begin{proof}
  Let $E(m)$ be the following event : $\alpha_k(\bm{X}_{k-1})=\alpha_k^*$ holds for all $k=1,2,\cdots,m$, and there exists a unique maximum likelihood estimator $\hat{\theta}_m$ of $\mathcal{L}_m(\theta;\bm{X}_{m})$ which converges to $\theta^*$.
  Note that if we establish $\hat{\theta}_m\to \theta^*$, then this immediately leads to $\cos{\hat{\theta}_m}\to \cos{\theta^*}$ due to the continuous mapping theorem. 
  
  In the case of $m=1$, the likelihood function $\mathcal{L}_1(\theta;\bm{X}_1)$ is equivalent to $F_1(X^{(1)};\alpha_1^*,\theta)$.
  Because the success probability $\mathbf{P}_{\mathcal{D}}(\alpha_1^*;\cdot) : (0,\pi)\to (0,1)$ is injective and the likelihood equation
  \begin{align}
  \partial_{\theta}\ln F_1(X^{(1)};\alpha_1^*,\theta)=0
  \end{align}
  has a unique solution when a solution exists, it follows that
  \begin{equation}\label{apdx:convergence_P(E1)}
    P\left[E(1)\right]\to 1~~as~~N\to \infty,
  \end{equation}
  from the theory of maximum likelihood estimation~\cite{shao2003mathematical}.

  Here, we assume that $P[E(M-1)]\to 1$ for $M\geq 2$.
  For simplicity, we identify the classical Fisher information $\mathcal{I}_{\rm c}(\alpha;\theta)$ with the objective function in Eq.~(\ref{eq:alpha_optimization}) in this proof.
  Fixing $\alpha$ in the objective function $\mathcal{I}_{\rm c}(\alpha;\theta)$, we can confirm that $\mathcal{I}_{\rm c}(\alpha;\hat{\theta}_{M-1})$ converges smoothly to $\mathcal{I}_{\rm c}(\alpha;{\theta}^*)$ when $\hat{\theta}_{M-1}\to \theta^*$.
  If there is a unique point maximizing $\mathcal{I}_{\rm c}(\alpha;{\theta}^*)$ in the range of $D_{M-1}$, then there exists some small number $\epsilon'>0$, and the following implication holds
  \begin{equation}\label{apdx:convergence_alpha}
    \Bigr\{\left|\hat{\theta}_{M-1}-\theta^*\right|<\epsilon'\Bigr\}\subseteq\Bigr\{\alpha_M(\bm{X}_{M-1})=\alpha_M^*\Bigr\},
  \end{equation}
  where $\{\cdot\}$ denotes a set of events such as
  \begin{align}
      \left\{\bm{x}_{M-1}:\left|\hat{\theta}_{M-1}(\bm{x}_{M-1})-\theta^*\right|<\epsilon'\right\}.
  \end{align}
  Note that even if the maximum point of $\mathcal{I}_{\rm c}(\alpha;{\theta}^*)$ in $D_{M-1}$ is not unique, Eq.~(\ref{apdx:convergence_alpha}) also holds when we adopt the minimum value in a subset $\bar{D}_{M-1}$ of ${D}_{M-1}$ as the realization of $\alpha_M(\bm{X}_{M-1})$.
  Here, the subset $\bar{D}_{M-1}$ consists of elements in ${D}_{M-1}$ satisfying
  \begin{align}
    \left|\mathcal{I}_{\rm c}\left(\alpha;\hat{\theta}_{M-1}\right)-\mathcal{I}_{{\rm c},M-1}^{\rm (max)}\left(\hat{\theta}_{M-1}\right)\right|<\epsilon_{\rm threshold},
  \end{align}
  where $\epsilon_{\rm threshold}$ is a small positive number, and
  \begin{equation}
      \mathcal{I}_{{\rm c},M-1}^{\rm (max)}\left(\hat{\theta}_{M-1}\right):=\max_{\alpha\in D_{M-1}}\mathcal{I}_{\rm c}\left(\alpha;\hat{\theta}_{M-1}\right).
  \end{equation}
  Note that $\alpha_M^*$ is defined as the minimum value of $\bar{D}_{M-1}$ for $\theta^*$.
  In the limit of $\hat{\theta}_{M-1}\to \theta^*$, $\bar{D}_{M-1}$ for $\hat{\theta}_{M-1}$ has the same elements in $\bar{D}_{M-1}$ for $\theta^*$, and therefore the implication in Eq.~(\ref{apdx:convergence_alpha}) holds.
  From the assumption $P[E(M-1)]\to 1$, 
  there exists $N_1^*\in\mathbb{N}$ for all $\eta>0$ such that
  \begin{equation}
    P\left[E(M-1)\cap \left\{\left|\hat{\theta}_{M-1}-\theta^*\right|<\epsilon'\right\}\right]>1-\eta
  \end{equation}
  holds for all $N>N_1^*$.
  The implication in  Eq.~(\ref{apdx:convergence_alpha}) yields the following inequality
  \begin{align}\label{apdx:inequality_1}
    1-\eta&<P\left[E(M-1)\cap \left\{\left|\hat{\theta}_{M-1}-\theta^*\right|<\epsilon'\right\}\right]\notag\\
    &\leq P\left[E(M-1)\cap \Bigr\{\alpha_M(\bm{X}_{M-1})=\alpha_M^*\Bigr\}\right]\notag\\
    &\leq P\left[\Bigr\{\alpha_k(\bm{X}_{k-1})=\alpha_k^*~\forall k\Bigr\}\right].
  \end{align}
  Note that if $\alpha_k(\bm{x}_{k-1})=\alpha_k^*$ holds for all $k=1,2,\cdots,M$, then the likelihood function $\mathcal{L}_M\left(\theta;\bm{x}_M\right)$ can be written as
  \begin{align}
    \mathcal{L}_M\left(\theta;\bm{x}_M\right)&=\prod_{m=1}^MF_m\left(x^{(m)};\alpha_m\left(\bm{x}_{m-1}\right),\theta\right)\notag\\
    &=\prod_{m=1}^MF_m\left(x^{(m)};\alpha_m^*,\theta\right).
  \end{align}
  On the other hand, Lemma 2 guarantees the existence of some natural number $N_2^*$ such that
  \begin{equation}\label{apdx:inequality_2}
    P\left[\mathbf{A}\right]>1-\eta,
  \end{equation}
  for all $N>N_2^*$, where an event {\bf{A}} is defined as
  \begin{description}
    \item[A] $\prod_{m=1}^MF_m\left(X^{(m)};\alpha_m^*,\theta\right)$ has a unique maximum point $\hat{\theta}_M$ such that $\hat{\theta}_M\to\theta^*$.
  \end{description}
  From Eqs.~(\ref{apdx:inequality_1}) and (\ref{apdx:inequality_2}), if we choose $N>\max\{N_1^*,N_2^*\}$, then
  \begin{widetext}
  \begin{align}
    &P\left[E(M)\right]\notag\\
    &=P\left[
    \Bigr\{\mathcal{L}_{M}\left(\theta;\bm{X}_M\right) {\rm~has~a~unique~maximum~point~} \hat{\theta}_M {\rm~s.t.~} \hat{\theta}_M \to \theta^*
    \Bigr\}\cap
    \Bigr\{\alpha_k(\bm{X}_{k-1})=\alpha_k^*~\forall k=1,2,\cdots,M\Bigr\}
    \right]\notag\\
    &=P\left[\Bigr\{\prod_{m=1}^MF_m\left(X^{(m)};\alpha_m^*,\theta\right){\rm~has~a~unique~maximum~ point~}\hat{\theta}_M{\rm~s.t.~}\hat{\theta}_M\to\theta^*\Bigr\}\cap\Bigr\{\alpha_k(\bm{X}_{k-1})=\alpha_k^*~\forall k\Bigr\}\right]\notag\\
    &=P\left[\mathbf{A}\cap\Bigr\{\alpha_k(\bm{X}_{k-1})=\alpha_k^*~\forall k\Bigr\}\right]>1-2\eta.
  \end{align}
  \end{widetext}
  In conclusion, we have proved that $P\left[E(M)\right]\to 1$ also holds, and the mathematical induction establishes the theorem.
  
\end{proof}

To prove Lemma 2 or Eq.~(\ref{apdx:inequality_2}), we first establish the existence of a solution for the likelihood equation on $F_m$.
\begin{lemma}
A natural number $2\leq m\leq M$ is fixed arbitrarily.
Suppose $P[\alpha_m(\boldsymbol{X}_{m-1})=\alpha_m^*]\to 1$ as the number of measurements $N$ increases, then $F_m(X^{(m)};\alpha_m^*,\theta)$ has a unique maximum point $\hat{\theta}^{(m)}$ in $\Theta_m$ such that $\hat{\theta}^{(m)}\to\theta^*$ with the probability approaching 1.
\end{lemma}

\begin{proof}
  We first show
  \begin{equation}\label{apdx:ineq_for_convexity}
    P\left[F_m\left(X^{(m)};\alpha_m^*,\theta^*\right)>F_m\left(X^{(m)};\alpha_m^*,\theta\right)\right]\to 1
  \end{equation}
  holds for any $\theta\in\Theta_m,~\theta\neq\theta^*$.
  The l.h.s can be transformed as
  \begin{widetext}
  \begin{align}\label{apdx:simplification}
    &P\left[F_m\left(X^{(m)};\alpha_m^*,\theta^*\right)>F_m\left(X^{(m)};\alpha_m^*,\theta\right)\right]\notag\\
    &=\sum_{{x}^{(m)}}\sum_{\bm{x}_{m-1}}\mathcal{L}_{m-1}\left(\bm{x}_{m-1};\theta^*\right)
    F_m\left(x^{(m)};\alpha_m(\bm{x}_{m-1}),\theta^*\right)\mathbf{1}{\left\{F_m\left(X^{(m)};\alpha_m^*,\theta^*\right)>F_m\left(X^{(m)};\alpha_m^*,\theta\right)\right\}}({x}^{(m)})\notag\\
    &=\sum_{\alpha'_m\in D_{m-1}}\sum_{\substack{\bm{x}_{m-1}:\\ \alpha_m(\bm{x}_{m-1})= \alpha'_m}}
    \mathcal{L}_{m-1}\sum_{{x}^{(m)}}F_m\left(x^{(m)};\alpha'_m,\theta^*\right)\mathbf{1}{\left\{F_m\left(X^{(m)};\alpha_m^*,\theta^*\right)>F_m\left(X^{(m)};\alpha_m^*,\theta\right)\right\}}({x}^{(m)})\notag\\
    &=\sum_{\alpha'_m\in D_{m-1}}
    P\left[\alpha_m(\boldsymbol{X}_{m-1})=\alpha'_m\right]
    \sum_{{x}^{(m)}}F_m\left(x^{(m)};\alpha'_m,\theta^*\right)\mathbf{1}{\left\{F_m\left(X^{(m)};\alpha_m^*,\theta^*\right)>F_m\left(X^{(m)};\alpha_m^*,\theta\right)\right\}}({x}^{(m)}),
  \end{align}
  \end{widetext}
  where $\mathbf{1}\{\cdot\}(x)$ denotes an indicator function.
  By assumption, the summation for $\alpha'_m$ of the last line vanishes except for the term with $\alpha'_m=\alpha_m^*$ as $N$ increases.
  Thus, the remaining term can be written as
  \begin{equation}
      P\left[F_m\left(Y^{(m)};\alpha_m^*,\theta^*\right)>F_m\left(Y^{(m)};\alpha_m^*,\theta\right)\right],
  \end{equation}
  where $Y^{(m)}$ denotes a random variable following 
  \begin{equation}
    {\rm Bin}\left(N,\mathbf{P}_{\mathcal{D}}\left(\alpha_m^*;\theta^*\right)\right).
  \end{equation}
  Since the probability $\mathbf{P}_{\mathcal{D}}\left(\alpha_m^*;\theta\right)$ is an injective function with respect to $\theta\in \Theta_m$; and its support is independent to $\theta$, these conditions yield
  \begin{equation}\label{apdx:regularity_of_Y}
    P\left[F_m\left(Y^{(m)};\alpha_m^*,\theta^*\right)>F_m\left(Y^{(m)};\alpha_m^*,\theta\right)\right]\to1,
  \end{equation}
  as $N\to\infty$.
  We can prove this from law of large numbers for independent Bernoulli trials on $Y^{(m)}$ and the non-negativity of the Kullback-Leibler divergence.
  Combining Eqs.~(\ref{apdx:regularity_of_Y}),~(\ref{apdx:simplification}), we complete the proof of Eq.~(\ref{apdx:ineq_for_convexity}).
  
  By definition of $\Theta_m$, we can choose a small positive number $\delta$ such that $(\theta^*-\delta,\theta^*+\delta)\subset \Theta_m$.
  Let $\mathbf{B}_{\pm}$ be the following events
  \begin{description}
    \item[B$_{\pm}$] $F_m\left(X^{(m)};\alpha_m^*,\theta^*\right)>F_m\left(X^{(m)};\alpha_m^*,\theta^*\pm\delta\right)$ holds.
  \end{description}
  Considering $x^{(m)}\in\mathbf{B}_{+}\cap\mathbf{B}_{-}$, the equation
  \begin{equation}\label{apdx:likelihhod_eq}
      \frac{\partial}{\partial\theta}F_m\left(x^{(m)};\alpha_m^*,\theta\right)=0,
  \end{equation}
  has solutions, and we write $\hat{\theta}^{(m)}$ as the solution closest to $\theta^*$ among them. 
  Since the difference between $\hat{\theta}^{(m)}$ and $\theta^*$ is less than $\delta$; and for all $\epsilon>0$ there exists $N^*\in\mathbb{N}$ such that $P\left[\mathbf{B}_{\pm}\right]>1-\epsilon$ for all $N>N^*$ by Eq.~(\ref{apdx:ineq_for_convexity}), we obtain
  \begin{align}
    P\left[\left|\hat{\theta}^{(m)}-\theta^*\right|<\delta\right]&\geq P\left[\mathbf{B}_{+}\cap\mathbf{B}_{-}\right]\notag\\
    &>1-2\epsilon.
  \end{align}
  In addition, the likelihood equation Eq.~(\ref{apdx:likelihhod_eq}) has a unique solution maximizing $F_m(x^{(m)};\alpha_m^*,\theta)$ in $\Theta_m$ if a solution of Eq.~(\ref{apdx:likelihhod_eq}) exists in $\Theta_m$.
  Therefore, $\hat{\theta}^{(m)}$ is a unique maximum point of $F_m\left(x^{(m)};\alpha_m^*,\theta\right)$ in $\Theta_m$.
  
\end{proof}

Applying Lemma 1 to all the functions $\{F_m\}$, we establish the existence of 
a unique maximum point of the product function $\Pi_m F_m$.

\begin{lemma}
  Suppose $P[\alpha_k(\boldsymbol{X}_{k-1})=\alpha_k^*]\to 1$ for all $k=2,\cdots,M$ and $\hat{\theta}^{(1)}\to\theta^*$ (convergence in probability). 
  Then, there exists a unique maximum point $\hat{\theta}_M$ of the product function $\Pi_{m=1}^M F_m(X^{(m)};\alpha_m^*,\theta)$ in $(0,\pi)$ such that $\hat{\theta}_{M}\to\theta^*$ with the probability approaching 1 as $N\to \infty$.
\end{lemma}
\begin{proof}
  From the assumption and Lemma 1, we can choose $N^*\in\mathbb{N}$ for all $\epsilon>0$ 
  such that the probability of occurrence of the following event $\mathbf{C}$ is bigger 
  than $1-\epsilon$ for all $N>N^*$; 
  \begin{description}\centering
    \item[C] For all $m$, $\hat{\theta}^{(m)}$ exists in $\Theta_M$.
  \end{description}
  Here, $\hat{\theta}^{(m)}$ denotes a (unique) maximum point of $F_m$ in $\Theta_M$.
  Actually, $F_m$ is globally maximized at the point $\hat{\theta}^{(m)}$ because of the periodicity of $F_m$:
  \begin{equation}\label{mt_maxFm}
      \underset{\theta\in(0,\pi)}{\rm max}~F_m\left(x^{(m)};\alpha_m^*,\theta\right)=F_m\left(x^{(m)};\alpha_m^*,\hat{\theta}^{(m)}\right).
  \end{equation}
  Let $\hat{\theta}_{\rm max}~(\hat{\theta}_{\rm min})$ be the maximum (minimum) value 
  of a set $\{\hat{\theta}^{(m)}\}$. 
  Since the product function $\Pi_{m=1}^M F_m$ is a continuous function with respect to $\theta$, it has a maximum point $\hat{\theta}_M$ in $[\hat{\theta}_{\rm min},\hat{\theta}_{\rm max}]$.
  $\hat{\theta}_{\rm max}$ and $\hat{\theta}_{\rm min}$ themselves cannot be a maximum point, because the gradient at these points does not vanish as follows.
  Since $\hat{\theta}_{\rm min}\leq \hat{\theta}^{(m)}$ holds, we obtain
  \begin{align}
      &\left.\frac{\partial}{\partial \theta}\prod_{m=1}^M F_m\left(x^{(m)};\alpha_m^*,\theta\right)\right|_{\hat{\theta}_{\rm min}}\notag\\
      &=\sum_{k=1}^M F_1\cdots F_{k-1}\left.\frac{\partial F_k}{\partial\theta}\right|_{\hat{\theta}_{\rm min}}F_{k+1}\cdots F_{M}\notag\\
      &=\sum_{k:\hat{\theta}^{(k)}\neq \hat{\theta}_{\rm min}} F_1\cdots F_{k-1}\left.\frac{\partial F_k}{\partial\theta}\right|_{\hat{\theta}_{\rm min}}F_{k+1}\cdots F_{M}>0,
  \end{align}
  where $F_m>0$ and $\partial_\theta F_m>0$ for $\theta <\hat{\theta}^{(m)}$ lead to the final inequality.
  In the same way, we can prove that the gradient at $\hat{\theta}_{\rm max}$ is negative.
  Therefore, we obtain
  \begin{widetext}
  \begin{align}
      &
      P\left[{\rm there~exists~a~maximum~point~}\hat{\theta}_M{\rm~of~}\prod_{m=1}^M F_m\left(X^{(m)};\alpha_m^*,\theta\right){\rm~in~}\Theta_M{\rm~such~that~}\hat{\theta}_M\xrightarrow{p}\theta^*\right]
      \notag\\
      &
      \geq P\left[\mathbf{C}\right]>1-\epsilon,
  \end{align}
  \end{widetext}
  where we use the fact that $\hat{\theta}_{\rm min},\hat{\theta}_{\rm max}\to\theta^*$ yields $\hat{\theta}_M\to\theta^*$.
  
  Finally, we show that $\Pi_m F_m$ is globally maximized at $\hat{\theta}_M$.
  For simplicity, we write $F_m(x^{(m)};\alpha_m^*,\theta)$ as $F_m(\theta)$.
  By Taylor's theorem, there is a scalar $c_m$ between $\hat{\theta}^{(m)}$ and $\hat{\theta}_M$ such that
  \begin{align}\label{mt_taylor_exp}
      \ln F_m(\hat{\theta}^{(m)})&-\ln F_m(\hat{\theta}_M)\notag\\
      &=-\frac{\partial^2_{\theta}\ln F_m(c_m)}{2}\left(\hat{\theta}^{(m)}-\hat{\theta}_M\right)^2.
  \end{align}
  Here, the second derivative in Eq.~(\ref{mt_taylor_exp}) can be upper bounded by $N$ and a positive constant $K$ (independent to $N$) as
  \begin{align}\label{mt_taylor_exp_up}
      -\frac{\partial^2_{\theta}\ln F_m(c_m)}{2}\left(\hat{\theta}^{(m)}-\hat{\theta}_M\right)^2<\frac{NK}{2}\left(\hat{\theta}^{(m)}-\hat{\theta}_M\right)^2.
  \end{align}
  Recalling that $F_1(\theta)$ has a single peak in $(0,\pi)$, for an arbitrary $\lambda\notin\Theta_M$, we obtain
  \begin{align}\label{mt_meanvalthm_low}
      &\frac{\ln {F_1(\hat{\theta}_M)}-\ln {F_1(\lambda)}}{\left|\hat{\theta}_M-\lambda\right|}\notag\\
      &=\frac{N}{2}\left|\frac{p_{\rm q}\sin{\left(c_1\right)}({x^{(1)}}/{N}-\mathbf{P}_{\mathcal{D}}(1;c_1))}{\mathbf{P}_{\mathcal{D}}(1;c_1)\left[1-\mathbf{P}_{\mathcal{D}}(1;c_1)\right]}\right|>\frac{NL}{2},
  \end{align}
  where $c_1$ is a scalar between $\lambda$ and $\hat{\theta}_M$, and $L>0$ is a positive constant.
  Combining Eqs.~(\ref{mt_maxFm}) and (\ref{mt_taylor_exp})--(\ref{mt_meanvalthm_low}), the following inequalities hold:
  \begin{align}
      \sum_{m=1}^M &\ln \frac{F_m(\hat{\theta}_M)}{F_m(\lambda)}\geq \ln \frac{F_1(\hat{\theta}_M)}{F_1(\lambda)}+\sum_{m=2}^M \ln \frac{F_m(\hat{\theta}_M)}{F_m(\hat{\theta}^{(m)})}\notag\\
      &>\frac{N}{2}\left[L\left|\lambda-\hat{\theta}_M\right|-{K} \sum_{m=2}^{M}\left(\hat{\theta}^{(m)}-\hat{\theta}_M\right)^2\right].
  \end{align}
  For sufficiently large $N$, the second term in the square bracket becomes small compared to the first term, and therefore we conclude that 
  \begin{equation}
      {\prod_m F_m(\hat{\theta}_M)} > \prod_m F_m(\lambda).
  \end{equation}
  Consequently, the maximum of $\Pi_{m=1}^M F_m$ in $\Theta_M$ becomes the global maximum in $(0,\pi)$.
  
\end{proof}


Finally, we prove that, when $N$ is large, the asymptotic variance of our estimator 
can achieve the classical Cram\'{e}r-Rao lower bound. 
\begin{thm}
  If $\hat{\theta}_M$ is a maximum likelihood estimator on $\mathcal{L}_{M}(\theta;\bm{X}_M)$ such that $\hat{\theta}_M\to \theta^*$ (convergence in probability), then the following convergence holds;
  \begin{equation}
      \sqrt{\mathcal{I}^{*}_{\rm c,tot}(\theta^*)}\left(\cos\hat{\theta}_M-\cos{\theta^*}\right)\to\mathcal{N}(0,\sin^2\theta^*),
  \end{equation}
  where $\mathcal{N}$ denotes a centered normal distribution with variance 
  $\sin^2{\theta^*}$ and $\to$ means the convergence in distribution as $N$ increases.
\end{thm}

\begin{proof}
  Using the Taylor's theorem, we can expand $\partial_\theta \ln\mathcal{L}_M(\theta;\bm{X}_M)$ around $\theta=\theta^*$ as
  \begin{align}\label{apdx:L_taylor}
    \partial_\theta\ln{\mathcal{L}_M(\theta;\bm{X}_M)}&=\partial_\theta\ln{\mathcal{L}_M(\theta^*;\bm{X}_M)}\notag\\
    &~~+\partial^2_\theta\ln{\mathcal{L}_M(\theta^*;\bm{X}_M)}(\theta-\theta^*)\notag\\
    &~~+\frac{1}{2}\partial^3_\theta\ln{\mathcal{L}_M(\theta';\bm{X}_M)}(\theta-\theta^*)^2,
\end{align}
where $\theta'$ denotes some real number between $\theta$ and $\theta^*$.
Since $\hat{\theta}_M$ is a root of the likelihood equation, we obtain
\begin{align}\label{apdx:Asymptotic effectiveness}
    &\sqrt{\mathcal{I}^*_{\rm c,tot}(\theta^*)}(\hat{\theta}_M-\theta^*)= \frac{A}{B-C},
\end{align}
where for simplicity we introduced the random variables $A,B,$ and $C$ defined as
\begin{align}
    A:=\left(\mathcal{I}^*_{\rm c,tot}(\theta^*)\right)^{-1/2}\partial_\theta\ln{\mathcal{L}_M(\theta^*;\bm{X}_M)},\notag
\end{align}
\begin{align}
    B:=-\left({\mathcal{I}^*_{\rm c,tot}(\theta^*)}\right)^{-1}\partial^2_\theta\ln{\mathcal{L}_M(\theta^*;\bm{X}_M)},\notag
\end{align}
\begin{align}
    C:=\left({2\mathcal{I}^*_{\rm c,tot}(\theta^*)}\right)^{-1}\partial^3_\theta\ln{\mathcal{L}_M(\theta';\bm{X}_M)}(\hat{\theta}_M-\theta^*)\notag.
\end{align}

Here, we show that $A$ converges to the standard normal distribution.
To apply Lemma 3, we first calculate the characteristic function of $A$ replacing $\bm{X}_M,{\alpha}_m\to\bm{Y}_M,{\alpha}^*_m$, where $\bm{Y}_M=(Y^{(1)},\cdots,Y^{(M)})$ is a random vector whose element follows Binomial distribution ${\rm Bin}(N,\mathbf{P}_{\mathcal{D}}\left(\alpha_m^*;\theta^*\right))$ independently. 
That is, for $t\in\mathbb{R}$ we have
\begin{align}
    &\mathbf{E}\left[{\rm exp}(it\left.A\right|_{\bm{X}_M,{\alpha}_m\to\bm{Y}_M,{\alpha}^*_m})\right] \notag\\
    &=\mathbf{E}\left[{\rm exp}\left(\frac{it\sum_{m=1}^M\partial_\theta\ln{ F_m({Y}^{(m)};{\alpha}^*_m,\theta^*)}}{\sqrt{\mathcal{I}^*_{\rm c,tot}(\theta^*)}}\right)\right] \notag\\
    &=\prod_{m=1}^M\left(\mathbf{E}\left[ {\rm exp}\left(\frac{it\partial_\theta\ln{ f_m({Y}_1^{(m)};{\alpha}^*_m,\theta^*)}}{\sqrt{\mathcal{I}^*_{\rm c,tot}(\theta^*)}}\right)\right]\right)^N\notag\\
    &=\prod_{m=1}^M\left(1-\frac{t^2}{2N}\frac{\mathcal{I}_{\rm c}(\alpha^*_m;\theta^*)
    }{\sum_{m=1}^M\mathcal{I}_{\rm c}(\alpha_m^*;\theta^*)}+{o}\left(\frac{1}{N}\right)\right)^N\notag\\
    &\to{\rm exp}\left(-\frac{t^2}{2}\right),
\end{align}
where $Y_1^{(m)}$ denotes a random variable following Bernoulli distribution with the success probability $\mathbf{P}_{\mathcal{D}}(\alpha_m^*;\theta^*)$, and $f_m$ is the probability function of $Y_1^{(m)}$.
Thus, Lemma 3 yields
\begin{align}
    \lim_{N\to\infty} \mathbf{E}\left[{\rm exp}(itA)\right]={\rm exp}\left(-\frac{t^2}{2}\right).
\end{align}
Therefore, $A\to\mathcal{N}(0,1)$ holds.

Next, we consider the convergence of $B$. 
According to the Chebyshev inequality, for all $c>0$ we obtain
\begin{align}
    P\left[\left|B-1\right|>c\right]&\leq \frac{\mathbf{E}\left[\left(B-1\right)^2\right]}{c^2}.
\end{align}
Here, $(B-1)^2$ is upper bounded, and the expectation of $(B-1)^2$ replacing $\bm{X}_M,{\alpha}_m\to\bm{Y}_M,{\alpha}^*_m$ also converges to 0 as follows.
For simplicity, we define $W_m(Y^{(m)})$ and $w_m(Y^{(m)})$ as
\begin{align}
    W_m(Y^{(m)})&:=\partial^2_\theta\ln{{F}_m({Y}^{(m)};{\alpha}^*_m,\theta^*)}\notag\\
    &~~~~~-\mathbf{E}\left[\partial^2_\theta\ln{{F}_m({Y}^{(m)};{\alpha}^*_m,\theta^*)}\right],\notag
\end{align}
\begin{align}
    w_{m}(Y_1^{(m)})&:=\partial^2_\theta\ln{{f}_m({Y}_1^{(m)};\alpha^*_m,\theta^*)}\notag\\
    &~~~~~-\mathbf{E}\left[\partial^2_\theta\ln{{f}_m({Y}_1^{(m)};\alpha^*_m,\theta^*)}\right],\notag
\end{align}
respectively.
\begin{align}
    &\mathbf{E}\left[\left.\left(B-1\right)^2\right|_{\bm{X}_M,{\alpha}_m\to\bm{Y}_M,{\alpha}^*_m}\right]\notag\\
    &=\frac{{\rm Var}\left[\partial^2_\theta\ln{\prod_{m=1}^M{F}_m({Y}^{(m)};{\alpha}^*_m,\theta^*)}\right]}{\left[\mathcal{I}^*_{\rm c,tot}(\theta^*)\right]^2}\notag\\
    &=\frac{\sum_{n,m=1}^M\mathbf{E}\left[W_m(Y^{(m)})W_n(Y^{(n)})\right]}{\left[\mathcal{I}^*_{\rm c,tot}(\theta^*)\right]^2} \notag\\
    &=\frac{N}{\left[\mathcal{I}^*_{\rm c,tot}(\theta^*)\right]^2}\sum_{m=1}^M\mathbf{E}\left[w^2_{m}(Y_1^{(m)})\right] \notag\\
    &=\frac{1}{N}\frac{\sum_{m=1}^M\mathbf{E}\left[w^2_{m}(Y_1^{(m)})\right]}{\left[\sum_{m=1}^M\mathcal{I}_{\rm c}(\alpha^*_m;\theta^*)\right]^2}\to0,
\end{align}
where we used $\mathcal{I}^*_{\rm c,tot}(\theta^*)=-\mathbf{E}[\partial^2_{\theta}\ln \prod_{m=1}^M F_m]$ for the first equality.
Thus, Lemma 3 establishes that $B\to1$ holds.

Since $N^{-1}\partial_{\theta}^3\ln\mathcal{L}_M(\theta';\bm{x}_M)$ is also upper bounded, $\hat{\theta}_M-\theta^*\to 0$ yields $C\to0$.
To sum up the above asymptotic properties, we conclude that the following convergence 
holds;
\begin{equation}
    \sqrt{\mathcal{I}_{\rm c,tot}^*(\theta^*)}(\hat{\theta}_M-{\theta}^*)\to\mathcal{N}(0,1).
\end{equation}
In addition, the delta method~\cite{shao2003mathematical} yields a more suitable 
expression 
\begin{equation}
    \sqrt{\mathcal{I}_{\rm c,tot}^*(\theta^*)}(\cos{\hat{\theta}_M}-\cos{\theta^*})\to\mathcal{N}(0,\sin^2{\theta^*}).
\end{equation}
\end{proof}

To complete the proof of Theorem~\ref{thm:asym_normal}, we finally show the following lemma.

\begin{lemma}
  Let $h(\bm{x}_M, \alpha_2,\cdots,\alpha_M)$ be an upper bounded function on $\bm{x}_M\in\{0,1,\cdots,N\}^{M}$ and $Y^{(m)}~(m=1,2,\cdots,M)$ be an independent random variable following Binomial distribution ${\rm Bin}(N,\mathbf{P}_{\mathcal{D}}\left(\alpha_m^*;\theta^*\right))$.
  Suppose that $\mathbf{E}\left[h(\bm{Y}_M,\alpha^*_2,\cdots,\alpha^*_M)\right]$ converges to some constant $\beta$ as $N\to \infty$, 
  then the following convergence holds
  \begin{align}\label{apdx:limit_lemma3}
      \lim_{N\to \infty}&\mathbf{E}\left[h(\bm{X}_M,\alpha_2(\bm{X}_1),\cdots,\alpha_M(\bm{X}_{M-1}))\right]=\beta.
  \end{align}
\end{lemma}

\begin{proof}
  \begin{widetext}
  Expanding the left hand side in Eq.~(\ref{apdx:limit_lemma3}), we obtain
  \begin{align}
      &\mathbf{E}\left[h(\bm{X}_M,\alpha_2(\bm{X}_1),\cdots,\alpha_M(\bm{X}_{M-1}))\right]=\sum_{\bm{x}_M}\mathcal{L}_M(\bm{x}_M;\theta^*)h(\bm{x}_M,\alpha_2(\bm{x}_1),\cdots,\alpha_M(\bm{x}_{M-1}))\notag\\
      &=\sum_{\substack{\bm{x}_M\\{\alpha}_k(\bm{x}_{k-1})=\alpha_k^*~\forall k}}\mathcal{L}_M(\bm{x}_M;\theta^*)h(\bm{x}_M,\alpha_2(\bm{x}_1),\cdots,\alpha_M(\bm{x}_{M-1}))\notag\\
      &~~~~~~~~~~~~~~~~~~~~~~~~+\sum_{\substack{\bm{x}_M\\{\alpha}_k(\bm{x}_{k-1})\neq\alpha_k^*~\exists k}}\mathcal{L}_M(\bm{x}_M;\theta^*)h(\bm{x}_M,\alpha_2(\bm{x}_1),\cdots,\alpha_M(\bm{x}_{M-1}))\notag\\
      &=\sum_{\substack{\bm{x}_M\\{\alpha}_k(\bm{x}_{k-1})=\alpha_k^*~\forall k}}\prod_{m=1}^M F_m\left(x^{(m)};\alpha_m^*,\theta^*\right)h(\bm{x}_M,\alpha^*_2,\cdots,\alpha^*_M)\notag\\
      &~~~~~~~~~~~~~~~~~~~~~~~~+\sum_{\substack{\bm{x}_M\\{\alpha}_k(\bm{x}_{k-1})\neq\alpha_k^*~\exists k}}\mathcal{L}_M(\bm{x}_M;\theta^*)h(\bm{x}_M,\alpha_2(\bm{x}_1),\cdots,\alpha_M(\bm{x}_{M-1}))\notag\\
      &=\mathbf{E}\left[h(\bm{Y}_M,\alpha^*_2,\cdots,\alpha^*_M)\right]\notag\\
      &~~~~~~~~~~~~~~~~~~~~~~~~+\sum_{\substack{\bm{x}_M\\{\alpha}_k(\bm{x}_{k-1})\neq\alpha_k^*~\exists k}}\mathcal{L}_M(\bm{x}_M;\theta^*)h(\bm{x}_M,\alpha_2(\bm{x}_1),\cdots,\alpha_M(\bm{x}_{M-1}))\notag\\
      &~~~~~~~~~~~~~~~~~~~~~~~~-\sum_{\substack{\bm{x}_M\\{\alpha}_k(\bm{x}_{k-1})\neq\alpha_k^*~\exists k}}\prod_{m=1}^M F_m\left(x^{(m)};\alpha_m^*,\theta^*\right)h(\bm{x}_M,\alpha^*_2,\cdots,\alpha^*_M).
  \end{align}
  The above expression yields the following inequality
  \begin{align}\label{apdx:beta_convergence}
      &\left|\mathbf{E}\left[h(\bm{X}_M,\alpha_2(\bm{X}_1),\cdots,\alpha_M(\bm{X}_{M-1}))\right]-\beta\right|\notag\\
      &\leq \left|\mathbf{E}\left[h(\bm{Y}_M,\alpha^*_2,\cdots,\alpha^*_M)\right]-\beta\right|\notag\\
      &~~~~~~~~~~~~~~~~~~~~~~~~+\sum_{\substack{\bm{x}_M\\{\alpha}_k(\bm{x}_{k-1})\neq\alpha_k^*~\exists k}}\mathcal{L}_M(\bm{x}_M;\theta^*)\left|h(\bm{x}_M,\alpha_2(\bm{x}_1),\cdots,\alpha_M(\bm{x}_{M-1}))\right|\notag\\
      &~~~~~~~~~~~~~~~~~~~~~~~~+\sum_{\substack{\bm{x}_M\\{\alpha}_k(\bm{x}_{k-1})\neq\alpha_k^*~\exists k}}\prod_{m=1}^M F_m\left(x^{(m)};\alpha_m^*,\theta^*\right)\left|h(\bm{x}_M,\alpha^*_2,\cdots,\alpha^*_M)\right|\notag\\
      &\leq \left|\mathbf{E}\left[h(\bm{Y}_M,\alpha^*_2,\cdots,\alpha^*_M)\right]-\beta\right| +\gamma\sum_{\substack{\bm{x}_M\\{\alpha}_k(\bm{x}_{k-1})\neq\alpha_k^*~\exists k}}\left(\mathcal{L}_M(\bm{x}_M;\theta^*)+\prod_{m=1}^M F_m\left(x^{(m)};\alpha_m^*,\theta^*\right)\right),
  \end{align}
  where $\gamma$ is the upper bound of $h$.
  The convergence of amplified levels implies that for all $\epsilon>0$ there exists $N^*\in\mathbb{N}$ such that
  \begin{align}\label{apdx:gamma_convergence}
    P\left[\alpha_k(\bm{X}_{k-1})=\alpha^*_k~\forall k\right]&=\sum_{\substack{\bm{x}_M\\{\alpha}_k(\bm{x}_{k-1})=\alpha_k^*~\forall k}}\mathcal{L}_{M}\left(\bm{x}_M;\theta^*\right)=\sum_{\substack{\bm{x}_M\\{\alpha}_k(\bm{x}_{k-1})=\alpha_k^*~\forall k}}\prod_{m=1}^M F_m(x^{(m)};\alpha_m^*,\theta^*)>1-\frac{\epsilon}{2\gamma},
  \end{align}
  \end{widetext}
  for all $N>N^*$. Combining Eqs.~(\ref{apdx:beta_convergence}),~(\ref{apdx:gamma_convergence}) and the assumption for the convergence of $\mathbf{E}\left[h(\bm{Y}_M,\alpha^*_2,\cdots,\alpha^*_M)\right]$, we obtain Eq.~(\ref{apdx:limit_lemma3}).
  
\end{proof}

\renewcommand{\theequation}{F.\arabic{equation}}
\setcounter{equation}{0}
\section{Additional experiment}\label{apdx:additional_exp}
\subsection{Impact of calibration errors}

Here, we provide the numerical simulations to clarify the impact of calibration errors of $p_{\rm q}$ on the performance of our method.
If an estimate of $p_q$ contains some calibration error, the resulting estimator can be biased even in the asymptotic case.
The experimental settings are the same as those in Section~\ref{sec:confirm_asym_prop}.
Figure~\ref{apdx_fig:cal_error} shows that in the current setting, the $\pm0.1\%$ calibration errors can be ignored up to $\sim 10^5$ queries.
In the case of the $\pm0.5\%$ calibration errors, the estimation performance deteriorates even for small queries $\sim 6\times 10^3$.
The results indicate that our method requires a relatively high-precision estimate of $p_{\rm q}$ that have the precision comparable to the target precision of $\cos{\theta^*}$.
Note that, in the quantum amplitude estimation problem, the previous paper~\cite{tanaka2021amplitude} provides the (non-adaptive) multiparameter estimation method with respect to both a target parameter and a noise parameter.
In addition, they also show the robustness of target parameter estimation against the estimation error of the noise parameter in terms of the Fisher information matrix.
Although a similar extension of our adaptive method is possible, it falls outside the scope of this paper.

\begin{figure*}[htb]
 \centering
 \begin{tabular}{ccc}
 \includegraphics[scale=0.8]{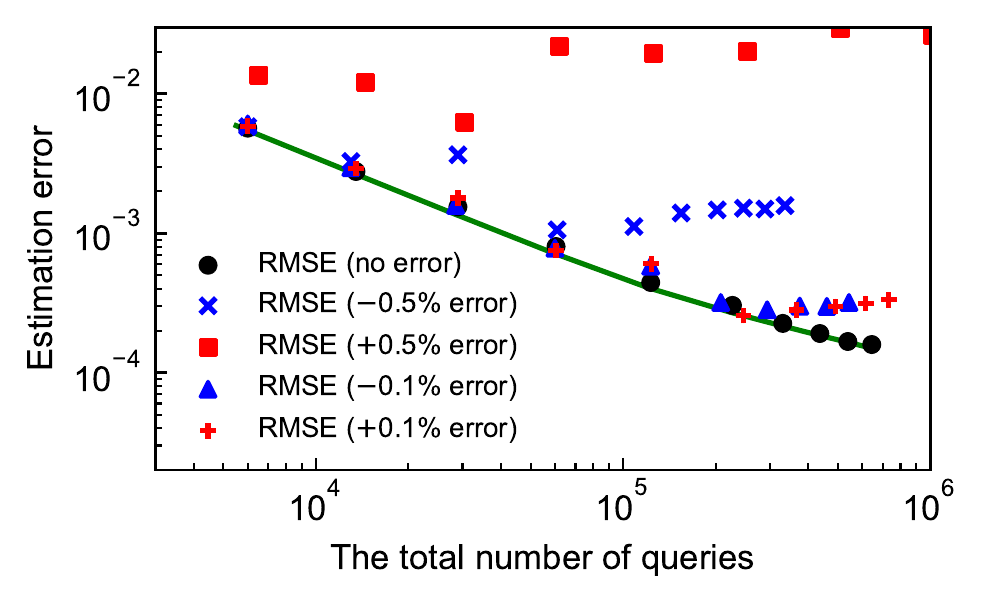}&&\includegraphics[scale=0.8]{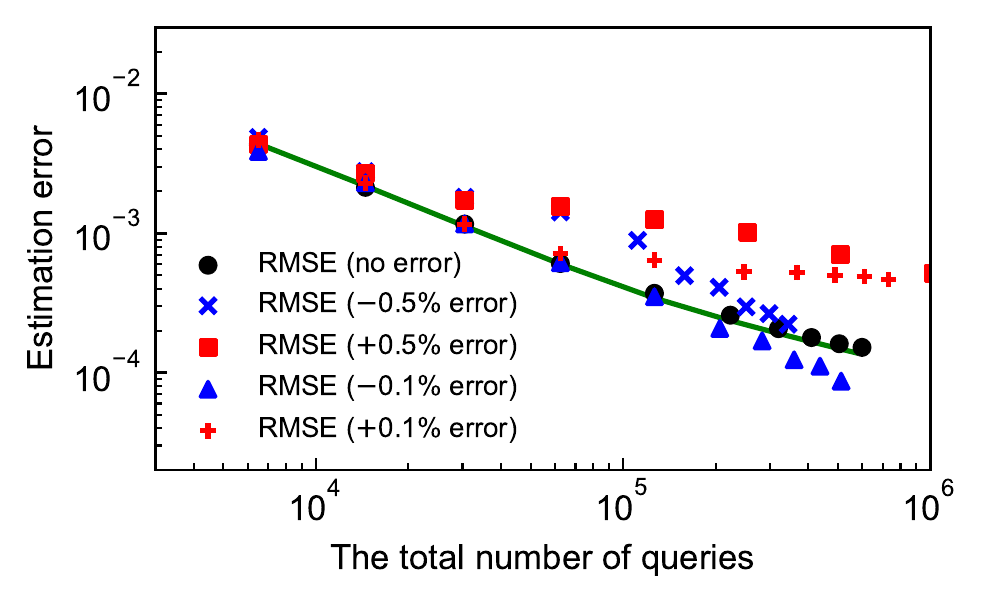}\\
 (a) $\cos\theta^*=0.042$&&(b) $\cos\theta^*=0.5$\\
 &&\\
 \end{tabular}
 \caption{
 The impact of calibration error of $p_{\rm q}$ on the estimation performance. 
 The plots indicate the root mean squared error (RMSE) of the method, evaluated with 200 trials, using the estimated noise parameter $p_{\rm q}$ including $\pm 0.1\%$ or $\pm 0.5\%$ calibration errors.
 The true value of $p_{\rm q}$ is taken as $0.995$.
 The green solid line represents the (asymptotic) QCR bound under the perfect calibration. 
 We note that the RMSE is plotted as a function of the maximum value of $N_{\rm q}$ in the trials because of the randomness of $N_{\rm q}$.
 }
 \label{apdx_fig:cal_error}
\end{figure*}

\subsection{Asymptotic properties in several noise parameters}

We provide additional results to confirm the asymptotic property of our estimator in several noise parameters.
Here, all settings except for the noise parameter are the same as Section~\ref{sec:confirm_asym_prop}, and the noise parameter is set to $p_{\rm q}=0.999$ (i.e., $0.1\%$ depolarization noise) or $p_{\rm q}=0.99$ (i.e., $1\%$ depolarization noise).
In Fig.~\ref{apdx_fig:rmse_vs_crbound}, we can confirm that the RMSE achieves the CCR bound (also QCR bound) sufficiently in $N=500$ as in Section~\ref{sec:confirm_asym_prop}.

\begin{figure*} 
 \centering
 \begin{tabular}{ccc}
     \includegraphics[scale=0.8]{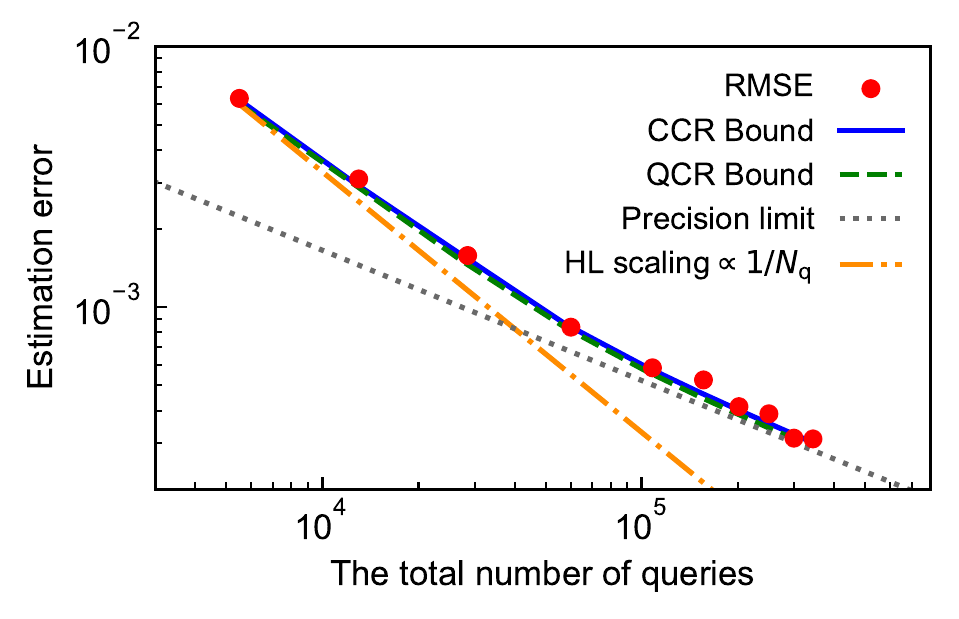}&&\includegraphics[scale=0.8]{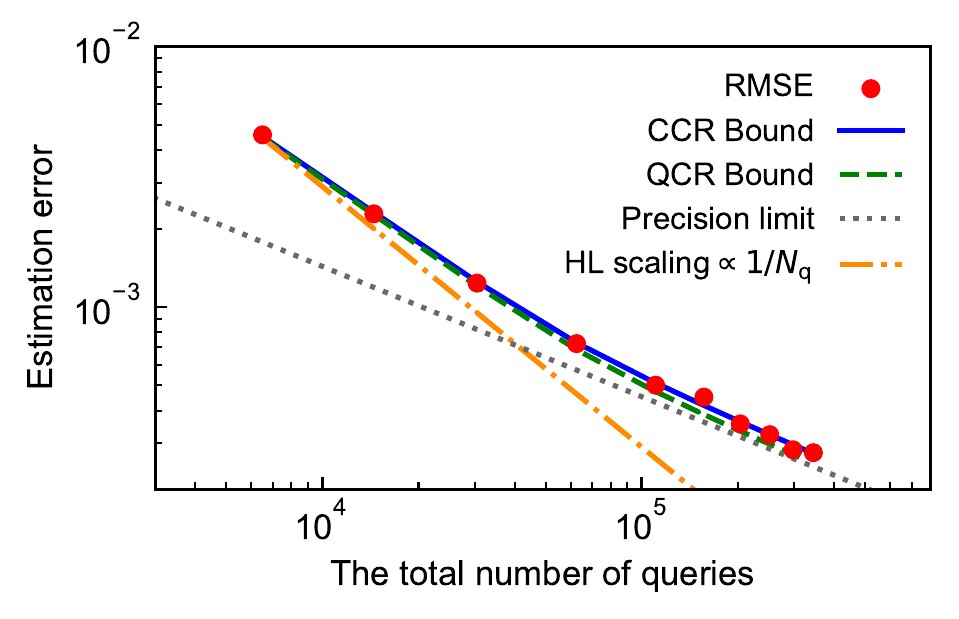}\\
     (a) $\cos\theta^*=0.042$ and $p_{\rm q}=0.99$&&(b) $\cos\theta^*=0.5$ and $p_{\rm q}=0.99$\\
     &&\\
     \includegraphics[scale=0.8]{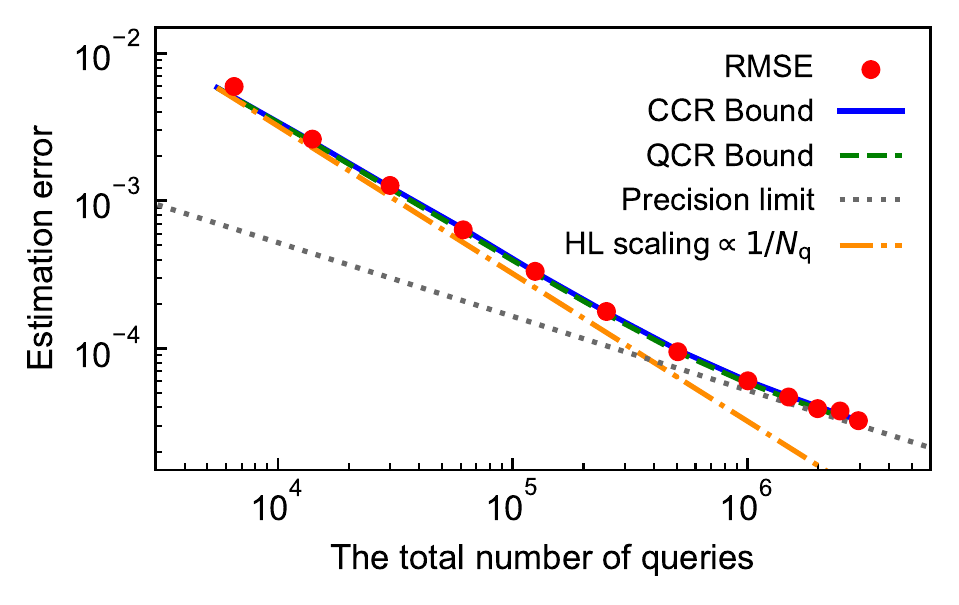}&&\includegraphics[scale=0.8]{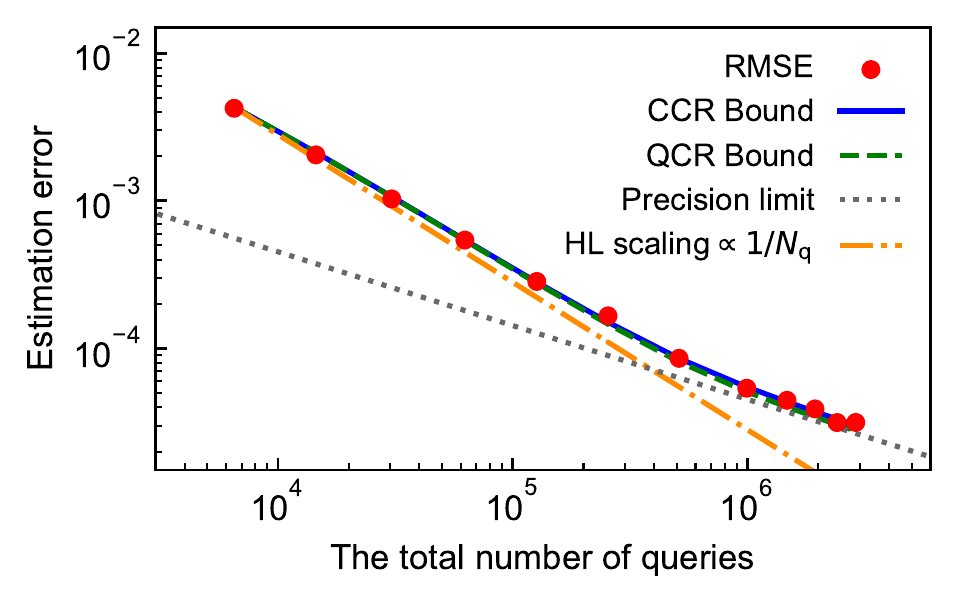}\\
     (c) $\cos\theta^*=0.042$ and $p_{\rm q}=0.999$&&(d) $\cos\theta^*=0.5$ and $p_{\rm q}=0.999$\\
     &&
 \end{tabular}
 \caption{The estimation error of $\braket{\mathcal{O}}=\cos{\theta^*}$ and the total number of queries under several depolarization noise parameters.
 The root mean squared error (RMSE) of the estimator $\cos{\hat{\theta}}$ is evaluated by 300 trials, and the red dots correspond to $M=3,4,\cdots,12$ (or 14) from left to right.
 Since $N_{\rm q}$ in Eq.~(\ref{total q}) is a random variable in our method, the RMSE is plotted as a function of the maximum value of $N_{\rm q}$ in 300 trials.
 The blue solid and green dashed lines represent the asymptotic values of CCR/QCR bounds obtained from the corresponding classical/quantum Fisher information $\mathcal{I}^*_{\rm c/q, tot}(\theta^*)$, respectively.
 The orange dash-dotted and the gray dotted lines represent the Heisenberg-limited scaling ${\rm RMSE}=O(1/N_{\rm q})$ and the precision limit derived in Theorem~\ref{thm:QFI_limit}, respectively.}
 \label{apdx_fig:rmse_vs_crbound}
\end{figure*}

\end{document}